\documentclass{article}

\PassOptionsToPackage{numbers,sort,compress}{natbib}
\usepackage[final]{neurips_2022}

\usepackage{amsmath}
\usepackage{times}
\usepackage{graphicx}
\usepackage{color}
\usepackage{multirow}
\usepackage{optidef}
\usepackage{wrapfig}

\usepackage{tikz}
\usepackage{rotating}
\usepackage{bbm}
\usepackage{latexsym}
\usetikzlibrary{external}

\usepackage{tabularx}
\usepackage{amsfonts,amsthm,bm}
\usepackage{appendix}
\usepackage{mathtools}
\usepackage{algorithm,algorithmic}
\usepackage{subfig}
\usepackage{array}
\usepackage{wrapfig}
\usepackage{enumitem}
\usepackage{colortbl}
\setitemize{itemsep=2pt,topsep=0pt,parsep=0pt,partopsep=0pt,leftmargin=*}

\newcommand{\vR}{\mbox{$\mathbf {R}$}}

\newcommand{\vu}{\mathbf{u}}
\newcommand{\vv}{\mathbf{v}}

\newcommand{\vh}{\mathbf{h}}
\newcommand{\vy}{\mathbf{y}}

\newcommand{\vx}{\mathbf{x}}
\newcommand{\vwx}{\mathbf{x}}
\newcommand{\vwy}{{\mathbf{y}}}

\newcommand{\vs}{\mathbf{s}}
\newcommand{\bX}{\mathbf{X}}
\newcommand{\bwX}{\bm{\mathcal{X}}}
\newcommand{\bwY}{\bm{\mathcal{Y}}}
\newcommand{\bwH}{\bm{\mathcal{H}}}
\newcommand{\bY}{\mathbf{Y}}

\newcommand{\vW}{\mathbf{W}}
\newcommand{\vU}{\mathbf{U}}
\newcommand{\vV}{\mathbf{V}}

\newcommand{\vM}{\mathbf{M}}
\newcommand{\vY}{\mathbf{Y}}

\newcommand{\vX}{\mathbf{X}}
\newcommand{\vA}{\mathbf{A}}

\newcommand{\vS}{\mathbf{S}}
\newcommand{\vH}{\mathbf{H}}
\newcommand{\vG}{\mathbf{G}}
\newcommand{\vP}{\mathbf{P}}
\newcommand{\vD}{\mathbf{D}}
\newcommand{\vQ}{\mathbf{Q}}
\newcommand{\vI}{\mathbf{I}}

\newcommand{\bGam}{\bm{\Gamma}}
\newcommand{\bSig}{\bm{\Sigma}}

\newcommand{\bTheta}{\bm{\Theta}}

\newcommand{\Pcal}{\mathcal{P}}
\newcommand{\bLambd}{\bm{\Lambda}}
\newtheorem{theorem}{Theorem}
\newtheorem{lemma}{Lemma}

\newif\iftikzinclude
\tikzincludetrue
\newif\ifpalgoinclude
\newif\ifantisparseexinclude
\newif\ifsparseexinclude
\newif\ifappinclude

\usepackage[utf8]{inputenc} 
\usepackage[T1]{fontenc}    
\usepackage{hyperref}       
\usepackage{url}            
\usepackage{booktabs}       
\usepackage{amsfonts}       
\usepackage{nicefrac}       
\usepackage{microtype}      
\usepackage{xcolor}         

\title{Biologically-Plausible Determinant Maximization Neural Networks for Blind Separation of Correlated Sources}

\makeatletter
\newcommand\extralabel[2]{{\edef\@currentlabel{\@currentlabel#2}\label{#1}}}
\makeatother

\makeatletter
\providecommand\phantomcaption{\caption@refstepcounter\@captype}
\makeatother

\author{Bariscan Bozkurt\textsuperscript{1,2} \quad Cengiz Pehlevan\textsuperscript{3} \quad Alper T. Erdogan\textsuperscript{1,2} \\
\textsuperscript{1}KUIS AI Center, Koc University, Turkey \quad \textsuperscript{2}EEE Department, Koc University, Turkey \\
\textsuperscript{3}John A. Paulson School of Engineering \& Applied Sciences and Center for\\   Brain Science,
Harvard University, Cambridge, 02138 MA, USA\\
\texttt{\{bbozkurt15, alperdogan\}@ku.edu.tr}\quad 
\texttt{cpehlevan@seas.harvard.edu}
}

\begin{document}

\maketitle

\begin{abstract}

Extraction of latent sources of complex stimuli is critical for making sense of the world. While the brain solves this blind source separation (BSS) problem continuously, its algorithms remain unknown. Previous work on biologically-plausible BSS algorithms assumed that observed signals are linear mixtures of statistically independent or uncorrelated sources, limiting the domain of applicability of these algorithms. To overcome this limitation, we propose novel biologically-plausible neural networks for the blind separation of potentially dependent/correlated sources. Differing from previous work, we assume some general geometric, not statistical, conditions on the source vectors allowing separation of potentially dependent/correlated sources. Concretely, we assume that the source vectors are sufficiently scattered in their domains which can be described by certain polytopes. Then, we consider recovery of these sources by the Det-Max criterion, which maximizes the determinant of the output correlation matrix to enforce a similar spread for the source estimates. Starting from this normative principle, and using a weighted similarity matching approach that enables arbitrary linear transformations adaptable by local learning rules, we derive two-layer biologically-plausible neural network algorithms that can separate mixtures into sources coming from a variety of source domains. We demonstrate that our algorithms outperform other biologically-plausible BSS algorithms on correlated source separation problems. \looseness = -1
\end{abstract}

\section{Introduction}

Our brains constantly and effortlessly extract latent causes, or sources, of complex visual, auditory or olfactory stimuli sensed by sensory organs \citep{bell1995information,olshausen1996emergence,bronkhorst2000cocktail,lewicki2002efficient,asari2006sparse,narayan2007cortical,bee2008cocktail,mcdermott2009cocktail,mesgarani2012selective,golumbic2013mechanisms,isomura2015cultured}. This extraction is mostly done without any instruction, in an unsupervised manner, making the process an instance of the blind source separation (BSS) problem \citep{comon2010handbook,cichocki2009nonnegative}. Indeed, visual and auditory cortical receptive fields were argued to  be the result of performing BSS on natural images \cite{bell1995information,olshausen1996emergence} and sounds \cite{lewicki2002efficient}. The wide-spread use of BSS in the brain suggests the existence of generic circuit motifs that perform this task \cite{sharma2000induction}. Consequently, the literature on biologically-plausible neural network algorithms for BSS is growing \citep{eagleman2001cerebellar,pehlevan2017blind, isomura2018error, erdogan2020blind,bahroun2021normative}.

 Because BSS is an underdetermined inverse problem, BSS algorithms make generative assumptions on observations. In most instances of the biologically-plausible BSS algorithms, complex stimuli are assumed to be linear mixtures of latent sources. This assumption is particularly fruitful and is used to model, for example, natural images  \cite{olshausen1997sparse,bell1995information}, and responses of olfactory neurons to complex odorants \citep{zhang2016robust,krishnamurthy2017disorder,singh2021odor}. However, linear mixing by itself is not sufficient for source identifiability; further assumptions are  needed. Previous work on biologically-plausible algorithms for BSS of linear mixtures assumed sources to be statistically independent \citep{isomura2018error,bahroun2021normative,lipshutz2022biologically} or uncorrelated \cite{pehlevan2017blind,erdogan2020blind}. However, these assumptions are very limiting when considering real data where sources can themselves be correlated. \looseness = -1 

In this paper, we address the limitation imposed by independence assumptions and provide biologically-plausible BSS neural networks that can separate potentially correlated sources. We achieve this by considering various general geometric identifiability conditions on sources instead of statistical assumptions like independence or uncorrelatedness. In particular, 1) we make natural assumptions on the domains of source vectors--like nonnegativity, sparsity, anti-sparsity or boundedness (Figure \ref{fig:sourcedomains})--and 2) we assume that latent source vectors are sufficiently spread in their domain \cite{lin2015identifiability,tatli2021tspsubmitted}. Because these identifiability conditions are not stochastic in nature, our neural networks are able to separate both independent and dependent sources.

We derive our biologically-plausible algorithms from a normative principle. A common method for exploiting our geometric identifiability conditions is to disperse latent vector estimates across their presumed domain by maximizing the determinant of their sample correlation matrix, i.e., the Det-Max approach \citep{schachtner2011towards,lin2015identifiability,erdogan2013class,inan2014convolutive,babatas2018algorithmic}. Starting from a Det-Max objective function with constraints that specify the domain of source vectors, and using  mathematical tools introduced for mapping optimization algorithms to adaptive Hebbian neural networks \cite{pehlevan2015normative,pehlevan2019neuroscience,erdogan2020blind}, we derive two-layered neural networks that can separate potentially correlated sources from their linear mixtures (Figure \ref{fig:GenericNN}). These networks contain feedforward, recurrent and feedback synaptic connections updated via Hebbian or anti-Hebbian update rules. The domain of latent sources determines the structure of the output layer of the neural network (Figure \ref{fig:GenericNN}, Table \ref{tab:nnoutdynamicsexamples} and Appendix \ref{sec:dersourcedomains}).  

In summary, our main contributions in this article are the following:
\begin{itemize}
  \item We propose a normative framework for generating biologically plausible neural networks that are capable of separating correlated sources from their mixtures by deriving them from a Det-Max objective function subject to source domain constraints.
\item Our framework can handle infinitely many source types by exploiting their source domain topology.
\item We demonstrate the performance of our networks in simulations with synthetic and realistic data.
\end{itemize}




\begin{figure}[t]%
    \centering
    \subfloat[ $\mathcal{B}_{\ell_1}$ \\ (sparse)]{{\includegraphics[width=2cm, trim=7.5cm 12.0cm 6.5cm 12.0cm,clip]{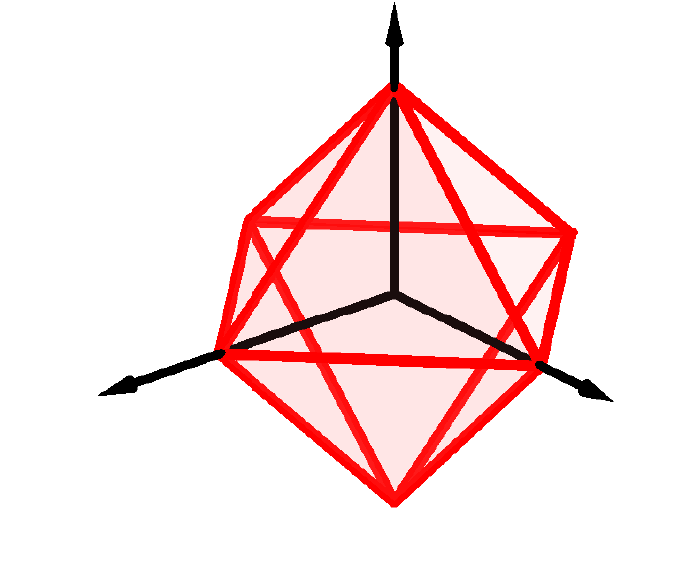} }}%
    \hspace{0.25in}
    \subfloat[ $\mathcal{B}_{\ell_\infty}$\\ (anti-sparse)]{{\includegraphics[width=2cm, trim=5cm 9.8cm 4cm 10.5cm,clip]{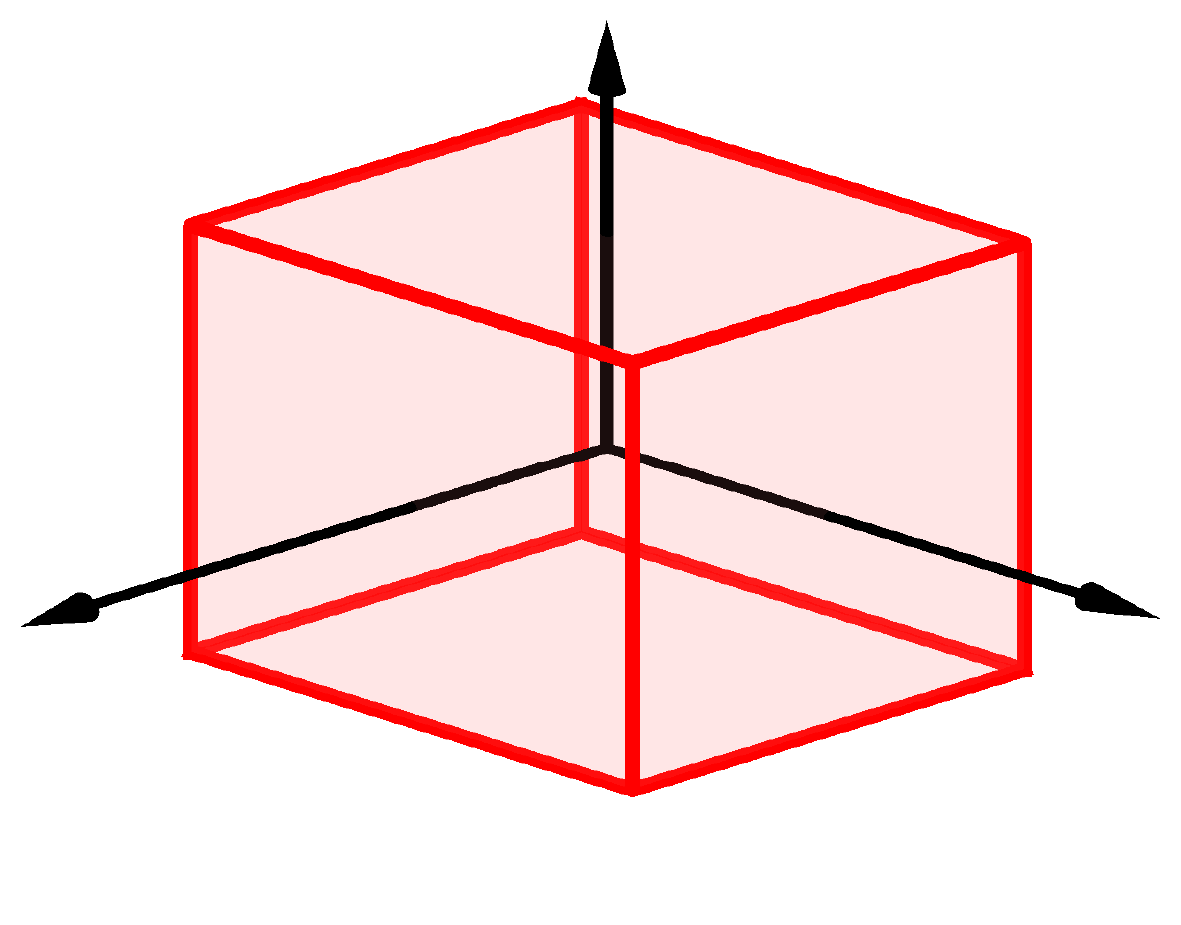} }}
    \hspace{0.25in}
    \subfloat[ $\Delta$ \\ (normalized nonnegative)]{{\includegraphics[width=1.8cm, trim=7.0cm 12.3cm 8.5cm 11.0cm,clip]{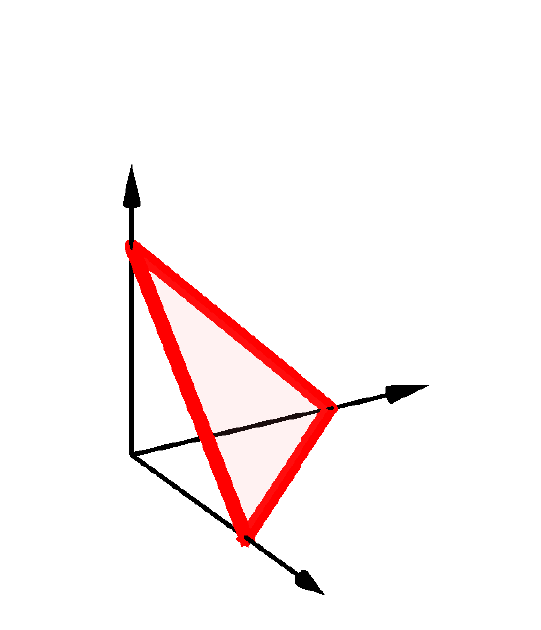} }}
    \hspace{0.25in}
    \subfloat[ $\mathcal{B}_{\ell_\infty,+}$\\ (nonnegative \\anti-sparse)]{\includegraphics[width=2.4cm, trim=6cm 12.0cm 5cm 10.5cm,clip]{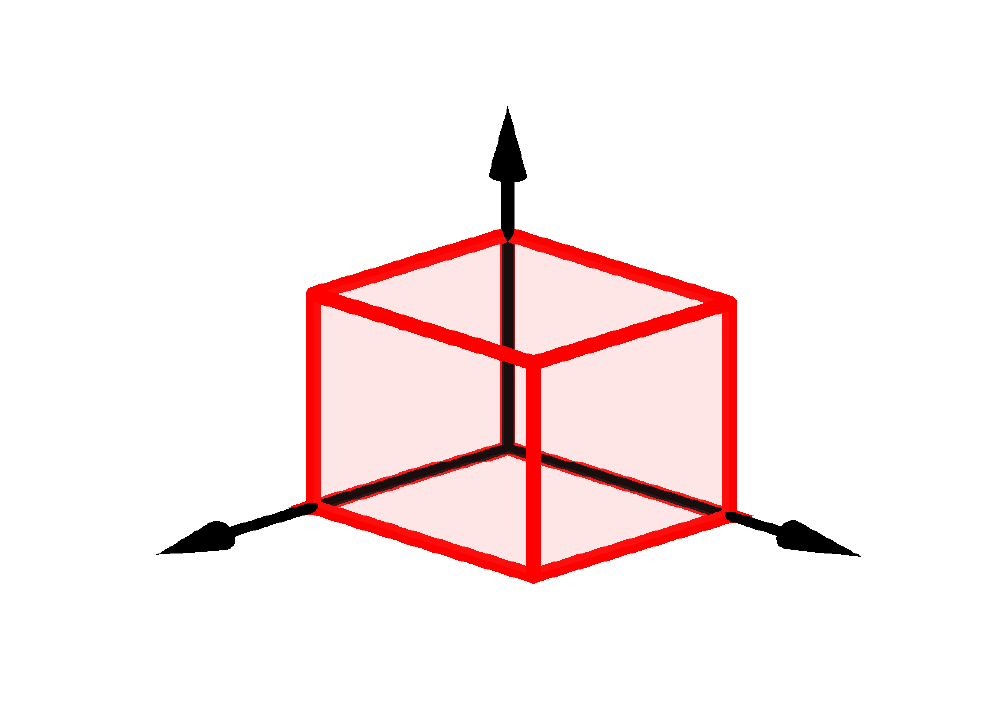} }%
    \hspace{0.25in}
    \subfloat[ $\mathcal{B}_{\ell_1,+}$\\ (nonnegative  sparse)]{\includegraphics[width=1.8cm, trim=7.0cm 12.3cm 8.5cm 12.0cm,clip]{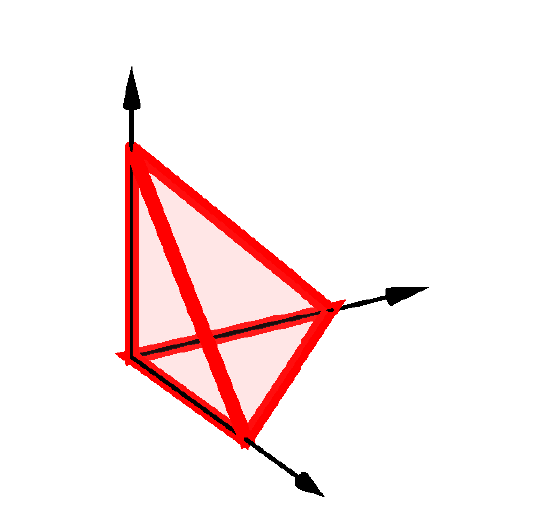} }
    \caption{Examples of source domains leading to identifiable generative models.}%
    \label{fig:sourcedomains}%
\end{figure}

\begin{figure}[ht]
\centering
\subfloat[General source domain with  sparse components]{{\includegraphics[width=7.5cm, trim=13.5cm 13.0cm 15.5cm 0.0cm,clip]{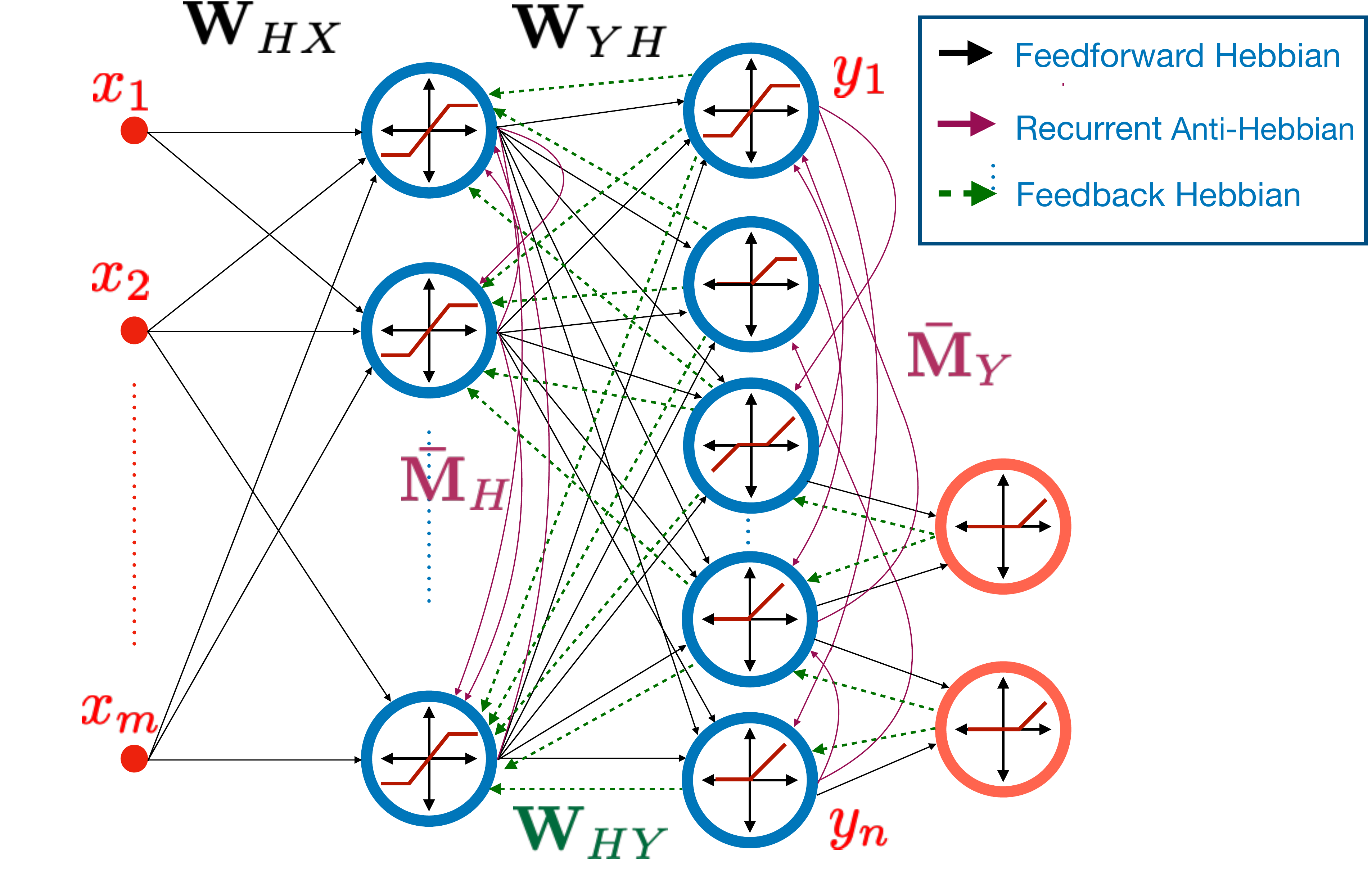} }}
\hspace{0.1in}
\subfloat[ Antisparse sources]{{\includegraphics[width=5.5cm, trim=13.5cm 15.0cm 27.5cm 0.0cm,clip]{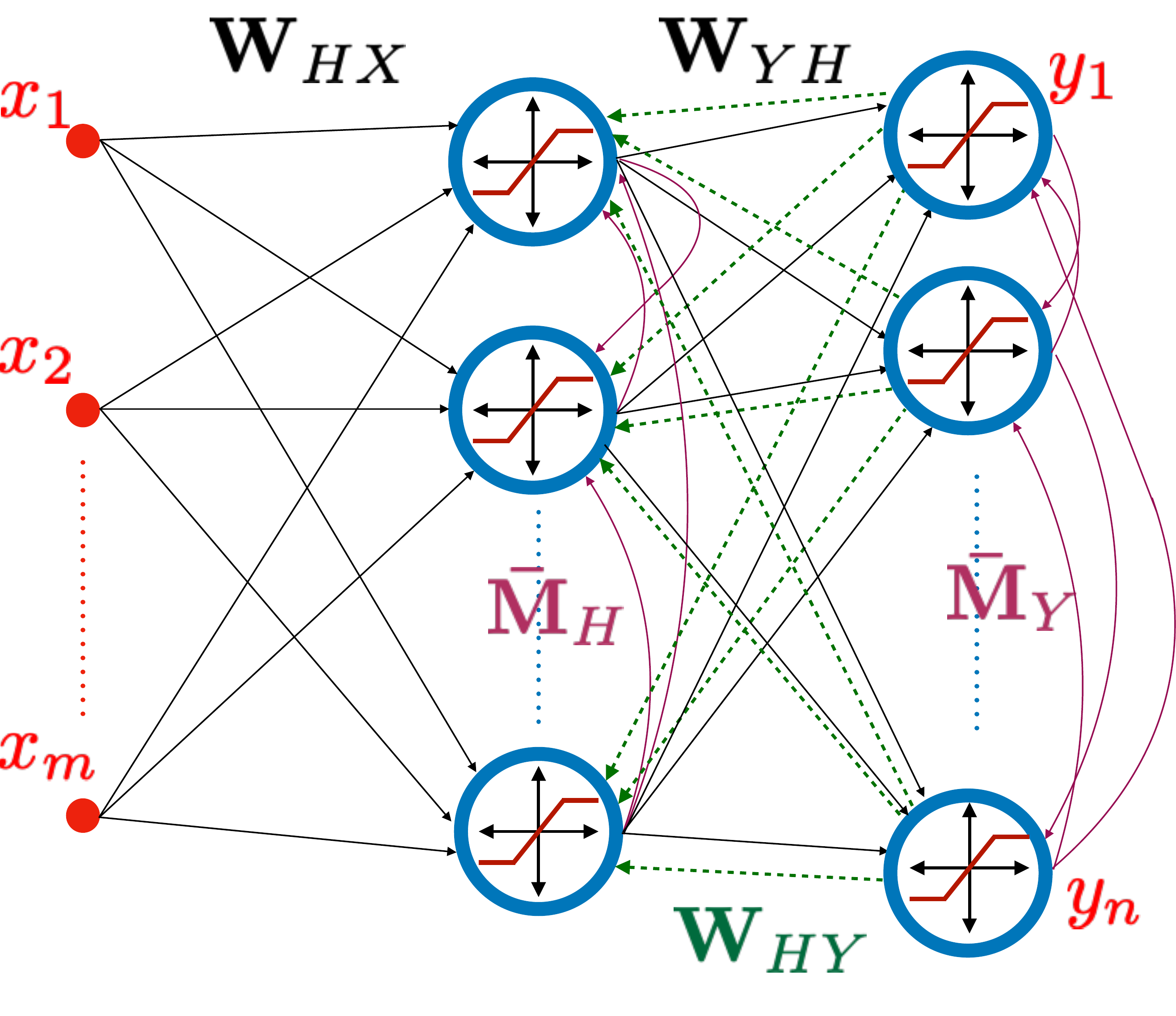} }}
	\caption{ Det-Max WSM neural network for blind source separation. The network takes a mixed input $\vx$ and produces latent components $\vy$ at the output. The output layer depends on the choice of source domain. Mutually sparse components are connected by inhibitory neurons at the output layer.}
	\label{fig:GenericNN}
\end{figure}

\subsection{Other related work}

 

Several algorithms for separation of linearly mixed and correlated sources have been proposed outside the domain of biologically-plausible BSS.  These algorithms make other forms of assumptions on the latent sources. 
Nonnegative matrix factorization (NMF) assumes that the latent vectors are nonnegative 
\citep{chen1984nonnegative, paatero1994positive,cichocki2009nonnegative,fu2019nonnegative}. Simplex structured matrix factorization (SSMF) assumes that the latent vectors are members of the unit-simplex 
\citep{chan2011simplex,lin2018maximum,lin2015identifiability}.
Sparse component analysis (SCA) often assumes that the latent vectors lie in the unity $\ell_1$-norm-ball \citep{georgiev2005sparse,donoho2006most,rozell2007locally,elad2010sparse,babatas2018algorithmic,babatas2020time}. Antisparse bounded component analysis (BCA) assumes latent vectors are in the $\ell_\infty$-norm-ball \citep{cruces2010bounded,erdogan2013class,inan2014convolutive}. Recently introduced polytopic matrix factorization (PMF) extends the identifiability-enabling domains to infinitely many polytopes obeying a particular symmetry restriction \citep{tatli:2021icassp,tatli2021tspsubmitted,bozkurt:2022icassp}.

The mapping of optimization algorithms to biologically-plausible neural networks have been formalized in the similarity matching framework \citep{pehlevan2014hebbian,pehlevan2015normative,sengupta2018manifold,pehlevan2019neuroscience}. 
Several BSS algorithms were proposed within this framework: 1) Nonnegative Similarity Matching (NSM) \citep{pehlevan2017blind,pehlevan2019spiking} separates linear mixtures of uncorrelated nonnegative sources, 2) \cite{bahroun2021normative} separates independent sources,
 and 3) Bounded Similarity Matching (BSM) separates uncorrelated anti-sparse bounded sources from $\ell_\infty$-norm-ball \citep{erdogan2020blind}. BSM  introduced a weighted inner product-based similarity criterion, referred to as the weighted similarity matching (WSM). Compared to these algorithm, the neural network algorithms we propose in this article 1) cover more general source domains, 
2) handle potentially correlated sources, 
3) use a two-layer WSM architecture (relative to single layer WSM architecture of BSM, which is not capable of generating arbitrary linear transformations) and  4) offer a general framework for neural-network-based optimization of the Det-Max criterion.




\section{Problem statement}
	\label{sec:BSSsetting}
	

\subsection{Sources}\label{sec:sources}

We assume that there are $n$ real-valued sources, represented by the vector $\vs\in \Pcal$, where $\Pcal$ is a  particular subset of $\mathbb{R}^n$. Our algorithms will address a wide range of source domains. We list some examples before giving a more general criterion:
\begin{itemize}
\item {\it Bounded sparse sources}: A natural convex domain choice for sparse sources is the unit    $\ell_1$ norm ball $\mathcal{B}_{\ell_1}=\{\vs \hspace{0.1in} \vert  \hspace{0.1in}  \|\vs\|_1 \le  1\}$ (Figure \ref{fig:sourcedomains}.(a)).  The use of $\ell_1$-norm as a convex (non)sparsity measure has been quite successful with various applications including sparse dictionary learning/component analysis \citep{donoho2006most,elad2010sparse,kreutz2003dictionary, li2004analysis,babatas2018algorithmic} and modeling of V1 receptive fields \cite{olshausen1996emergence}.
		
\item {\it Bounded anti-sparse sources}: A common domain choice for anti-sparse sources is the unit $\ell_\infty$-norm-ball:
$\mathcal{B}_{\ell_\infty}=\{\vs \hspace{0.1in} \vert  \hspace{0.1in}  \|\vs\|_\infty\le  1\}$
		(Figure \ref{fig:sourcedomains}.(b)).  If  vectors drawn from  $\mathcal{B}_{\ell_\infty}$ are well-spread inside this set, some samples would contain near-peak magnitude values simultaneously  at all their components. The potential equal spreading of values among the components justifies  the term ``anti-sparse'' \citep{elvira2016bayesian} or ``democratic'' \citep{studer2014democratic} component representations. This choice is well-suited  for both  applications in natural images and digital communication constellations \citep{cruces2010bounded, erdogan2013class}.
		
\item {\it Normalized nonnegative sources}:  Simplex structured matrix factorization \citep{chan2011simplex,lin2015identifiability,lin2018maximum}
uses the unit simplex \citep{donoho2003does, fu2019nonnegative} ${\Delta}=\{\vs \hspace{0.1in} \vert  \hspace{0.1in}  \vs\ge 0, \mathbf{1}^T\vs=1\}$ (Figure \ref{fig:sourcedomains}.(c)) as the source domain. Nonnegativity of sources naturally arises in biological context, for example in demixing olfactory mixtures \cite{grabska2017probabilistic}.
		
\item {\it Nonnegative bounded anti-sparse sources}: A non-degenerate polytopic choice of the nonnegative sources can be obtained through the combination of anti-sparseness and nonnegativity constraints. This corresponds to the intersection of  $\mathcal{B}_{\ell_\infty}$ with the nonnegative orthant $\mathbb{R}^n_+$,  represented as  $\mathcal{B}_{\ell_\infty,+}= \mathcal{B}_{\ell_\infty}\cap \mathbb{R}^n_+$ \citep{tatli2021tspsubmitted} (Figure \ref{fig:sourcedomains}.(d)).
\item {\it Nonnegative bounded sparse sources}: Another polytopic choice for nonnegative sources can be obtained through combination of the sparsity and nonnegativity constraints which yields the intersection of  $\mathcal{B}_{\ell_1}$ with the nonnegative orthant $\mathbb{R}_+$, \citep{tatli2021tspsubmitted}: 
$\mathcal{B}_{\ell_1,+}= \mathcal{B}_{\ell_1}\cap \mathbb{R}^n_+$ (Figure \ref{fig:sourcedomains}.(e)).
\end{itemize}
	
 Except the unit simplex $\Delta$, all the examples above are examples of an infinite set of identifiable polytopes whose symmetry groups are restricted to the combinations of component permutations and sign alterations  
as formalized in 
PMF framework for BSS \citep{tatli:2021icassp}. 
%
Further, instead of a homogeneous choice of features, such as sparsity and nonnegativity, globally imposed on all elements of the component vector,  we can assign these attributes at the subvector level and still obtain identifiable polytopes. For example, the reference  \citep{tatli2021tspsubmitted} provides the set
$\Pcal_{ex}=\left\{\mathbf{s}\in \mathbb{R}^3\ \middle\vert  s_1,s_2\in[-1,1],\, s_3\in[0,1],\,  \left\|\left[\begin{array}{c} s_1 \\ s_2 \end{array}\right]\right\|_1\le 1,\, \left\|\left[\begin{array}{c} s_2 \\ s_3 \end{array}\right]\right\|_1\le 1 \right\}$,
as a simple  illustration  of such polytopes with heterogeneous structure where $s_3$ is nonnegative, $s_1,s_2$ are signed, and $ \left[\begin{array}{cc} s_1 & s_2 \end{array}\right]^T$, $\left[\begin{array}{cc} s_2 & s_3 \end{array}\right]^T$ are sparse subvectors, while sparsity is not globally imposed. 
In this article, we concentrate on particular source domains including the unit simplex, and the subset of identifiable polytopes for which the attributes such as sparsity and nonnegativity are defined at the subvector level in the general form
\begin{eqnarray}
\Pcal=\left\{\mathbf{s}\in \mathbb{R}^n\ \middle\vert  s_i\in[-1,1] \, \forall i\in \mathcal{I}_s,\, s_i\in[0,1] \, \forall i\in \mathcal{I}_+, \, \left\|\vs_{\mathcal{J}_k}\right\|_1\le 1, \, \mathcal{J}_k\subseteq \mathbb{Z}_n, \, k\in\mathbb{Z}_L  \right\}, \label{eq:polygeneral}
\end{eqnarray}
where $\mathcal{I}_+\subseteq \mathcal{Z}_n$ is the index set for nonnegative sources, and $I_s$ is its complement, $\vs_{\mathcal{J}_k}$ is the subvector constructed from the elements with indices in $\mathcal{J}_k$, and $L$ is the number of sparsity constraints imposed in the subvector level.
	


The Det-Max criterion for BSS is based on the assumption that the source samples are well-spread in their presumed domain. The references \cite{fu2018identifiability} and \cite{tatli2021tspsubmitted} provide precise conditions on the scattering of source samples which guarantee their identifiability for the unit simplex and polytopes, respectively. Appendix \ref{appsec:suffscat} provides a brief summary of these conditions.

We emphasize that our assumptions about the sources are deterministic.  Therefore, our proposed algorithms do not exploit any stochastic assumptions such as independence or uncorrelatedness, and can separate both independent and dependent (potentially correlated) sources. 

\subsection{Mixing}
\label{sec:mixingmodel}
 The sources $\vs_t$ are mixed through a mixing  matrix $\vA\in \mathbb{R}^{m \times n}$. 
\begin{eqnarray}
\vx_t=\vA \vs_t, \hspace{0.4in} t \in \mathbb{Z}.  \label{eq:mixing}
\end{eqnarray}
We only consider the (over)determined case with $m\ge n$ and assume that the mixing matrix is full-rank. While we consider noiseless mixtures to achieve perfect separability, the optimization setting proposed for the online algorithm features a particular objective function that safeguards against potential noise presence. 
We use 
$ \vS(t)=\left[\begin{array}{ccc} \vs_1 & \ldots & \vs_t \end{array}\right] \in \mathbb{R}^{n \times t}$ and $\vX(t)=\left[\begin{array}{ccc} \vx_1 & \ldots & \vx_t \end{array}\right]\in \mathbb{R}^{m \times t}$
to represent data snapshot matrices, at time $t$,  for sources and mixtures, respectively.

\subsection{Separation}
The goal of the source separation  is to obtain an estimate of  $\vS(t)$ from the mixture measurements $\vX(t)$ when the mixing matrix $\vA$ is unknown.  We use the notation $\vy_t$ to refer to source estimates, which are linear transformations of  observations, i.e., 
$\vy_i=\vW \vx_i$, 
where $\vW \in \mathbb{R}^{n \times m}$. We define
$\vY(t)=\left[\begin{array}{cccc} \vy_1 & \vy_2& \ldots & \vy_t \end{array}\right]\in \mathbb{R}^{n \times t}$
as the output snapshot matrix. "Ideal separation" is defined as the condition where the  outputs are scaled and permuted versions of original sources, i.e., they satisfy
$\vy_t=\vP \bLambd\vs_t$, 
 where $\vP$ is a permutation matrix, and $\bLambd$ is a full rank diagonal matrix.


\section{Determinant maximization based blind source separation}
\label{sec:DetMaxBSS}

 Among several alternative solution methods for the BSS problem, the determinant-maximization (Det-Max) criterion has been proposed within the NMF, BCA, and PMF frameworks,  \citep{schachtner2011towards,fu2019nonnegative,erdogan2013class,babatas2018algorithmic, tatli:2021icassp,tatli2021tspsubmitted}. Here, the separator is trained to maximize  the  (log)-determinant of the sample  correlation matrix for the separator outputs, 
$J(\vW)= \log(\det(\hat{\vR}_y(t)))$, 
where $\hat{\vR}_y(t)$ is the sample correlation matrix
$\hat{\vR}_y(t)=\frac{1}{t}\sum_{i=1}^{t}\vy_i \vy_i^T = \frac{1}{t} \vY(t)\vY(t)^T$.
Further, during the training process, the separator outputs are constrained to lie inside the presumed source domain, i.e. $\mathcal{P}$. As a result, we can  pose the corresponding optimization problem  as \cite{fu2019nonnegative, tatli2021tspsubmitted} 
\begin{maxi!}[l]<b>
{\vY(t)}{ \log(\det(\vY(t)\vY(t)^T))\label{eq:detmaxobjective}}{\label{eq:bssobjective}}{}
\addConstraint{\vy_i \in \mathcal{P}, i=1, \dots, t,}{\label{eq:detmax}}{}
\end{maxi!}
where we ignored the constant $\frac{1}{t}$ term. Here, the determinant of the correlation matrix acts as a spread measure for the output samples. If the original source samples $\{\vs_1, \ldots, \vs_t\}$ are sufficiently scattered inside the source domain $\mathcal{P}$, as described in Section \ref{sec:sources} and Appendix \ref{appsec:suffscat}, then the global solution of this optimization can be shown to achieve perfect separation  \citep{fu2018identifiability, fu2019nonnegative, tatli2021tspsubmitted}.




\section{An alternative optimization formulation of determinant-maximization based on weighted similarity matching}
\label{sec:WSM}

Here, we reformulate the Det-Max problem \ref{eq:bssobjective} described above in a way that allows derivation of a biologically-plausible neural network for the linear BSS setup  in Section \ref{sec:BSSsetting}. 
Our formulation applies to all source types discusses in \ref{sec:sources}. 

We propose the following optimization problem:
\begin{mini!}[l]<b>
{\substack{\vY(t),\vH(t),  D_{1,11}(t), \ldots D_{1,nn}(t),\vD_1(t)\\
 D_{2,11}(t), \ldots D_{2,nn}(t),\vD_2(t)}} {\sum_{i=1}^n \log(D_{1,ii}(t))+\sum_{i=1}^n \log(D_{2,ii}(t))\label{eq:wsm3objective}}{\label{eq:wsm3optimization}}{}
\addConstraint{\vX(t)^T\vX(t) -\vH(t)^T\vD_1(t)\vH(t)=0}{\label{eq:wsm3constr1}}{}
\addConstraint{\vH(t)^T\vH(t)-\vY(t)^T\vD_2(t)\vY(t)=0}{\label{eq:wsm3constr2}}{}
\addConstraint{ \vy_i \in \mathcal{P}, i=1, \ldots, n}{\label{eq:wsm3constr3}}{}
\addConstraint{\vD_l(t)=\text{diag}(D_{l,11}(t), \ldots, D_{l,nn}(t)), \hspace{0.1in} l=1,2}{\label{eq:wsm3constr4}}{}
\addConstraint{D_{l,11}(t), D_{l,22}(t), \ldots, D_{l,nn}(t)>0,  \hspace{0.1in} l=1,2}{\label{eq:wsm3constr5}}{}
\end{mini!}
Here,  $\vX(t)\in\mathbb{R}^{m\times t}$ is the matrix containing input (mixture) vectors, $\vY(t)\in\mathbb{R}^{n\times t}$ is the matrix containing output vectors, $\vH(t)\in\mathbb{R}^{n\times t}$ is a slack variable containing an intermediate signal \\ $\{\vh_i \in \mathbb{R}^n, i=1, \ldots, t\}$, corresponding to the hidden layer of the neural network implementation in Section \ref{sec:nnsolutiontoWSM}, in its columns $\vH(t)=\left[\begin{array}{cccc} \vh_1 & \vh_2 & \ldots & \vh_t \end{array}\right]$. $D_{l,11}(t), D_{l,22}(t), \ldots, D_{l,nn}(t)$ for $l=1,2$ are nonnegative slack variables  to be described below, and $\vD_l$ is the diagonal matrix containing weights $D_{l,ii}$ for  $i=1, \dots,n$ and $l=1,2$. The constraint  \eqref{eq:wsm3constr3} ensures that the outputs lie in the presumed domain of sources.

This problem is related to the weighted similarity matching (WSM) objective introduced in \cite{erdogan2020blind}. Constraints \eqref{eq:wsm3constr1} and \eqref{eq:wsm3constr2} define two separate WSM conditions.
In particular, the  equality constraint in \eqref{eq:wsm3constr1}  is a WSM constraint between inputs and the intermediate signal $\vH(t)$. This constraint imposes  that the pairwise weighted correlations of the signal $\{\vh_i, i=1,\ldots, t\}$ are the same as correlations among the elements of the input signal $\{\vx_i, i=1, \ldots, t\}$, i.e.,  $ \vx_i^T\vx_j= \vh_i^T\vD_1(t)\vh_j, \hspace{0.1in} \forall i,j \in \{1, \ldots, t\}$. $D_{1,11}(t), D_{1,22}(t), \ldots, D_{1,nn}(t)$ correspond to inner product weights used in these equalities. Similarly, the equality constraint in (\ref{eq:wsm3constr2})  defines a WSM constraint between the intermediate signal and outputs. This equality can be written as $\vh_i^T\vh_j=\vy_i^T\vD_2(t)\vy_j, \hspace{0.1in} i,j \in \{1, \ldots, t\}$, and $D_{2,11}(t), D_{2,22}(t), \ldots, D_{2,nn}(t)$ correspond to the inner product weights used in these equalities. The optimization involves minimizing the logarithm of the determinant of the weighting matrices. 

Now we state the relation between our WSM-based objective and the original Det-Max criterion \eqref{eq:bssobjective}.
\begin{theorem}
If $\vX(t)$ is  full column-rank, then global optimal $\bY(t)$ solutions   of \eqref{eq:bssobjective} and \eqref{eq:wsm3optimization} coincide. 
\label{th:1}
\end{theorem}
\begin{proof}[Proof of Theorem \ref{th:1}]
See Appendix \ref{app:thm1proof} for the proof. The proof relies on a lemma that states that the optimization constraints enforce inputs and outputs to be related by an arbitrary linear transformation.
\end{proof}


\section{Biologically-plausible neural networks for WSM-based BSS}
\label{sec:nnsolutiontoWSM}

The optimization problems we considered so far were in an offline setting, where all inputs are observed together and all outputs are produced together. However, biology operates in an online fashion, observing an input and producing the corresponding output, before seeing the next input. Therefore, in this section, we first introduce an online version of the batch WSM-problem  \eqref{eq:wsm3optimization}. 
Then we show that the corresponding  gradient descent algorithm leads to a two-layer neural network with biologically-plausible local update rules.

\subsection{Online optimization setting for WSM-based BSS}
\label{sec:onlinewsm}

We first propose an online extension of  WSM-based BSS   (\ref{eq:wsm3optimization}). In the online setting,
past outputs cannot be altered, but past inputs and outputs still carry valuable information about solving the BSS problem. We will write down an optimization problem whose goal is to produce the sources $\vy_t$ given a mixture $\vx_t$, while exploiting information from all the fixed previous inputs and outputs.
%
%

 We first introduce our notation. We consider exponential weighting of the signals  as a recipe for dynamical adjustment to potential nonstationarity in the data. We define the weighted  input data  snapshot matrix  by time $t$ as, $\bwX(t)=\left[\begin{array}{cccc}  \gamma^{t-1} \vwx_1 & \ldots & \gamma \vwx_{t-1} & \vwx_t \end{array}\right]=\bX(t)\bGam(t) 
$,
	where $\gamma$ is the forgetting factor and $\bGam(t)=\mbox{diag}(\gamma^{t-1}, \ldots, \gamma,1)$. The exponential weighting 
	emphasizes recent mixtures by reducing the impact of past samples. Similarly, we define the corresponding weighted output snapshot matrix for output as
	$
	\bwY(t)=\left[\begin{array}{cccc}  \gamma^{t-1} \vwy_1   & \ldots & \gamma \vwy_{t-1} & \vwy_t \end{array}\right]=\bY(t)\bGam(t)$, 
	and the hidden layer vectors as
		$ \bwH(t)=\left[\begin{array}{cccc}  \gamma^{t-1} \vh_1  & \ldots & \gamma \vh_{t-1} & \vh_t \end{array}\right]= \vH(t)\bGam(t)$.
 We  further define $\tau=\lim_{t\rightarrow \infty}=\sum_{k=0}^{t-1}\gamma^{2k}=\frac{1}{1-\gamma^2}$
  as a measure of  the effective time window length for sample correlation calculations based on the exponential weights. 
	
In order to derive an online cost function, we first converted  equality constraints in (\ref{eq:wsm3constr1}) and (\ref{eq:wsm3constr2})  to similarity matching cost functions $J_{1}(\vH(t),\vD_1(t))=\frac{1}{2\tau^2}\|\bwX(t)^T\bwX(t)-\bwH(t)^T\vD_1(t)\bwH(t)\|_F^2$, $J_{2}(\vH(t),\vD_2(t), \vY(t))=\frac{1}{2\tau^2}\|\bwH(t)^T\bwH(t)-\bwY(t)^T\vD_2(t)\bwY(t)\|_F^2$.
%
Then, a weighted combination of similarity matching costs and  the objective function in (\ref{eq:wsm3objective}) yields  the final cost function
	\begin{eqnarray}
\mathcal{J}(\vH(t),\vD_1(t),\vD_2(t),\vY(t))&=&\lambda_{SM}[\beta J_1(\vH(t),\vD_1(t))+(1-\beta)J_2(\vH(t),\vD_2(t),\vY(t))]\nonumber \\
	&&+(1 - \lambda_{SM})[\sum_{k=1}^n \log(D_{1,kk}(t))+\sum_{k=1}^n \log(D_{2,kk}(t))]. \label{eq:onlinecost}
	\end{eqnarray}
Here, $\beta\in [0,1]$ and $\lambda_{SM}\in [0,1]$ are parameters that convexly combine similarity matching costs and the objective function. Finally, we can state the online optimization problem for determining the current output $\vy_t$, the corresponding hidden state $\vh_t$ and for updating the gain parameters $\vD_l(t)$ for $l=1,2$, as 
	\begin{mini!}[l]<b>
{\substack{\vy_t,\vh_t,  D_{1,11}(t), \ldots D_{1,nn}(t),\vD_1(t)\\
 D_{2,11}(t), \ldots D_{2,nn}(t),\vD_2(t)}} { \mathcal{J}(\vH(t),\vD_1(t),\vD_2(t),\vY(t))\label{eq:wsmonlineobj}}{\label{eq:wsmonline}}{}
\addConstraint{\vy_t\in \Pcal}{\label{eq:wsmonlineconstr1}}{}
\addConstraint{\vD_l(t)=\text{diag}(D_{l,11}(t), \ldots, D_{l,nn}(t)), \hspace{0.1in} l=1,2}{\label{eq:wsmonlineconstr2}}{}
\addConstraint{D_{l,11}(t), D_{l,22}(t), \ldots, D_{l,nn}(t)>0,  \hspace{0.1in} l=1,2}{\label{eq:wsmonlineconstr3}}{}
\end{mini!}
As shown in Appendix \ref{sec:smcsimplify}, part of $\mathcal{J}$ that depends on  $\vh_t$ and $\vy_t$ can be written as 
\begin{eqnarray}
	C(\vh_t,\vy_t)&=&2\vh_t^T\vD_1\vM_H(t)\vD_1(t)\vh_t-4\vh_t^T\vD_1(t)\vW_{HX}(t)\vx_t \nonumber\\
	&& +2\vy_t^T\vD_2(t)\vM_Y(t)\vD_2(t)\vy_t-4\vy_t^T\vD_2(t)\vW_{YH}(t)\vh_t + 2\vh_t^T\vM_H(t)\vh_t, \label{eq:c}
\end{eqnarray}
where the dependence on past inputs and outputs appear in the weighted correlation matrices:
\begin{gather}
\begin{array}{c}
\vM_H(t)=\frac{1}{\tau}\sum_{k=1}^{t-1}(\gamma^2)^{t-1-k}\vh_k\vh_{k}^T, \quad
\vW_{HX}(t)=\frac{1}{\tau}\sum_{k=1}^{t-1}(\gamma^2)^{t-1-k}\vh_k\vx_{k}^T, \\
\vW_{YH}(t)=\frac{1}{\tau}\sum_{k=1}^{t-1}(\gamma^2)^{t-1-k}\vy_k\vh_{k}^T, \quad
\vM_Y(t)=\frac{1}{\tau}\sum_{k=1}^{t-1}(\gamma^2)^{t-1-k}\vy_k\vy_{k}^T.\end{array} \label{eq:synapseweights}
\end{gather}


\subsection{Description of network dynamics for bounded anti-sparse sources}
\label{linfdynamics}

We now show that the gradient-descent minimization of the online WSM cost function in (\ref{eq:wsmonline}) can be interpreted as the dynamics of a neural network with local learning rules.  
The exact network architecture is determined by the presumed identifiable source domain $\mathcal{P}$, which can be chosen in infinitely many ways. In this section,  we concentrate on the domain choice $\mathcal{P}=\mathcal{B}_{\infty}$ as an illustrative example.  In Section \ref{sec:generalsourcedomains}, we discuss how to generalize the results of this section by modifying the output layer for different identifiable source domains. We start by writing  the update expressions for the optimization variables based on the gradients of  $\mathcal{J}(\vh_t,\vy_t, \vD_1(t),\vD_2(t))$:

\underline{Update dynamics for $\vh_t$:}  Following previous work \citep{rozell2008sparse,pehlevan2019spiking}, and using the gradient of \eqref{eq:c} in (\ref{eq:gradht}) with respect to $\vh_t$, we can write down an update dynamics for $\vh_t$ in the form
	\begin{eqnarray}
	\frac{d\vv(\tau)}{d\tau} &=&-\vv(\tau)- \lambda_{SM} [((1-\beta )\bar{\vM}_H(t)+\beta\vD_1(t)\bar{\vM}_H(t)\vD_1(t))\vh(\tau) \nonumber \\
	                                && \hspace*{-0.25in}+\beta \vD_1(t)\vW_{HX}(t)\vx(\tau)+(1-\beta)\vW_{YH}(t)^T\vD_2(t)\vy(\tau)]  \label{eq:descv}\\
 \vh_{t,i}(\tau)&=&\sigma_A\left(\frac{\vv_i(\tau)}{{\lambda_{SM}\Gamma_{H}}_{ii}(t)((1 - \beta)+\beta {D_{1,ii}(t)}^2)}\right), \hspace{0.1in} \text{for } i=1, \ldots n, \label{eq:htnn}
\end{eqnarray}
where $\bGam_H(t)$ is a diagonal matrix containing diagonal elements of $\vM_H(t)$ and $\bar{\vM}_H(t)=\vM_H(t)-\bGam_H(t)$, $\sigma(\cdot)$ is the clipping function, defined as $\sigma_A(x)=\left\{\begin{array}{cc}  x & -A \le x \le A, \\
   A\text{sign}(x) & \text{otherwise.} \end{array} \right.$.
This dynamics can be shown to minimize \eqref{eq:c} \cite{rozell2008sparse}. Here $\vv(\tau)$ is an internal variable that could be interpreted as the voltage dynamics of a biological neuron, and is defined based on a linear transformation of $\vh_t$ in (\ref{eq:v}). 
Equation (\ref{eq:descv}) defines $\vv(\tau)$ dynamics from the gradient of \eqref{eq:c} with respect to $\vh_t$ in (\ref{eq:gradht}). Due to the positive definite linear map in (\ref{eq:v}), the expression in (\ref{eq:gradht}) also serves as a descent direction for $\vv(\tau)$. Furthermore, $\sigma(\cdot)$ function is the projection onto $A\mathcal{B}_\infty$, where $[-A,A]$ is the presumed dynamic range for the components of $\vh_t$. We note that there is no explicit constraint set for $\vh_t$ in the online optimization setting of Section \ref{sec:onlinewsm}, and therefore, $A$ can be chosen as large as desired in the actual implementation.  We included the nonlinearity in (\ref{eq:htnn}) to model the limited dynamic range of an actual (biological) neuron.

\underline{Update dynamics for output $\vy_t$:} We write the update dynamics for the output $\vy_t$, based on (\ref{eq:gradyt}) as
\begin{eqnarray}
\frac{d\vu(\tau)}{d\tau}&=&-\vu(\tau)+ \vW_{YH}(t)\vh(\tau)-\bar{\vM}_Y(t)\vD_2(t)\vy(\tau),\label{eq:descu}\\
\vy_{t,i}(\tau)&=& \sigma_1\left(\frac{\vu_i(\tau)}{{\Gamma_Y}_{ii}(t){D_{2,ii}(t)}}\right),  \hspace{0.2in} \text{for } i=1, \ldots n,   \label{eq:ytnn}
\end{eqnarray}
which is derived using the same approach for $\vh_t$, where we used the descent direction expression in (\ref{eq:gradyt}), and the substitution in (\ref{eq:ut}). Here, $\bGam_Y(t)$ is a diagonal matrix containing diagonal elements of $\vM_Y(t)$ and $\bar{\vM}_Y(t)=\vM_Y(t)-\bGam_Y(t)$. Note that the nonlinear mapping $\sigma_1(\cdot)$ is the  projection onto the presumed domain of sources, i.e.,  $\mathcal{P}=\mathcal{B}_\infty$, which is elementwise clipping operation.

The state space representations in (\ref{eq:descv})-(\ref{eq:htnn}) and (\ref{eq:descu})-(\ref{eq:ytnn}) correspond to a two-layer  recurrent neural network with input $\vx_t$, hidden layer activation $\vh_t$, output layer activation $\vy_t$, $\vW_{HX}$ ($\vW_{HX}^T$) and $\vW_{YH}$ ($\vW_{YH}^T$) are the feedforward (feedback) synaptic weight  matrices for the first and the second layers, respectively,  and $\bar{\vM}_{H}$ and $\bar{\vM}_Y$ are recurrent synaptic weight matrices for the first and  the second layers, respectively. The corresponding neural network schematic is provided in Figure \ref{fig:GenericNN}.(b). The gain and synaptic weight dynamics below describe the learning mechanism for this network:

\underline{Update dynamics for gains $D_{l,ii}$:} Using the derivative of the cost function with respect to $D_{1,ii}$ in (\ref{eq:derD1ii}),  we can write the dynamics corresponding to the gain variable $D_{1,ii}$ as
\begin{eqnarray}
&&\hspace*{-0.9in}\mu_{D_{1}}\frac{d D_{1,ii}(t)}{dt}=-{(\lambda_{SM} \beta) (\|{\vM_{H}}_{i,:}\|_{\vD_1(t)}^2-\|{\vW_{HX}}_{i,:}\|^2_2)-(1 - \lambda_{SM})\frac{1}{D_{1,ii}(t)}}, \label{eq:D1update}
\end{eqnarray}
where $\mu_{D_{1}}$ corresponds to the learning time-constant.
  Similarly, for the gain variable $D_{2,ii}$,  the corresponding coefficient dynamics expression based on (\ref{eq:derD2ii}) is given by
\begin{eqnarray}
&&\hspace*{-0.9in}\mu_{D_{2}}\frac{d D_{2,ii}(t)}{dt}=-{(\lambda_{SM} \beta) (\|{\vM_{Y}}_{i,:}\|_{\vD_2(t)}^2-\|{\vW_{YH}}_{i,:}\|^2_2)-(1 - \lambda_{SM})\frac{1}{D_{2,ii}(t)}  },\label{eq:D2update}
\end{eqnarray}
where $\mu_{D_{2}}$ corresponds to the learning time-constant.

The inverses of the inner product weights $D_{l,ii}$ correspond to homeostatic gain parameters. The inspection of the gain updates in (\ref{eq:D1update}) and (\ref{eq:D2update})  leads to an interesting observation: whether the corresponding gain is going to increase or decrease depends on the balance between the norms of the recurrent  and the feedforward synaptic strengths, which are the statistical indicators of the recent output and input activations, respectively. Hence, the homeostatic gain of the neuron will increase (decrease) if the level of recent output activations falls behind (surpasses)  the level of recent input activations to balance input/output energy levels. The resulting dynamics align with the experimental homeostatic balance observed in biological neurons \citep{turrigiano1999homeostatic}. 

Based on the definitions of  the synaptic weight matrices in (\ref{eq:synapseweights}), we can write their updates as
%
\begin{align}
\vM_H(t+1)&=\gamma^2\vM_H(t)+(1-\gamma^2)\vh_{t}\vh_{t}^T, \quad 
\vM_Y(t+1)=\gamma^2\vM_Y(t)+(1-\gamma^2)\vy_{t}\vy_{t}^T,\label{eq:sydynWYH}\\
\vW_{HX}(t+1)&= \gamma^2\vW_{HX}(t)+(1-\gamma^2)\vh_{t}\vx_{t}^T,\quad 
\vW_{YH}(t+1)=\gamma^2\vW_{YH}(t)+(1-\gamma^2)\vy_{t}\vh_{t}^T.\nonumber  
\end{align}
These updates are local in the sense that they only depend on variables available to the synapse, and hence are biologically plausible.





\subsection{Det-max WSM neural network examples for more general source domains}
\label{sec:generalsourcedomains}
Det-Max Neural Network obtained for the  source domain  $\Pcal=\mathcal{B}_\infty$ in Section \ref{linfdynamics} can be extended to more general identifiable source domains 
by only changing the output dynamics. 
In Appendix \ref{sec:dersourcedomains}, we provide illustrative examples 
for different identifiable domain choices. Table \ref{tab:nnoutdynamicsexamples} summarizes the output dynamics obtained for the identifiable source domain examples in Figure \ref{fig:sourcedomains}.

\begin{table}[h!]
    \centering
\caption{Example source domains from Figure \ref{fig:sourcedomains} and the corresponding output dynamics.}
    \label{tab:nnoutdynamicsexamples}
    \vspace{0.1in} 

    \begin{tabular}{c c l l}
    \hline
    & {\bf Source Domain} & {\bf Output Dynamics} & {\bf Output } \\
    &                    &                       & {\bf Activation}\\
    \arrayrulecolor{black}\hline
    \arrayrulecolor{white}\hline 
   \multirow{3}{*}{ \begin{minipage}{.11\textwidth}
   \includegraphics[width=2.0cm, trim=6cm 12.0cm 5cm 11.0cm,clip]{picture_results_new/Binftyplusplot.pdf}
   \end{minipage}} & $\Pcal=\mathcal{B}_{\infty,+}$ &  & \multirow{3}{*}{ \begin{minipage}{.125\textwidth}
   \includegraphics[width=1.6cm, trim=11cm 5.0cm 18.5cm 3.5cm,clip]{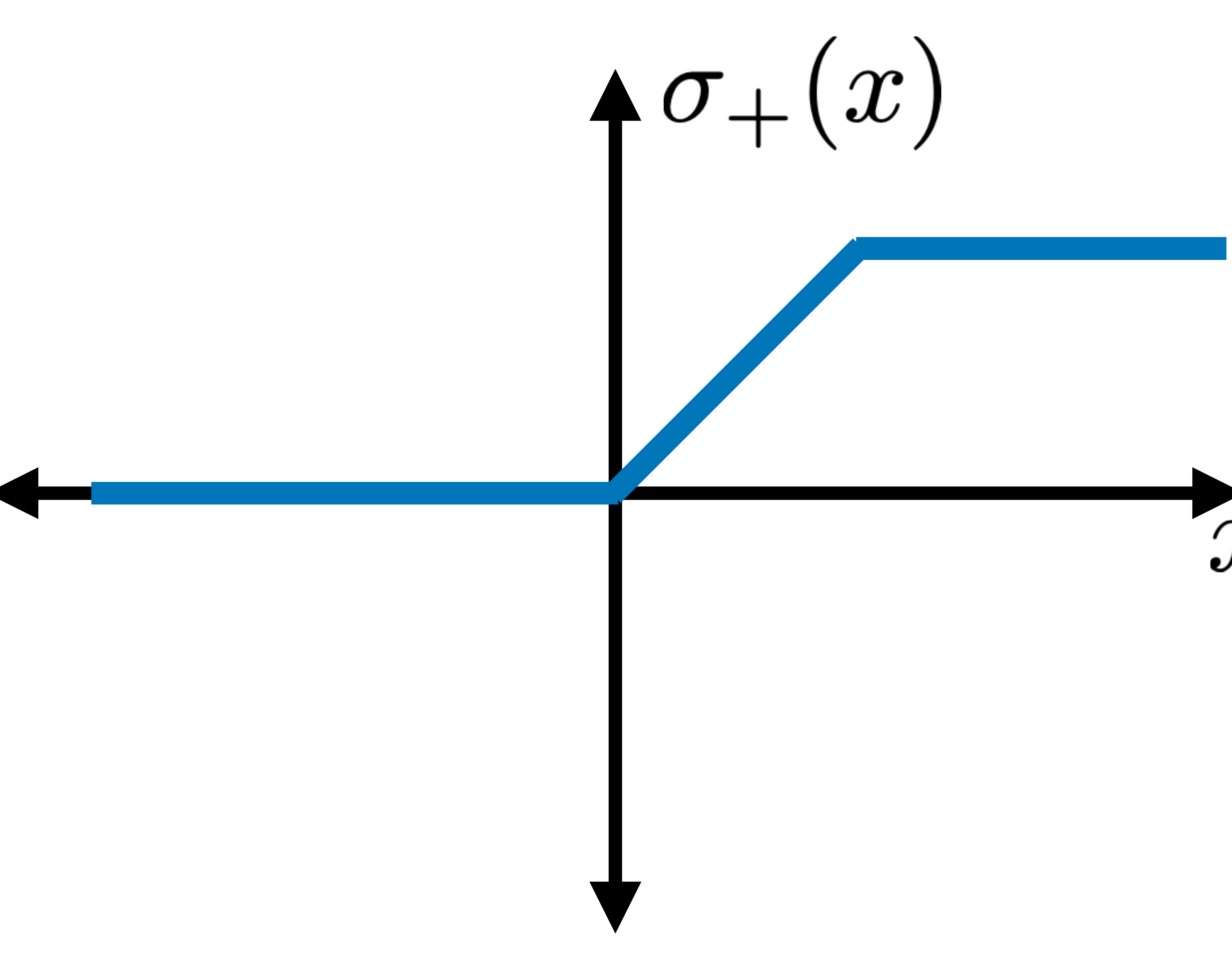}
   \end{minipage}}\\
    & Nonnegative & {\scriptsize $\vy_{t,i}(\tau)= \sigma_+\left(\frac{\vu_i(\tau)}{{\Gamma_Y}_{ii}(t){D_{2,ii}(t)}}\right)$ } & \\
   &  Anti-sparse & &  \\
    \arrayrulecolor{black}\hline 
    \multirow{3}{*}{ \begin{minipage}{.11\textwidth}
   \includegraphics[width=1.7cm, trim=7.5cm 12.7cm 6.5cm 12.0cm,clip]{picture_results_new/B1plot.pdf}
   \end{minipage}}& $\Pcal=\mathcal{B}_{1}$ &  {\scriptsize $\vy_{t,i}(\tau)= \text{ST}_{\lambda_1(\tau)}\left(\frac{u_i(\tau)}{\lambda_{SM}(1 - \beta){\Gamma_Y}_{ii}(t){D_{2,ii}(t)}} \right)$} &\multirow{3}{*}{ \begin{minipage}{.125\textwidth}
   \includegraphics[width=1.6cm, trim=11cm 5.0cm 18.5cm 3.5cm,clip]{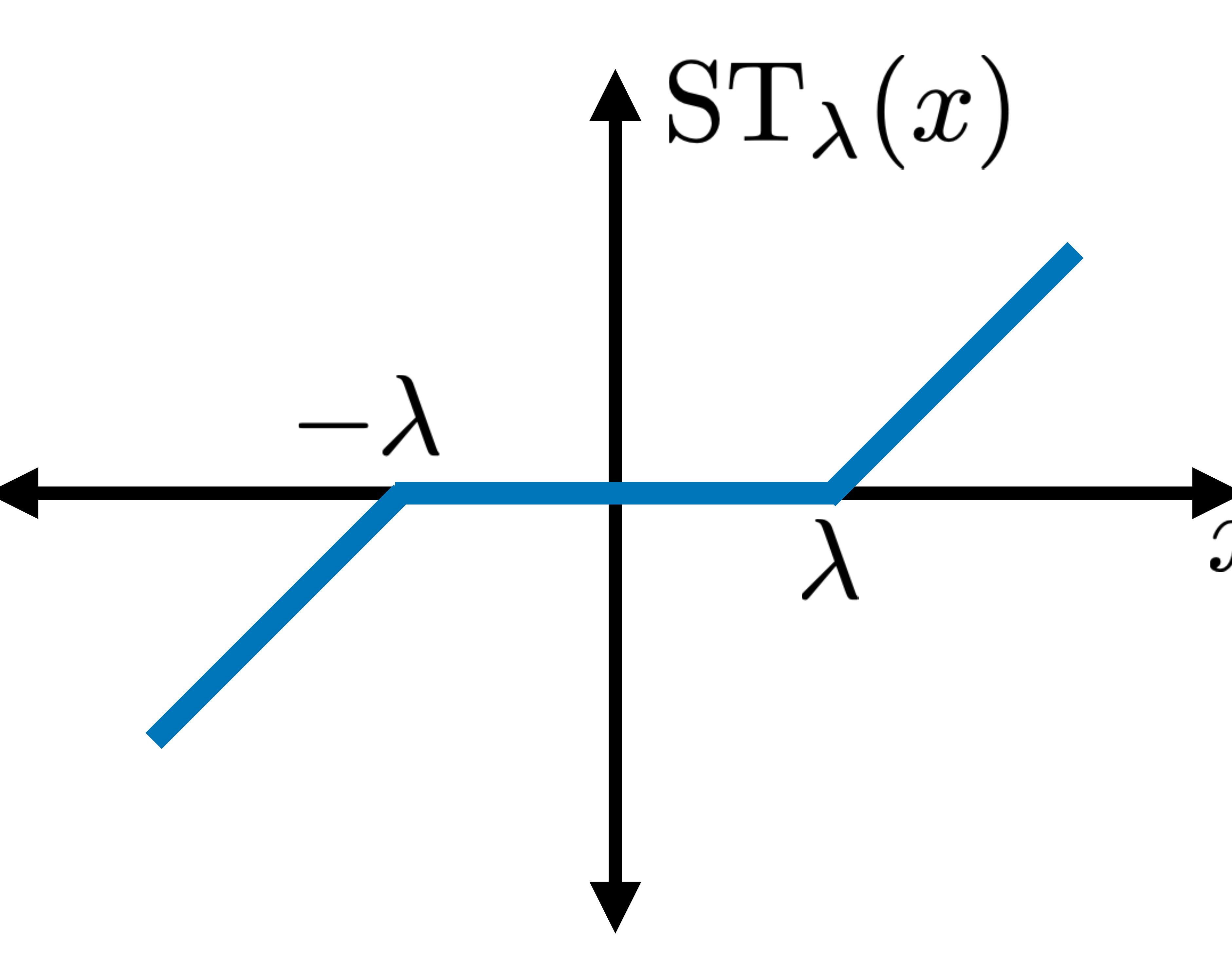}
   \end{minipage}}\\
    & Sparse &  {\scriptsize $\frac{da(\tau)}{d\tau}= -a(\tau)+\sum_{k=0}^n|\vy_{t,k}(\tau)|-1+\lambda_1(\tau)$},&\\ 
& &{\scriptsize $\lambda_1(\tau)=\text{ReLU}(a(\tau))$}&\\
    \hline
  \multirow{3}{*}{ \begin{minipage}{.05\textwidth}
   \includegraphics[width=1.3cm, trim=7.0cm 12.3cm 8.5cm 12.0cm,clip]{picture_results_new/B1plusplot.pdf}
   \end{minipage}}  & $\Pcal=\mathcal{B}_{1,+}$ & {\scriptsize $\vy_{t,i}(\tau)= \text{ReLU}\left(\frac{u_i(\tau)}{\lambda_{SM}(1 - \beta){\Gamma_Y}_{ii}(t){D_{2,ii}(t)}}- \lambda_1(\tau) \right)$} & \multirow{3}{*}{ \begin{minipage}{.125\textwidth}
   \includegraphics[width=1.6cm, trim=11cm 5.0cm 18.5cm 3.5cm,clip]{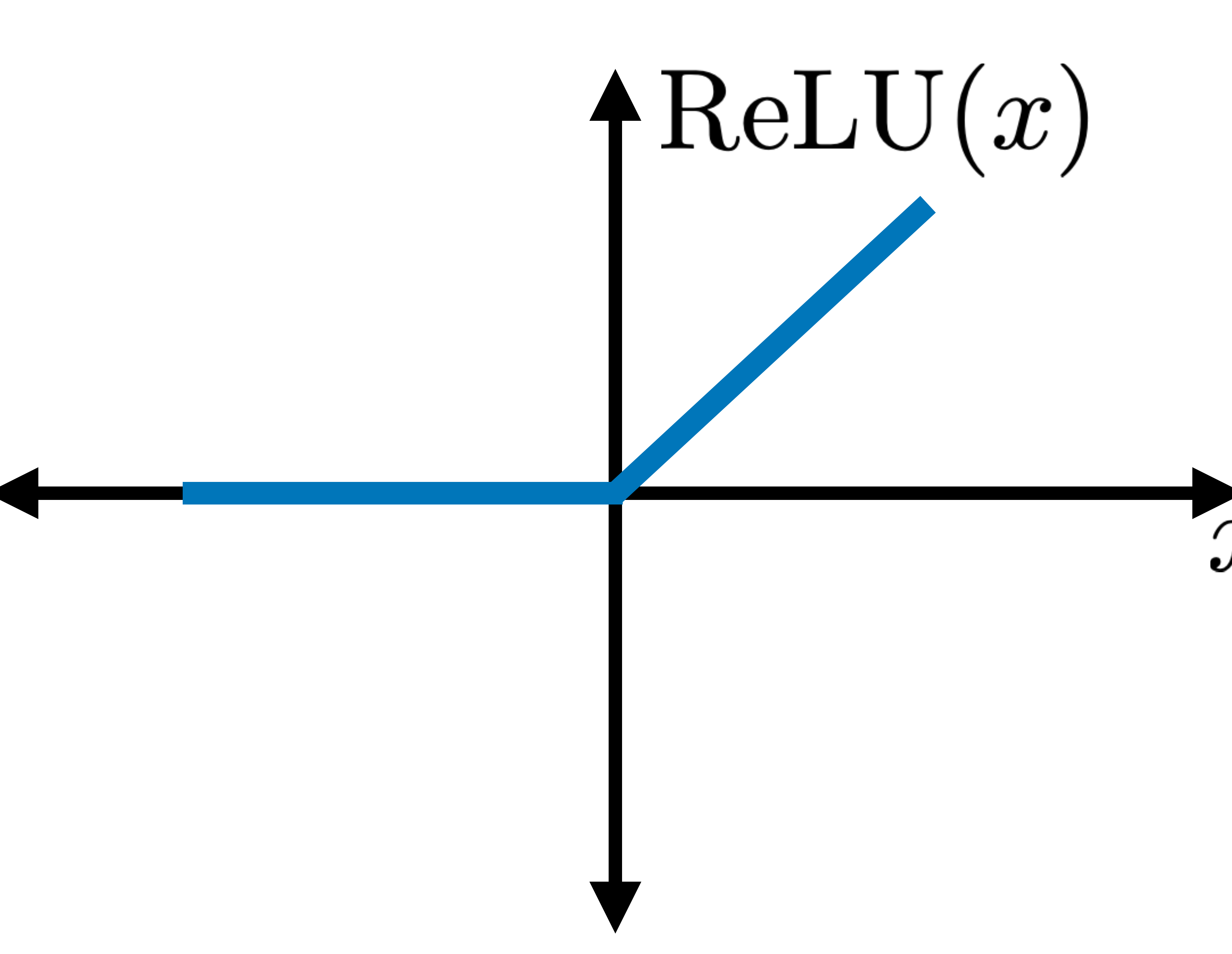}
   \end{minipage}}\\
   &  Nonnegative & {\scriptsize $\frac{da(\tau)}{d\tau}= -a(\tau)+\sum_{k=0}^n\vy_{t,k}(\tau)-1+\lambda_1(\tau)$}, & \\
   &  Sparse &  {\scriptsize $\lambda_1(\tau)=\text{ReLU}(a(\tau))$}&\\
    \hline 
 \multirow{3}{*}{ \begin{minipage}{.05\textwidth}
   \includegraphics[width=1.3cm, trim=7.0cm 12.3cm 8.5cm 12.0cm,clip]{picture_results_new/UnitSimplexplot.pdf}
   \end{minipage}} &  $\Pcal=\Delta$ & {\scriptsize $\vy_{t,i}(\tau)= \text{ReLU}\left(\frac{u_i(\tau)}{\lambda_{SM}(1 - \beta){\Gamma_Y}_{ii}(t){D_{2,ii}(t)}}- \lambda_1(\tau) \right)$} & \multirow{3}{*}{ \begin{minipage}{.125\textwidth}
   \includegraphics[width=1.6cm, trim=11cm 5.0cm 18.5cm 3.5cm,clip]{picture_results_new/relu.pdf}
   \end{minipage}} \\
   & Unit & {\scriptsize $\frac{d\lambda_1(\tau)}{d\tau}= -\lambda_1(\tau)+\sum_{k=0}^n\vy_{t,k}(\tau)-1+\lambda_1(\tau)$} & \\
   & Simplex & & \\
    \hline
    \end{tabular}
\end{table}

We can make the following observations on Table \ref{tab:nnoutdynamicsexamples}: (1) For sparse and unit simplex settings, there is an additional inhibitory neuron which takes input from all outputs and  whose activation is the inhibitory signal $\lambda_1(\tau)$, (2) The source attributes, which are globally defined over all sources, determine the activation functions at the output layer.  The proposed framework can be applied to any polytope described by (\ref{eq:polygeneral}) for which  the corresponding Det-Max neural network will contain combinations of activation functions in Table \ref{tab:nnoutdynamicsexamples} as illustrated in Figure \ref{fig:GenericNN}.(a).

\section{Numerical experiments}
\label{sec:numexp}
\looseness = -1 In this section, we  illustrate the applications of the proposed WSM-based BSS framework for both synthetic and natural sources. More details on these experiments and additional examples are provided in Appendix \ref{sec:numexpappendix}, including sparse dictionary learning. Our implementation  code is publicly available\footnote{ \href{https://github.com/BariscanBozkurt/Biologically-Plausible-DetMaxNNs-for-Blind-Source-Separation}{\footnotesize https://github.com/BariscanBozkurt/Biologically-Plausible-DetMaxNNs-for-Blind-Source-Separation}}.
\subsection{Synthetically correlated source separation}
\label{sec:numexpcorrcopula}

In order to illustrate the correlated source separation capability of the proposed WSM neural networks, we consider a numerical experiment with five copula-T distributed (uniform and correlated) sources.
For the correlation calibration matrix for these sources, we use Toeplitz matrix whose first row is $\begin{bmatrix}1 & \rho & \rho & \rho & \rho\end{bmatrix}$. The $\rho$ parameter determines the correlation level, and we considered the range $\left[0,0.6\right]$ for this parameter. These sources are mixed with a $10\times5$ random matrix with independent and identically distributed (i.i.d.) standard normal random variables.  
The mixtures are corrupted by i.i.d. normal noise corresponding to $30$dB signal-to-noise ratio (SNR) level. In this experiment, we employ the nonnegative-antisparse-WSM neural network (Figure \ref{fig:NNBinftyplus} in Appendix \ref{appsec:nnantisparse}) whose activation functions at the output layer are nonnegative-clipping functions, as the sources are nonnegative uniform random variables.

\begin{wrapfigure}{l}{0.55\textwidth}
\begin{center}
\includegraphics[trim = {1.3cm 0.0cm 1.9cm 1.0cm},clip,width=0.43\textwidth]{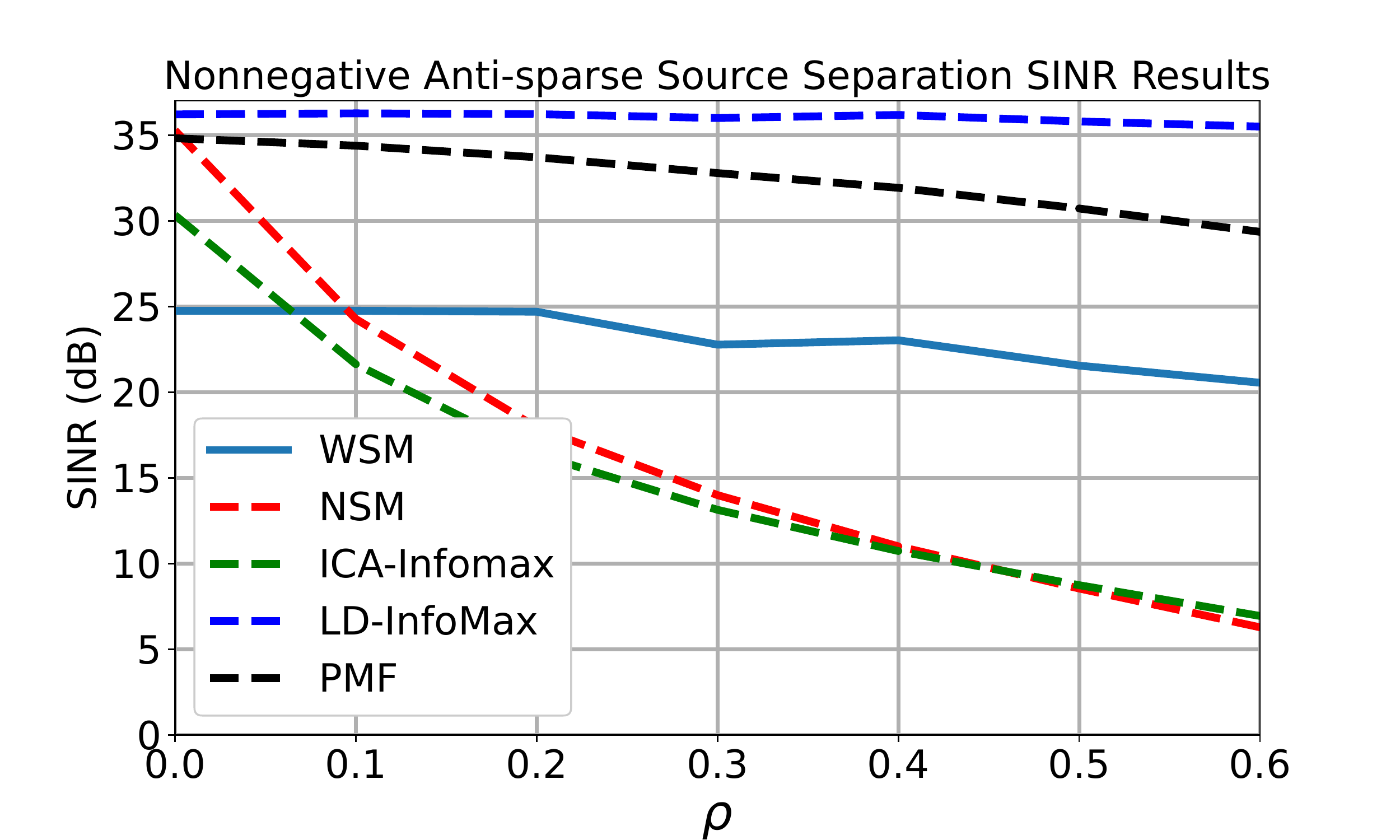}
\end{center}
\caption{The SINRs of WSM, NSM, ICA, PMF, and LD-InfoMax versus the correlation parameter $\rho$.}
\label{fig:Copulafig}
\end{wrapfigure}
We compared the signal-to-interference-plus-noise-power-ratio (SINR) performance of our algorithm with the NSM algorithm \citep{pehlevan2017blind}, Infomax ICA algorithm \citep{bell1995information}, as implemented in Python MNE Toolbox \citep{GramfortEtAl2013a}, LD-InfoMax algorithm \cite{erdogan2022information}, and PMF algorithm \cite{tatli2021tspsubmitted}. Figure \ref{fig:Copulafig} shows the SINR performances of these algorithms (averaged over 300 realizations) as a function of the correlation parameter $\rho$. We observe that our WSM-based network performs well despite correlations. In contrast, performance of NSM and ICA algorithms, which assume uncorrelated sources, degrade noticeably with increasing correlation levels. In addition, we note that the performance of batch Det-Max algorithms, i.e., LD-InfoMax and PMF, are also robust against source correlations. Furthermore, due to their batch nature, these algorithms typically achieved better performance results than our neural network with online-restriction, as expected.

\subsection{Image separation}
\label{sec:numexpimageseparation}
To further illustrate the correlated source separation advantage of our approach, 
we consider a natural image separation scenario. For this example, we have three RGB images with sizes $324 \times 432 \times 3$ as sources (Figure \ref{fig:imsources}). The sample Pearson correlation coefficients between the images are $\rho_{12} = 0.263$, $\rho_{13} = 0.066$, $\rho_{23} = 0.333$. We use a random $5 \times 3$ mixing matrix whose entries are drawn from i.i.d standard normal distribution. The corresponding mixtures are shown in Figure \ref{fig:immixtures}.


\begin{figure} [H]
\centering
\begin{tabular}{cccc}
\includegraphics[trim = {2cm 0cm 2cm 0cm},clip,width=0.33\textwidth]{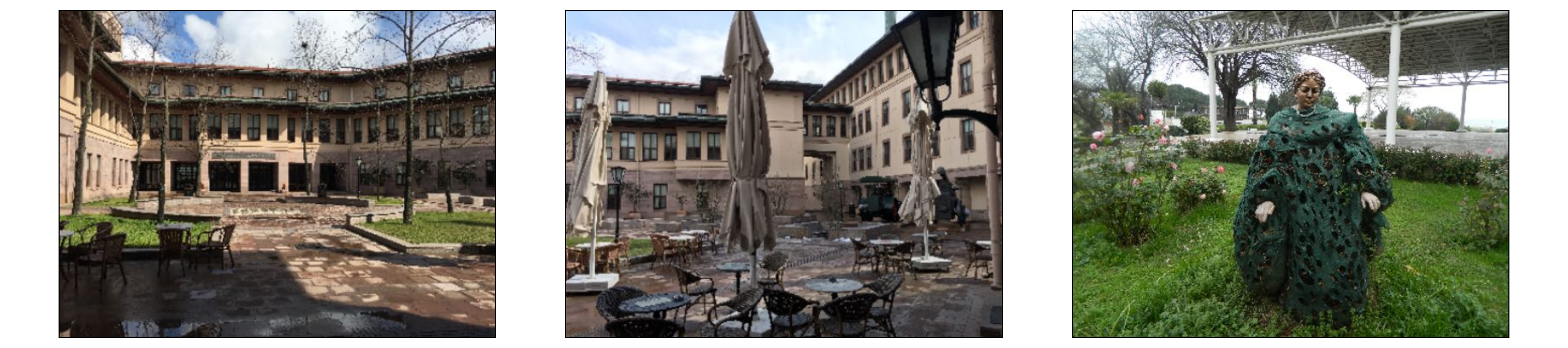} &
\includegraphics[trim = {1.2cm 0.6cm 1.2cm 0cm},clip,width=0.55\textwidth]{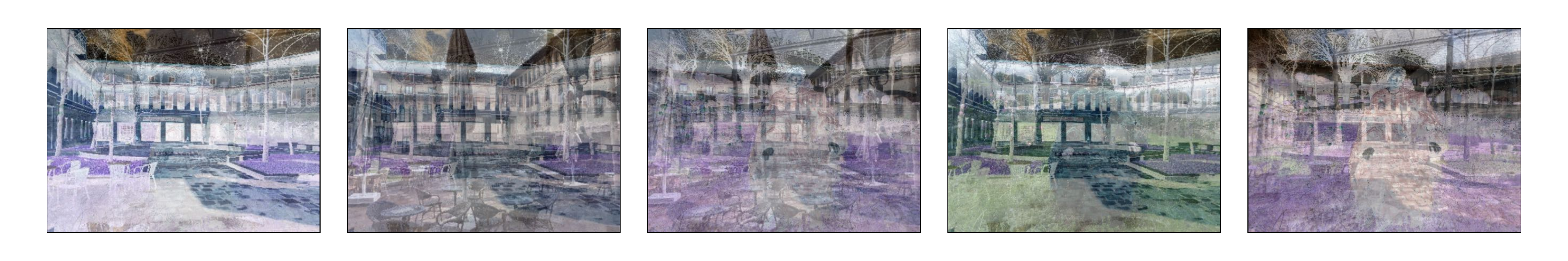}\\
\textbf{(a)}  & \textbf{(b)}  \\[6pt]
\end{tabular}
\begin{tabular}{cccc}
\includegraphics[trim = {2cm 0cm 2cm 0cm},clip,width=0.3\textwidth]{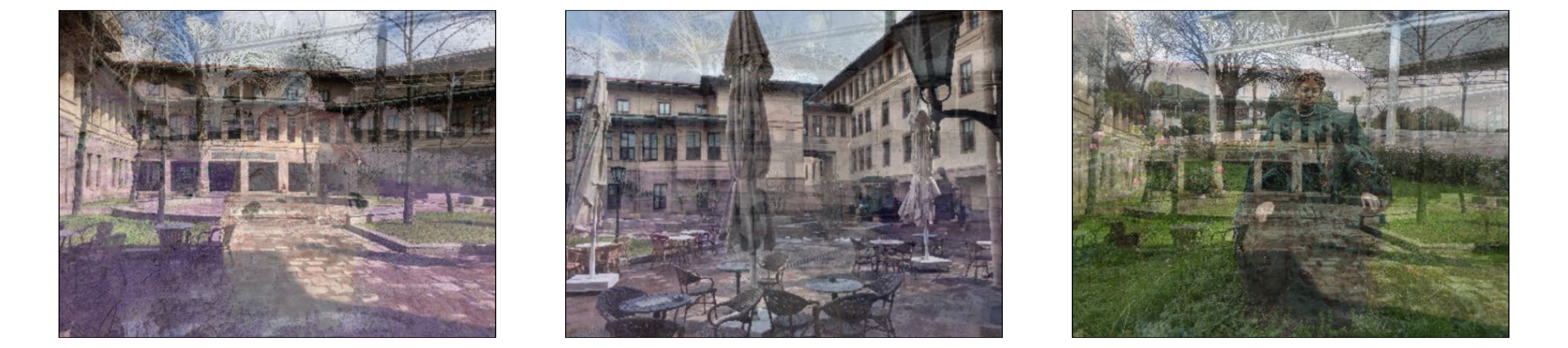} &
\includegraphics[trim = {2cm 0cm 2cm 0cm},clip,width=0.3\textwidth]{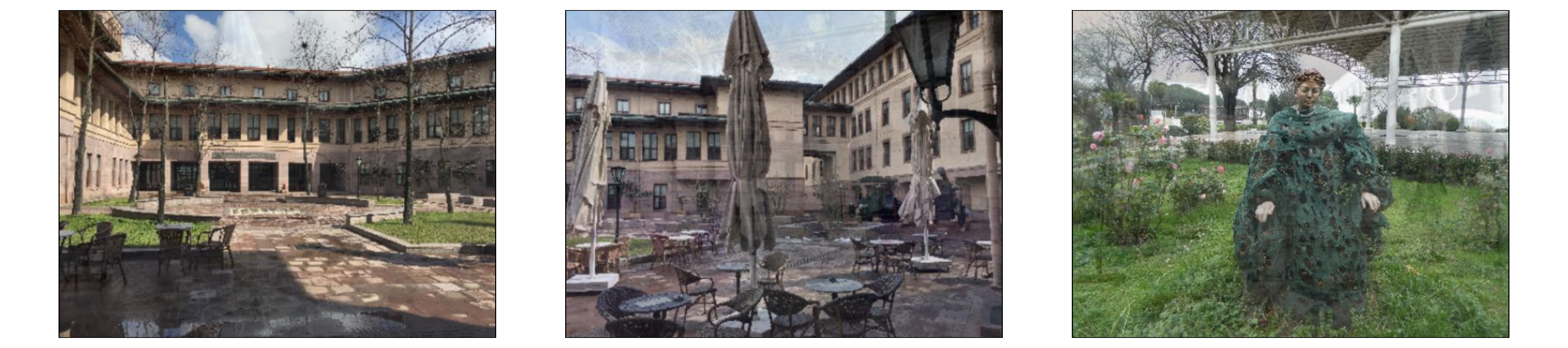} &
\includegraphics[trim = {2cm 0cm 2cm 0cm},clip,width=0.3\textwidth]{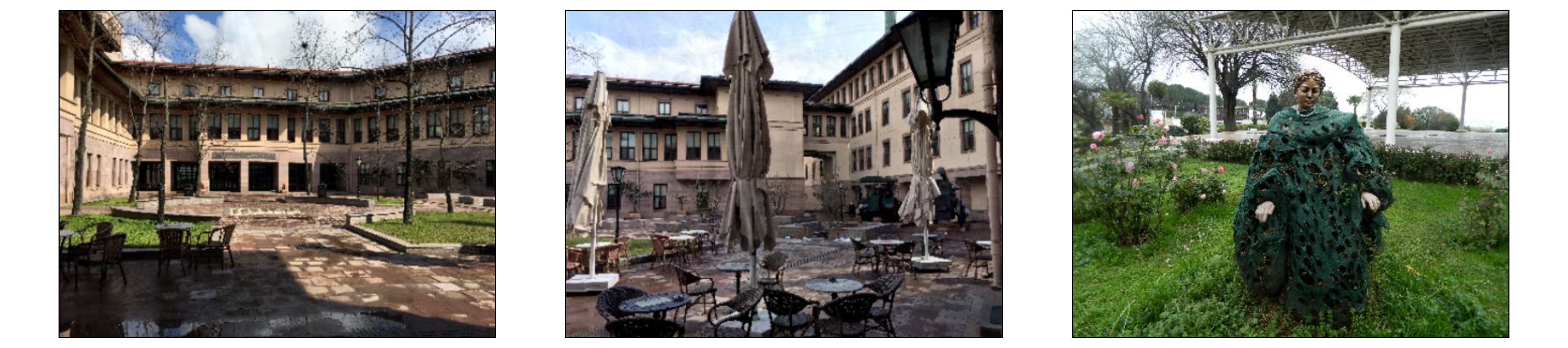} \\
\textbf{(c)} & \textbf{(d)}  & \textbf{(e)}  \\[6pt]
\end{tabular}
\caption{ \textbf{(a)} Original RGB images,
\textbf{(b)} Mixture RGB images,
\textbf{(c)} ICA outputs,
\textbf{(d)} NSM outputs (using pre-whitened mixtures),
\textbf{(e)} WSM outputs.}
\label{fig:imageseparation}
\extralabel{fig:imsources}{.(a)}
\extralabel{fig:immixtures}{.(b)}
\extralabel{fig:imoutputs}{.(c),(d),(e)}
\end{figure}

We  applied ICA, NSM and WSM algorithms to the mixtures. Figure \ref{fig:imoutputs} shows the corresponding outputs. High-resolution versions of all  images in this example are available in Appendix \ref{sec:imageseparaionappendix} in addition to the comparisons with LD-Infomax and PMF algorithms. The Infomax ICA algorithm's outputs have SINR level of  $13.92$dB, and this performance is perceivable as residual interference effects in the corresponding output images. The NSM algorithm achieves significantly higher SINR level of $17.45$dB and the output images visually reflect this better performance. Our algorithm achieves the best SINR level of $27.49$dB, and the corresponding outputs closely resemble the original source images.

\section{Discussion and Conclusion}
\label{sec:conclusion}

We proposed a general framework for generating biologically plausible neural networks that are capable of separating correlated sources from their linear mixtures, and demonstrated their successful correlated source separation capability through synthetic and natural sources. 

Another motivation for our work is to link network structure with function. This is a long standing goal of neuroscience, however examples where this link can be achieved are limited. Our work provides concrete examples where clear links between a network's architecture--i.e. number of interneurons, connections between interneurons and output neurons, nonlinearities (frequency-current curves)-- and its function, the type of source separation or feature extraction problem the networks solves, can be established. These links may provide insights and interpretations that might generalize to real biological circuits.

Our networks suffer from the same limitations of other recurrent biologically-plausible BSS networks. First, certain hyperparameters can significantly influence algorithm performance (see Appendix \ref{sec:hyperparameterablationappendix}). Especially, the inner product gains ($D_{ii}$) are sensitive to the combined choices of algorithm parameters, which require careful tuning. Second, the numerical experiments with our neural networks are relatively slow due to the recursive computations in  (\ref{eq:descv})-(\ref{eq:htnn}) and (\ref{eq:descu})-(\ref{eq:ytnn}) for hidden layer and output vectors, which is common to all biologically plausible recurrent source separation networks (see Appendix \ref{sec:complexityappendix}). This could perhaps be addressed by early-stopping the recursive computation \cite{minden2018biologically}.

\begin{ack}
This work/research was supported by KUIS AI Center Research Award. C. Pehlevan acknowledges support from the Intel Corporation.
\end{ack}
\bibliography{wsmbibfile}
\newpage

\section*{Checklist}


\begin{enumerate}

\item For all authors...
\begin{enumerate}
  \item Do the main claims made in the abstract and introduction accurately reflect the paper's contributions and scope?
    \answerYes{We noticeably express the  contributions and scope of the paper, i.e., providing a general framework for constructing biologically plausible neural networks for separating both independent and correlated sources while covering more generic source domains. }
  \item Did you describe the limitations of your work?
    \answerYes{See Section \ref{sec:conclusion}}
  \item Did you discuss any potential negative societal impacts of your work?
    \answerNA{We strongly believe that this work does not have any potential negative societal impacts.}
  \item Have you read the ethics review guidelines and ensured that your paper conforms to them?
    \answerYes{We have read the ethics review guidelines and our article conforms to them.}
\end{enumerate}

\item If you are including theoretical results...
\begin{enumerate}
  \item Did you state the full set of assumptions of all theoretical results?
    \answerYes{Section \ref{sec:BSSsetting} provides the underlying assumptions for our theoretical results.}
        \item Did you include complete proofs of all theoretical results?
    \answerYes{We provide the complete proof of Theorem \ref{th:1} in the appendix. Moreover, we include the detailed derivations of the proposed network dynamics in the appendix.}
\end{enumerate}

\item If you ran experiments...
\begin{enumerate}
  \item Did you include the code, data, and instructions needed to reproduce the main experimental results (either in the supplemental material or as a URL)?
    \answerYes{Yes, we include the code in the supplemental material with tutorial Python notebooks and multiple-run simulation scripts for our experiments in addition to README.md file.}
  \item Did you specify all the training details (e.g., data splits, hyperparameters, how they were chosen)?
    \answerYes{In the appendix, we provide the hyperparameter selections and training details of the experiments. We also share our code for numerical experiments.}
        \item Did you report error bars (e.g., with respect to the random seed after running experiments multiple times)?
    \answerYes{In the appendix, we provide the SINR convergence behaviors for several experiments based on multiple realizations with corresponding percentile envelopes. }
        \item Did you include the total amount of compute and the type of resources used (e.g., type of GPUs, internal cluster, or cloud provider)?
    \answerNo{We did not discuss the computational power and memory concerns since our method does not require an advanced computing system. An individual source separation experiment can run on a basic computer.}
\end{enumerate}

\item If you are using existing assets (e.g., code, data, models) or curating/releasing new assets...
\begin{enumerate}
  \item If your work uses existing assets, did you cite the creators?
    \answerYes{We refer to \citep{olshausen1997sparse} and link to the data in \ref{sec:sparsedictlearningappendix} for image patches we used for sparse dictionary learning.}
  \item Did you mention the license of the assets?
    \answerNo{We were unable to find the official license of the image patches we used for sparse dictionary learning. However, we cite the creators of the data \citep{olshausen1997sparse} and give the link for the image patches in \ref{sec:sparsedictlearningappendix}.} 
  \item Did you include any new assets either in the supplemental material or as a URL?
    \answerYes{Our code is the only new asset which we include in the supplemental material.}
  \item Did you discuss whether and how consent was obtained from people whose data you're using/curating?
    \answerNA{}
  \item Did you discuss whether the data you are using/curating contains personally identifiable information or offensive content?
    \answerNA{}
\end{enumerate}

\item If you used crowdsourcing or conducted research with human subjects...
\begin{enumerate}
  \item Did you include the full text of instructions given to participants and screenshots, if applicable?
    \answerNA{}
  \item Did you describe any potential participant risks, with links to Institutional Review Board (IRB) approvals, if applicable?
    \answerNA{}
  \item Did you include the estimated hourly wage paid to participants and the total amount spent on participant compensation?
    \answerNA{}
\end{enumerate}

\end{enumerate}


\pagebreak

\appendix

\section*{Appendix}
\renewcommand{\theequation}{A.\arabic{equation}}
\setcounter{equation}{0}

\section{On sufficient scattering condition for source vectors}
\label{appsec:suffscat}
The determinant maximization criterion used in matrix factorization frameworks is based on the assumption that the latent factors are sufficiently scattered in their presumed domain to somewhat reflect its shape. Both simplex structure matrix factorization and polytopic matrix factorization frameworks propose precise sufficient scattering conditions for the latent vectors to guarantee their identifiability under the  determinant maximization criterion. In this section, we briefly summarize these conditions.

\begin{figure}[h!]%
    \centering
    \subfloat[The unit simplex $\Delta$  and the second order cone $\mathcal{C}$.  ]{{\includegraphics[width=2.5cm, trim=0.0cm 0.0cm 0.0cm 0.0cm,clip]{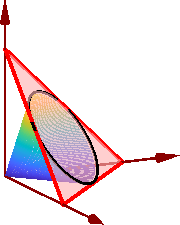} }}%
    \hspace{0.25in}
    \subfloat[ An illustration of the sufficient scattering condition for $\Delta$ samples.]{{\includegraphics[width=3.5cm, trim={2cm 0.0cm 2cm 1.1cm},clip]{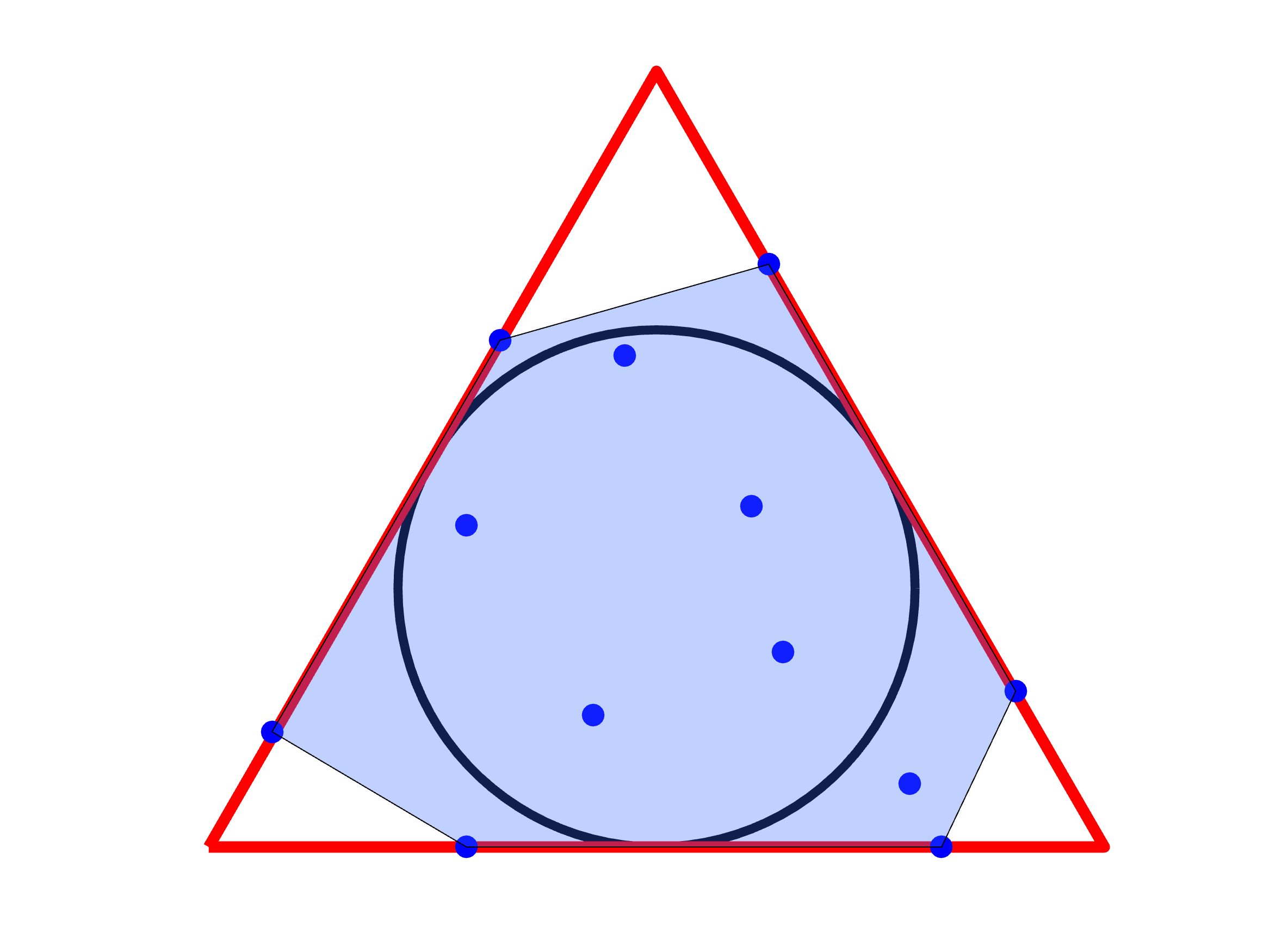} }}
    \hspace{0.25in}
        \subfloat[The polytope $\mathcal{B}_{1,+}$  and its MVIE $\mathcal{E}_{\mathcal{B}_1}$.  ]{{\includegraphics[width=2.2cm, trim=0.0cm 0.0cm 0.0cm 0.0cm,clip]{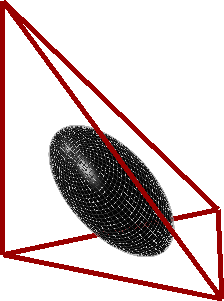} }}%
    \hspace{0.25in}
    \subfloat[ An illustration of the sufficient scattering condition for $\mathcal{B}_{1,+}$ samples.]{{\includegraphics[width=2.0cm, trim={0cm 0.0cm 0cm 0.0cm},clip]{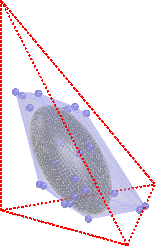} }}
    \caption{The geometry of sufficient scattering conditions for the unit simplex and polytopes illustrated in three-dimensions.}
    \label{fig:suffscat}
\end{figure}

\underline{\it Sufficient scattering condition for unit simplex sources:} 
An earlier latent factor identfiability assumption for SSMF required the inclusion of the vertices of the unit simplex in the generative latent vector samples \citep{donoho2003does}. This, so-called {\it separability} or {\it local dominance}, assumption was later replaced by a weaker {\it sufficiently scattered condition} (SSC) in \cite{fu2018identifiability}. This new condition uses the second order cone
\begin{eqnarray*}
\mathcal{C}=\{\mathbf{s} \hspace{0.02in} \vert \hspace{0.02in} \mathbf{1}^T\mathbf{s}\ge \sqrt{n-1}\|\mathbf{s}\|_2, \mathbf{s}\in\mathbb{R}^n\},
\end{eqnarray*}
which is illustrated in Figure \ref{fig:suffscat}.(a) together with the unit simplex $\Delta$, as a reference object for defining SSC. The SSC proposed in \cite{fu2018identifiability} for SSMF requires that conic hull of the simplex samples contains $\mathcal{C}$, i.e.,
\begin{eqnarray*}
\text{cone}(\{\vs_1, \vs_2, \ldots, \vs_t\})\supseteq \mathcal{C}.
\end{eqnarray*}
Let $\mathcal{A}_\Delta$ represent the affine hull of $\Delta$. Figure \ref{fig:suffscat}.(b) illustrates this requirement restricted to $\mathcal{A}_\Delta$: the red triangle is the boundary of $\Delta$,  the blue dots are sufficiently scattered samples from $\Delta$, the black circle and the blue polyhedral region are the boundary of $\mathcal{C}$ and the conic hull of sufficiently scattered samples from $\Delta$ restricted to $\mathcal{A}_\Delta$,   respectively. There is an additional requirement that the boundaries of  $\Delta$ and  $\text{cone}(\{\vs_1, \vs_2, \ldots, \vs_t\})\cap\mathcal{A}_\Delta$  intersect the boundary of $\mathcal{C}\cap \mathcal{A}_\Delta$ at the identical points.

\underline{\it Sufficient scattering condition for polytopic sources:} The reference \cite{tatli2021tspsubmitted} offers a similar SSC for polytopic sources for which the reference object for SSC is the maximum volume inscribed ellipsoid (MVIE), represented by $\mathcal{E}_\Pcal$ of the polytope $\Pcal$. Figure \ref{fig:suffscat}.(c) illustrates MVIE (the black ellipsoid) for the polytope selection $\Pcal=\mathcal{B}_{1,+}$ whose edges are the red lines. The SSC for polytopic sources require that convex hull of the polytopic samples contain the polytope's MVIE,i.e., 
\begin{eqnarray*}
\text{conv}(\{\vs_1, \vs_2, \ldots, \vs_t\})\supseteq \mathcal{E}_\Pcal.
\end{eqnarray*}
This condition is illustrated in Figure \ref{fig:suffscat}.(d), where the dots represent sufficiently scattered samples and the blue polyhedral region is their convex hull. The SSC in \cite{tatli2021tspsubmitted} further require that the boundaries of $\Pcal$ and $\text{conv}(\{\vs_1, \vs_2, \ldots, \vs_t\})$ intersect $\mathcal{E}_\Pcal$ at the identical points.


\section{Proof of theorem 1}
\label{app:thm1proof}
The proof of Theorem \ref{th:1} relies on the following lemma, which follows from equality constraints (\ref{eq:wsm3constr1}) and (\ref{eq:wsm3constr2}):

\begin{lemma}
Given the mixing model in Section \ref{sec:mixingmodel}, for sufficiently large sample sizes enabling full column-rank condition on  $\vX(t)$, the constraints in (\ref{eq:wsm3constr1}) and (\ref{eq:wsm3constr2}) define an arbitrary linear mapping between input and output vectors in the form
\begin{align}
\vy_i={\vW}(t) \vx_i, \qquad i \in \{1, \ldots, t \}.
\end{align}
where ${\vW}(t)$ is full-rank.\label{lemma:1}
\end{lemma}

\begin{proof}[Proof of Lemma \ref{lemma:1}]
 To see the relation between $\vh_i$ and $\vx_i$, note that mixing relation \eqref{eq:mixing} and equality constraint \eqref{eq:wsm3constr1} enforces $\vS(t)^T\vA^T\vA\vS(t) =\vH(t)^T\vD_1(t)\vH(t)$.
	Defining $\vA=\vU_\vA\bSig_\vA\vV_\vA^T$ as the reduced SVD decomposition for $\vA$ matrix with $\vU_\vA \in \mathbb{R}^{m \times n},\bSig_\vA\in\mathbb{R}^{n \times n}, \vV_\vA\in\mathbb{R}^{n \times n}$, we can write 
	$\vA^T\vA=\vV_\vA\bSig_\vA^2\vV_\vA^T$.
	For sufficiently large sample sizes enabling full rank $\vX(t)$, these imply $\vh_i=\vD_1(t)^{-1/2}\vQ^T_1(t)\bSig_\vA\vV_\vA^T \vs_i$, $i \in \{1, \ldots, t \}$,
for some real-orthogonal matrix $\vQ_1(t)$. From this expression, we can also write $\vh_i=\vD_1(t)^{-1/2}\vQ^T_1(t)\vU_\vA^T \vx_i = \vD_1(t)^{-1/2}\bTheta_1^T(t) \vx_i$, $i \in \{1, \ldots, t \}$, where $\bTheta_1(t)=\vU_\vA\vQ_1(t) \in\mathbb{R}^{m\times n}$ is a matrix with orthonormal columns.

The weighted inner product matching condition in \eqref{eq:wsm3constr2} implies that the output $\vy_i$'s  are related to the the slack vectors $\vh_i$'s through the relationship $\vy_i=\vD_2(t)^{-1/2}(t)\bTheta_2(t)  \vh_i$, $ i \in \{1, \ldots, t \}$,
where $\bTheta_2(t)$ is another real-orthogonal matrix. Consequently $
\vy_i= \vD_2(t)^{-1/2}\bTheta_2(t) \vD_1(t)^{-1/2} \bTheta^T_1(t) \vx_i$, $i \in \{1, \ldots, t \}$. Here, the multiplier of $\vx_i$ is in the form of a Singular Value Decomposition of a full rank matrix, i.e., $\bTheta_2(t) \vD_1(t)^{-1/2} \bTheta^T_1(t)$ which is left multiplied by a full rank diagonal matrix, i.e., $\vD_2(t)^{-1/2}$ . This implies that the equality constraints (\ref{eq:wsm3constr1}) and (\ref{eq:wsm3constr2}) in conjunction define an arbitrary linear mapping between input and output vectors, through inner product matching. 
%
%
\end{proof}
Now we can prove Theorem \ref{th:1}.
\begin{proof}[Proof of Theorem \ref{th:1}]
 By Lemma \ref{lemma:1},
$\vy_i=\vW(t) \vx_i$, $i=1,\ldots t$, where $\vW(t)$ admits the form $\vW(t)=\vD_2(t)^{-1/2}\bTheta_2(t) \vD_1(t)^{-1/2} \bTheta^T_1(t)$.
Therefore, using the mixture-source relationship $\vx_i=\vA\vs_i$, we can write $\vy_i=\vG(t)\vs_i$,
where $\vG(t)=\vW(t)\vA$ is the linear mapping relationship between outputs and sources. Plugging this into maximization objective in (\ref{eq:bssobjective}), we obtain
$\log(\det(\vY(t)\vY(t)^T))
=2\log(|\det(\vG(t))|)+\log(\det(\vS(t)\vS(t)^T))$.
 To proceed, we define $\vA=\vU_\vA\bSig_\vA\vV_\vA^T$ as the reduced SVD for  matrix $\vA$ to obtain
$\vG(t)=\vD_2(t)^{-1/2}\bTheta_2(t) \vD_1(t)^{-1/2}\bTheta_1(t)^T\vU_\vA\bSig_\vA\vV_\vA^T$. 
Consequently,
$2\log(|\det(\vG(t))|)=-\log(\det(\vD_1))-\log(\det(\vD_2))+2\log(\det(\bSig_\vA))$.
 As a result, maximizing the objective in (\ref{eq:bssobjective}) is equivalent to minimizing $\log(\det(\vD_1))+\log(\det(\vD_2))$ with some additional constant terms. Since $\vD_1$ and $\vD_2$ are diagonal, the equivalent function can be written as 
$\sum_{i=1}^n \log(D_{1,ii}(t))+\sum_{i=1}^n \log(D_{2,ii}(t))$, which is the objective function in (\ref{eq:wsm3objective}).
\end{proof}


\section{Derivations}

\subsection{The simplification of the similarity matching cost functions}
\label{sec:smcsimplify}
In this section, we provide the simplification of the similarity matching cost functions $J_1$ and $J_2$ in Section \ref{sec:onlinewsm}, by preserving only the quadratic terms that are relevant to online optimization with respect to $\vh_t$ and $\vy_t$.

	Using the matrix partitions $\bwX(t)=\left[\begin{array}{cc} \gamma \bwX(t-1) & \vx_t\end{array}\right]$ and $\bwH(t)=\left[\begin{array}{cc} \gamma \bwH(t-1) & \vh_t\end{array}\right]$, we can write $J_1(\vH(t),\vD_1(t))$  more explicitly as
	\begin{eqnarray*}
		&&	J_{1}(\vH(t),\vD_1(t)) \nonumber \\
		&&={\scriptsize \frac{\gamma^2}{2\tau^2}\left\|\left[\begin{array}{cc} \gamma \bwX(t-1)^T\bwX(t-1)- \gamma \bwH(t-1)^T\vD_1(t) \bwH(t-1)& \bwX(t-1)^T\vx_{t}-\bwH(t-1)^T\vD_1(t) \vh_{t} \\
			\vx_{t}^T\bwX(t-1)-\vh_{t}^T\vD_1(t) \bwH(t-1)  & \frac{ \|\vx_{t}\||_2^2-\|\vh_{t}\|_{\vD_1(t)}^2}{\gamma}   \end{array}\right]\right\|_{F}^2}\\
		&&=\frac{\gamma^4}{\tau^2} \|\bwX(t-1)^T\bwX(t-1)\|_{F}^2 +\frac{\gamma^4}{\tau^2}\|\bwH(t-1)^T\vD_1(t) \bwH(t-1) \|_{F}^2 \\
		&& -\frac{2\gamma^4}{\tau^2} \text{Tr}(\bwH(t-1)^T\vD_1(t) \bwH(t-1)\bwX(t-1)^T\bwX(t-1)) \\
		&&+ \frac{2\gamma^2}{\tau^2} \|\vx_{t}^T\bwX(t-1)\|_{F}^2+\frac{2\gamma^2}{\tau^2} \|\vh_{t}^T\vD_1(t) \bwH(t-1)\|_{F}^2\\
		&&-\frac{4\gamma^2}{\tau^2} \vh_{t}^T\vD_1(t) \bwH(t-1) \bwX(t-1)^T\vx_{t}+\frac{1}{\tau^2} ( \|\vx_{t}\|^4+\|\vh_{t}\|_{\vD_1(t)}^4-2 \|\vx_{t}\|^2\|\vh_{t}\|_{\vD_1(t)}^2 ).
	\end{eqnarray*}
By keeping only the relevant part of this cost function for online optimization with respect to $\vh_t$,  by scaling with $\tau/\gamma^2$, and ignoring the small final term,  we obtain the effective 
online cost function corresponding to $J_1$ as 
\begin{eqnarray*}
	c_1(\vh_t)&=&2\vh_t^T\vD_1\vM_H(t)\vD_1(t)\vh_t-4\vh_t^T\vD_1(t)\vW_{HX}(t)\vx_t,
\end{eqnarray*}
where
\begin{eqnarray}
\vM_H(t)=\frac{1}{\tau} \bwH(t-1)\bwH(t-1)^T=\frac{1}{\tau}\sum_{k=1}^{t-1}(\gamma^2)^{t-1-k}\vh_k\vh_{k}^T \label{eq:MH}\\ 
\vW_{HX}(t)=\frac{1}{\tau} \bwH(t-1)\bwX(t-1)^T=\frac{1}{\tau}\sum_{k=1}^{t-1}(\gamma^2)^{t-1-k}\vh_k\vx_{k}^T. \label{eq:WHX}
\end{eqnarray}
If we apply the same procedure to $J_2(\vH(t),\vD_2(t),\vY(t))$:
	\begin{eqnarray*}
		&&	J_{2}(\vH(t),\vD_2(t), \vY(t)) \nonumber \\
		&&={\scriptsize \frac{\gamma^2}{2\tau^2}\left\|\left[\begin{array}{cc} \gamma \bwH(t-1)^T\bwH(t-1)- \gamma \bwY(t-1)^T\vD_2(t) \bwY(t-1)& \bwH(t-1)^T\vh_{t}-\bwY(t-1)^T\vD_2(t) \vy_{t} \\
			\vh_{t}^T\bwH(t-1)-\vy_{t}^T\vD_2(t) \bwY(t-1)  & \frac{ \|\vh_{t}\||_2^2-\|\vy_{t}\|_{\vD_2(t)}^2}{\gamma}   \end{array}\right]\right\|_{F}^2}\\
		&&=\frac{\gamma^4}{\tau^2} \|\bwH(t-1)^T\bwH(t-1)\|_{F}^2 +\frac{\gamma^4}{\tau^2}\|\bwY(t-1)^T\vD_2(t) \bwY(t-1) \|_{F}^2 \\
		&& -\frac{2\gamma^4}{\tau^2} \text{Tr}(\bwY(t-1)^T\vD_2(t) \bwY(t-1)\bwH(t-1)^T\bwH(t-1)) \\
		&&+ \frac{2\gamma^2}{\tau^2} \|\vh_{t}^T\bwH(t-1)\|_{F}^2+\frac{2\gamma^2}{\tau^2} \|\vy_{t}^T\vD_2(t) \bwY(t-1)\|_{F}^2\\
		&&-\frac{4\gamma^2}{\tau^2} \vy_{t}^T\vD_2(t) \bwY(t-1) \bwH(t-1)^T\vh_{t}+\frac{1}{\tau^2} ( \|\vh_{t}\|^4+\|\vy_{t}\|_{\vD_2(t)}^4-2 \|\vh_{t}\|^2\|\vy_{t}\|_{\vD_2(t)}^2 ).
	\end{eqnarray*}
Similar to $J_1$, we can simplify the part of the $J_2$ cost function that is   dependent on $\vh_t$ and $\vy_t$  as
\begin{eqnarray*}
	c_2(\vh_t,\vy_t)&=&	2\vy_t^T\vD_2(t)\vM_Y(t)\vD_2(t)\vy_t-4\vy_t^T\vD_2(t)\vW_{YH}(t)\vh_t + 2\vh_t^T\vM_H(t)\vh_t,
\end{eqnarray*}
where
\begin{eqnarray}
\vW_{YH}(t)&=&\frac{1}{\tau} \bwY(t-1)\bwH(t-1)^T=\frac{1}{\tau}\sum_{k=1}^{t-1}(\gamma^2)^{t-1-k}\vy_k\vh_{k}^T, \label{eq:WYH}\\
\vM_Y(t)&=&\frac{1}{\tau} \bwY(t-1)\bwY(t-1)^T=\frac{1}{\tau}\sum_{k=1}^{t-1}(\gamma^2)^{t-1-k}\vy_k\vy_{k}^T. \label{eq:MY}
\end{eqnarray}
As a result, we can write the effective online cost function  $\mathcal{J}$, corresponding to $\vh_t$ and $\vy_t$  as
\begin{eqnarray}
C(\vh_t,\vy_t)=\beta c_1(\vh_t)+(1-\beta)c_2(\vh_t,\vy_t).\label{eq:Ccost}
\end{eqnarray}


\subsection{Derivatives of the WSM cost function}
\label{sec:onlinewsmgradients}
In this section, we provide the expressions for the gradients of the online 
WSM cost function $\mathcal{J}$ in (\ref{eq:onlinecost}) to be used in the descent algorithm formulation in Section \ref{linfdynamics} and Appendix \ref{sec:dersourcedomains}. For the gradients with respect to $\vh_t$ and $\vy_t$, we use  $C(\vh_t,\vy_t)$ in \eqref{eq:Ccost}, which is the simplified version of $\mathcal{J}$, as derived in Section \ref{sec:smcsimplify}.
\begin{itemize}
	\item The (scaled) gradient with respect to $\vh_t$:
	\begin{eqnarray}
	&&\hspace*{-2.2cm}	\frac{1}{4}\nabla_{\vh_t}\mathcal{J}(\vH(t),\vD_1(t),\vD_2(t),\vY(t))=\frac{1}{4}\nabla_{\vh_t}C(\vh_t,\vy_t) \nonumber\\
	&&\hspace*{2.9cm}=((1-\beta )\vM_H(t)+\beta\vD_1(t)\vM_H(t)\vD_1(t))\vh_t\nonumber \\
	&&\hspace*{2.9cm}-\beta \vD_1(t)\vW_{HX}(t)\vx_t-(1-\beta)\vW_{YH}(t)^T\vD_2(t)\vy_t. \label{eq:gradht1}
	\end{eqnarray}
	By applying the decomposition $\vM_H(t)=\bar{\vM}_H(t)+\bGam_H(t)$, where
	\begin{eqnarray}
	{\bGam_H}(t)&=&\text{diag}({\vM_H}_{11}(t), {\vM_H}_{22}(t), \ldots, {\vM_H}_{dd}(t)),\nonumber
	\end{eqnarray}
	we can rewrite the gradient expression in (\ref{eq:gradht1}) as
	\begin{eqnarray}
	\frac{1}{4}\nabla_{\vh_t}C(\vh_t,\vy_t) 
&&  =\vv_t+ ((1-\beta )\bar{\vM}_H(t)+\beta\vD_1(t)\bar{\vM}_H(t)\vD_1(t))\vh_t\nonumber \\
&&-\beta \vD_1(t)\vW_{HX}(t)\vx_t-(1-\beta)\vW_{YH}(t)^T\vD_2(t)\vy_t. \label{eq:gradht}
	\end{eqnarray}
	In (\ref{eq:gradht}), we used the substitution,
	\begin{eqnarray}
   \vv_t= ((1-\beta ){\bGam}_H(t)+\beta\vD_1(t){\bGam}_H(t)\vD_1(t)))\vh_t. \label{eq:v}
	\end{eqnarray}
	\item The gradient with respect to $\vy_t$:
	\begin{eqnarray}
		\frac{1}{4}\nabla_{\vy_t}C(\vh_t,\vy_t)
	= (1-\beta)(-\vD_2(t)\vW_{YH}(t)\vh_t+\vD_2(t)\vM_Y(t)\vD_2(t)\vy_t).\nonumber
\end{eqnarray}	
Note that, since $\vD_2(t)$ is positive,
\begin{eqnarray}
	-\frac{1}{4(1-\beta)}\vD_2(t)^{-1}\nabla_{\vy_t}\mathcal{J}(\vh_t,\vy_t)= \vW_{YH}(t)\vh_t-\vM_Y(t)\vD_2(t)\vy_t \label{eq:ytdescent} 
\end{eqnarray}
is  a descent direction. Furthermore, by decomposing $\vM_Y(t)=\bar{\vM}_Y(t)+\bGam_Y(t)$, where
\begin{eqnarray}
{\bGam_Y}(t)&=&\text{diag}({\vM_Y}_{11}(t), {\vM_Y}_{22}(t), \ldots, {\vM_Y}_{dd}(t)),\nonumber
\end{eqnarray}
we can rewrite the the descent direction in (\ref{eq:ytdescent}) as
\begin{eqnarray}
	-\frac{1}{4(1-\beta)}\vD_2(t)^{-1}\nabla_{\vy_t}C(\vh_t,\vy_t)=-\vu_t+ \vW_{YH}(t)\vh_t-\bar{\vM}_Y(t)\vD_2(t)\vy_t, \label{eq:gradyt}
\end{eqnarray}
where we substituted 
\begin{eqnarray}
\vu_t=\bGam_Y(t)\vD_2(t)\vy_t. \label{eq:ut}
\end{eqnarray}
\def\hshift{-0.0in}
\item  The derivative with respect to $D_{1,ii}(t)$:
\begin{align}
&\frac{\partial \mathcal{J}(\vh_t,\vy_t,\vD_1(t),\vD_2(t))}{\partial D_{1,ii}(t)}\nonumber \\&\hspace*{\hshift}=\lambda_{SM}\beta\mbox{Tr}((\bwH(t)^T\mathbf{E}_{ii}\bwH(t))^T(\bwH(t)^T\vD_1(t)\bwH(t)-\bwX(t)^T\bwX(t)))+\frac{1-\lambda_{SM}}{D_{1,ii}(t)}\MoveEqLeft[50]  \nonumber \\
 &\hspace*{\hshift}=\lambda_{SM} \beta\mbox{Tr}(\bwH(t)_{i,:}^T\bwH(t)_{i,:}(\bwH(t)^T\vD_1(t)\bwH(t)-\bwX(t)^T\bwX(t)))+\frac{1-\lambda_{SM}}{D_{1,ii}(t)} \nonumber\\
 &\hspace*{\hshift}= \lambda_{SM}\beta(\bwH(t)_{i,:}\bwH(t)^T\vD_1(t)\bwH(t)\bwH(t)_{i,:}^T-\bwH(t)_{i,:}\bwX(t)^T\bwX(t))   \bwH(t)_{i,:}^T) +\frac{1-\lambda_{SM}}{D_{1,ii}(t)}\nonumber \\
 &\hspace*{\hshift}= \lambda_{SM} \beta(\|{\vM_{H}}_{i,:}\|_{\vD_1(t)}^2-\|{\vW_{HX}}_{i,:}\|^2_2)+\frac{1-\lambda_{SM}}{D_{1,ii}(t)}.\label{eq:derD1ii}
\end{align}

\item  The derivative with respect to $D_{2,ii}(t)$:
\begin{align}
&\frac{\partial \mathcal{J}(\vh_t,\vy_t,\vD_1(t),\vD_2(t))}{\partial D_{2,ii}(t)} \nonumber \\&\hspace*{\hshift}=\lambda_{SM}(1 - \beta)\mbox{Tr}((\bwY(t)^T\mathbf{E}_{ii}\bwY(t))^T(\bwY(t)^T\vD_2(t)\bwY(t)-\bwH(t)^T\bwH(t)))+ \frac{1-\lambda_{SM}}{D_{2,ii}(t)}\MoveEqLeft[50]  \nonumber \\
&\hspace*{\hshift}=\lambda_{SM}(1 - \beta)\mbox{Tr}(\bwY(t)_{i,:}^T\bwY(t)_{i,:}(\bwY(t)^T\vD_2(t)\bwY(t)-\bwH(t)^T\bwH(t)))+ \frac{1-\lambda_{SM}}{D_{2,ii}(t)} \nonumber\\
&\hspace*{\hshift}= \lambda_{SM}(1 - \beta) (\bwY(t)_{i,:}\bwY(t)^T\vD_2(t)\bwY(t)\bwY(t)_{i,:}^T-\bwY(t)_{i,:}\bwH(t)^T\bwH(t))   \bwY(t)_{i,:}^T) + \frac{1-\lambda_{SM}}{D_{2,ii}(t)}\nonumber \\
&\hspace*{\hshift}=  \lambda_{SM}(1 - \beta) (\|{\vM_{Y}}_{i,:}\|_{\vD_2(t)}^2-\|{\vW_{YH}}_{i,:}\|^2_2)+\frac{1-\lambda_{SM}}{D_{2,ii}(t)}. \label{eq:derD2ii}
\end{align}

\end{itemize}


\section{Det-max WSM neural networks for example source domains}
\label{sec:dersourcedomains}
The proposed Det-Max WSM framework is applicable to infinitely many source domains corresponding to different assumptions on the sources. 
In this section, we provide derivations and illustrations of WSM-based Det-Max neural networks for some selected source domains.
\subsection{Anti-sparse sources}
\label{appsec:antisparse}
Section \ref{linfdynamics} covers the derivation of the network dynamics and the learning rules for antisparse sources, i.e., the source domain selection of $\Pcal=\mathcal{B}_\infty$. If we summarize the dynamics equations obtained:

\underline{Update dynamics for the hidden layer $\vh_t$:}  
	\begin{eqnarray}
	\frac{d\vv(\tau)}{d\tau} &=&-\vv(\tau)- \lambda_{SM} [((1-\beta )\bar{\vM}_H(t)+\beta\vD_1(t)\bar{\vM}_H(t)\vD_1(t))\vh(\tau) \nonumber \\
	                                && \hspace*{-0.25in}+\beta \vD_1(t)\vW_{HX}(t)\vx(\tau)+(1-\beta)\vW_{YH}(t)^T\vD_2(t)\vy(\tau)]  \nonumber\\
 \vh_{t,i}(\tau)&=&\sigma_A\left(\frac{\vv_i(\tau)}{{\lambda_{SM}\Gamma_{H}}_{ii}(t)((1 - \beta)+\beta {D_{1,ii}(t)}^2)}\right) \hspace{0.1in} \text{for } i=1, \ldots n. \nonumber
\end{eqnarray}
where $\sigma_A(\cdot)$ is the clipping nonlinearity with level $A$.

\underline{Update dynamics for the output $\vy_t$:} 
\begin{eqnarray}
\frac{d\vu(\tau)}{d\tau}&=&-\vu(\tau)+ \vW_{YH}(t)\vh(\tau)-\bar{\vM}_Y(t)\vD_2(t)\vy(\tau)\nonumber\\
\vy_{t,i}(\tau)&=& \sigma_1\left(\frac{\vu_i(\tau)}{{\Gamma_Y}_{ii}(t){D_{2,ii}(t)}}\right),  \hspace{0.2in} \text{for } i=1, \ldots n, \nonumber
\end{eqnarray}
Figure \ref{fig:NNBinfty} shows the corresponding two-layer neural network.
\begin{figure}[ht]
\begin{center}
    {\includegraphics[width=0.95\columnwidth, trim={0cm 15cm 5cm 0cm},clip]{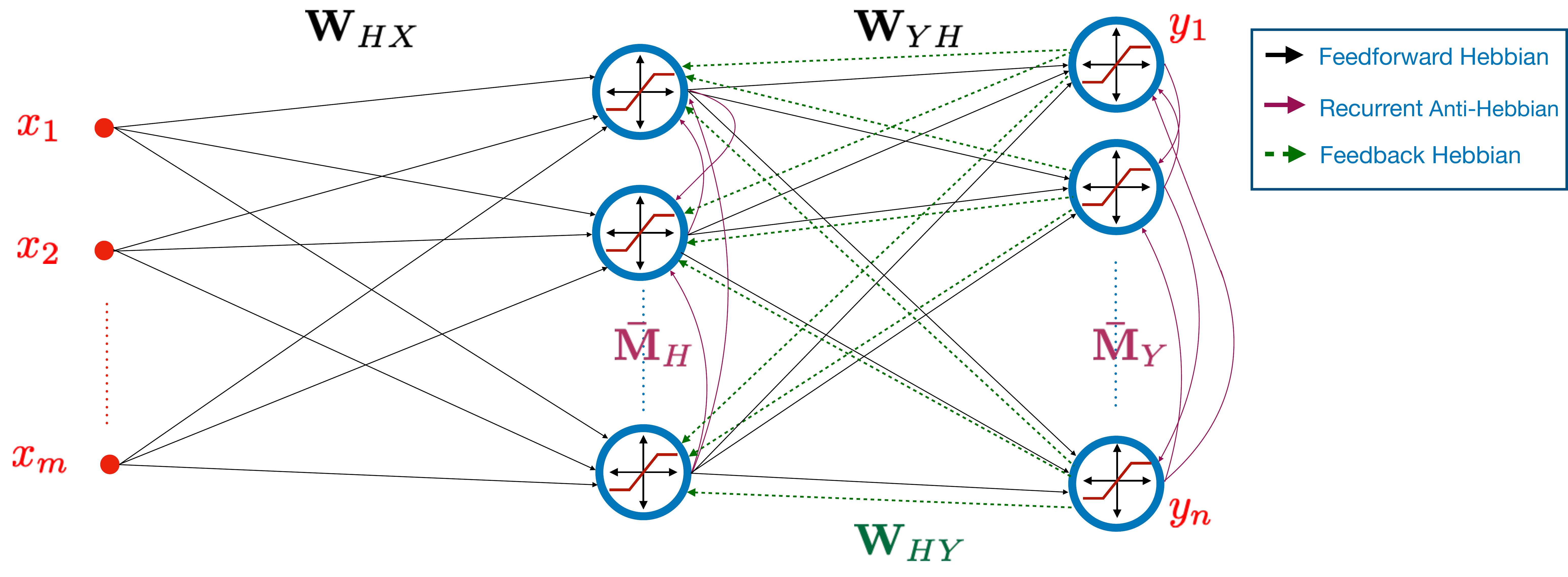}}
\end{center}
\caption{WSM Det-Max neural network for antisparse sources ($\Pcal=\mathcal{B}_\infty$).}
\label{fig:NNBinfty}
\end{figure}

\subsection{Nonnegative anti-sparse sources}
\label{appsec:nnantisparse}
For the case of nonnegative anti-sparse sources, the corresponding network is essentially the same as the antisparse case in Appendix \ref{appsec:antisparse}. The only difference is that the clipping activation functions at the output layer are replaced with nonnegative clipping function $\sigma_+(x)$ illustrated in Figure \ref{fig:nonnegclipping}.
\begin{figure}[ht]
\begin{center}
    {\includegraphics[width=0.6\columnwidth, trim={0cm 5cm 5cm 5cm},clip]{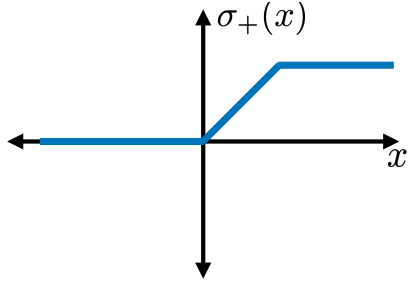}}
\end{center}
\caption{Nonnegative clipping function for elementwise projection to $\mathcal{B}_{\infty,+}$.}
\label{fig:nonnegclipping}
\end{figure}

As a result, we can write the network dynamics corresponding to the nonnegative anti-sparse case as 

\underline{Update dynamics for the hidden layer $\vh_t$:}  
	\begin{eqnarray}
	\frac{d\vv(\tau)}{d\tau} &=&-\vv(\tau)- \lambda_{SM} [((1-\beta )\bar{\vM}_H(t)+\beta\vD_1(t)\bar{\vM}_H(t)\vD_1(t))\vh(\tau) \nonumber \\
	                                && \hspace*{-0.25in}+\beta \vD_1(t)\vW_{HX}(t)\vx(\tau)+(1-\beta)\vW_{YH}(t)^T\vD_2(t)\vy(\tau)]  \nonumber\\
 \vh_{t,i}(\tau)&=&\sigma_A\left(\frac{\vv_i(\tau)}{{\lambda_{SM}\Gamma_{H}}_{ii}(t)((1 - \beta)+\beta {D_{1,ii}(t)}^2)}\right) \hspace{0.1in} \text{for } i=1, \ldots n, \nonumber
\end{eqnarray}
 
\underline{Update dynamics for the output $\vy_t$:} 
\begin{eqnarray}
\frac{d\vu(\tau)}{d\tau}&=&-\vu(\tau)+ \vW_{YH}(t)\vh(\tau)-\bar{\vM}_Y(t)\vD_2(t)\vy(\tau)\nonumber\\
\vy_{t,i}(\tau)&=& \sigma_1\left(\frac{\vu_i(\tau)}{{\Gamma_Y}_{ii}(t){D_{2,ii}(t)}}\right),  \hspace{0.2in} \text{for } i=1, \ldots n. \nonumber
\end{eqnarray}
The network corresponding to nonnegative anti-sparse sources is shown in Figure \ref{fig:NNBinftyplus}.
\begin{figure}[ht]
\begin{center}
    {\includegraphics[width=0.95\columnwidth, trim={0cm 15cm 5cm 0cm},clip]{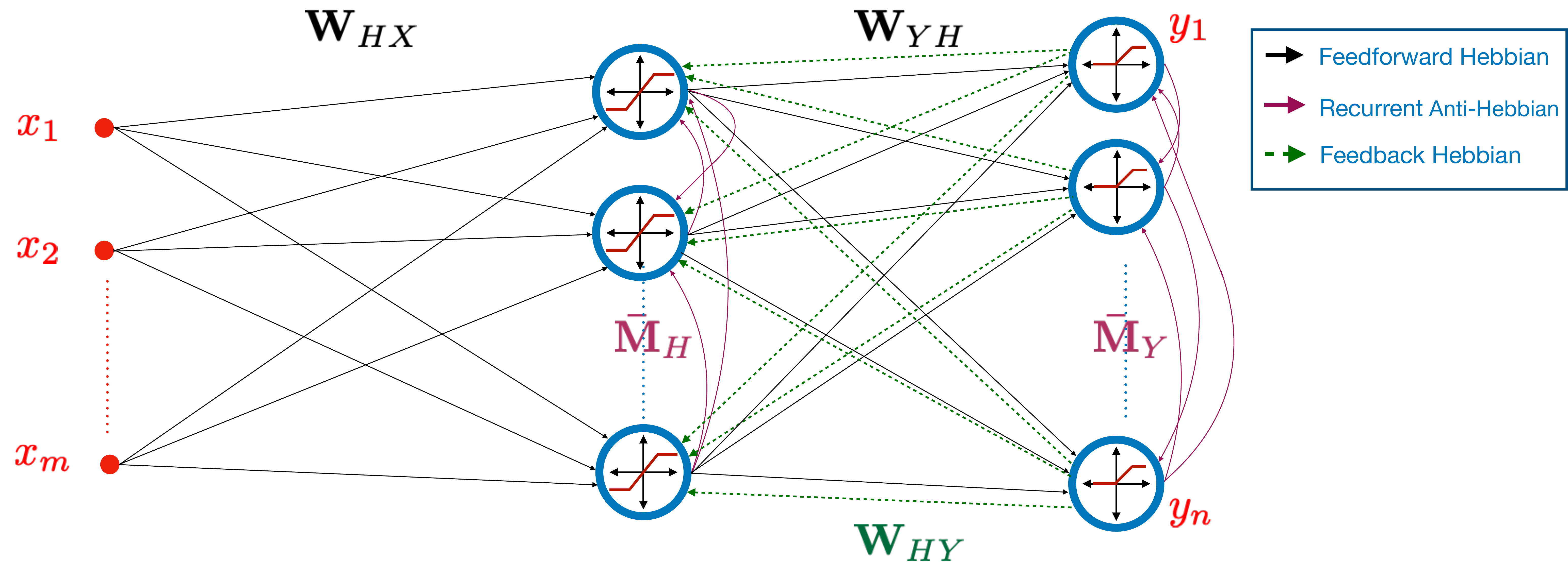}}
\end{center}
\caption{WSM Det-Max neural network for nonnegative anti-sparse sources ($\Pcal=\mathcal{B}_{\infty,+}$).}
\label{fig:NNBinftyplus}
\end{figure}

\subsection{  Nonnegative sparse sources}
\label{appsec:nl1dynamics}
For nonnegative sparse sources, i.e., $\mathcal{P}=\mathcal{B}_{1,+}$, we consider the following optimization setting:
\begin{eqnarray}
\underset{\vh_t,\vy_t}{\text{minimize}}& & \beta c_1(\vh_t)+(1-\beta)c_2(\vh_t,\vy_t) \nonumber \\ 
\mbox{subject to}& &\|\vy_t\|_1\le 1, \quad \vy_t\ge 0\label{eq:nonnegsparse2} 
\end{eqnarray}
for which the Lagrangian based reformulation can be written as
\begin{eqnarray}
\underset{\lambda_1\ge 0}{\text{maximize }}\underset{\vh_t,\vy_t}{\text{minimize}}& & \beta c_1(\vh_t)+(1-\beta)c_2(\vh_t,\vy_t)+\lambda_1(\|\vy_t\|_1-1)\nonumber 
\end{eqnarray}
The updates for $\vh_t$, gain variables $D_{1,ii},D_{2,ii}$ and the synaptic weights follow the equations provided in Section \ref{linfdynamics}. 

For the output component $\vy_t$, the corresponding cost function is an $\ell_1$ regularized quadratic cost function. Following the primal-dual approach in \citep{pehlevan2017clustering}, we can obtain the dynamic equations for output update as
\begin{eqnarray}
\frac{d\vu(\tau)}{d\tau}&=&-\vu(\tau)+ \lambda_{SM}(1 - \beta)[\vW_{YH}(t)\vh(\tau)-\bar{\vM}_Y(t)\vD_2(t)\vy(\tau)],\nonumber\\
\vy_{t,i}(\tau)&=& \text{ReLU}\left(\frac{\vu_i(\tau)}{\lambda_{SM}(1 - \beta){\Gamma_Y}_{ii}(t){D_{2,ii}(t)}}- \lambda_1(\tau) \right), \hspace{0.2in}\text{for } i=1, \ldots n, \nonumber
\end{eqnarray}
where
$\text{ReLU}(x,\lambda_1)$ is the rectified-linear unit mapping defined by $
\text{ReLU}(x)=\left\{  \begin{array}{cc} x & x>0, \\
                                     0 & \text{otherwise} \end{array} \right.\nonumber$.
 Based on the dual maximization, the Lagrangian variable $\lambda_1(\tau)$ is updated by
\begin{eqnarray}
\frac{da(\tau)}{d\tau}= -a(\tau)+\sum_{k=0}^n\vy_{t,k}(\tau)-1+\lambda_1(\tau), \qquad
\lambda_1(\tau)=\text{ReLU}(a(\tau)). \label{eq:inhibneur} 
\end{eqnarray}

According to the expressions obtained above, in addition to the hidden layer and the output layer neurons, there is an additional neuron corresponding to the Lagrangian variable $\lambda_1$ of whose dynamics is governed by (\ref{eq:inhibneur}). The corresponding neuron generates an inhibition signal for the output neurons, based on the total output activation. The corresponding network structure is shown in Figure  \ref{fig:NNBoneplus}.

\begin{figure}[ht]
\begin{center}
    {\includegraphics[width=0.95\columnwidth, trim={0cm 15cm 5cm 0cm},clip]{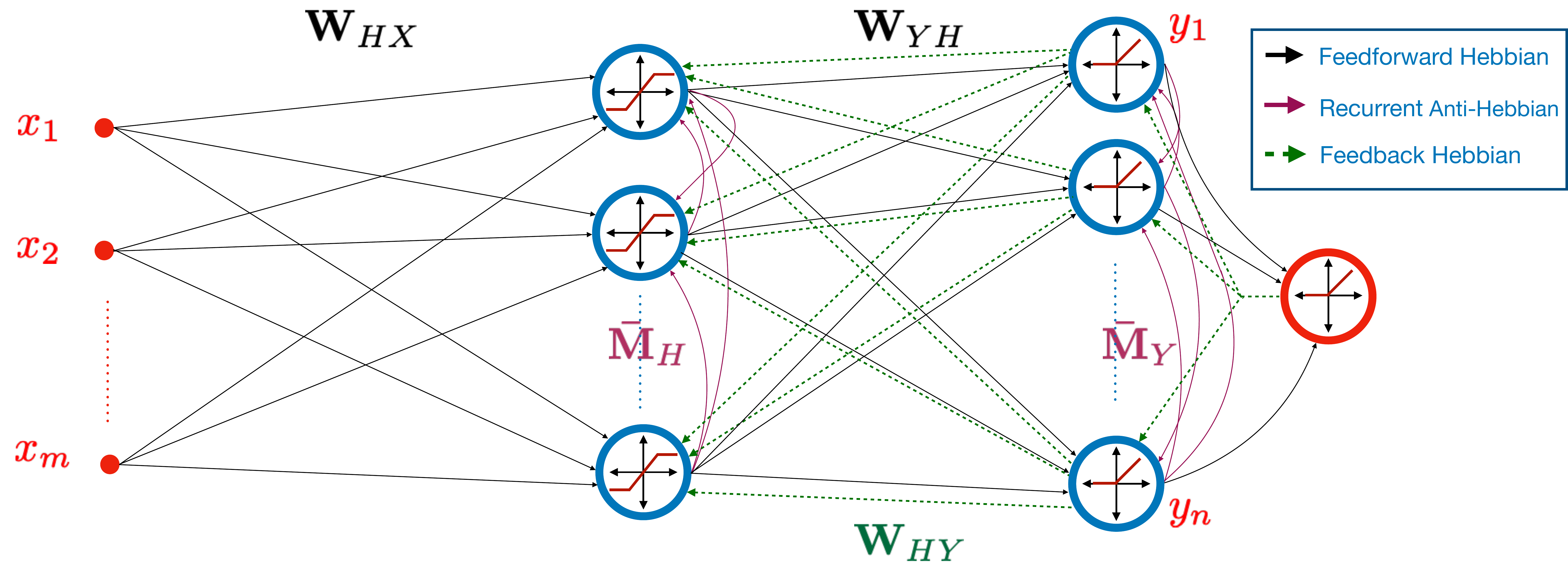}}
\end{center}
\caption{WSM Det-Max neural network for nonnegative sparse sources ($\Pcal=\mathcal{B}_{1,+}$).}
\label{fig:NNBoneplus}
\end{figure}

\subsection{Sparse sources}
\label{appsec:l1dynamics}
In the sparse source setting where $\Pcal=\mathcal{B}_1$, the only change relative to the nonnegative sparse case is the replacement of the ReLU output activation function with the soft thresholding function
\begin{eqnarray*}
ST_\lambda(x)=\left\{\begin{array}{cc}
                        0 & |x|\le \lambda \\
                        x-sign(x)\lambda & \mbox{otherwise}. \end{array}\right.
\end{eqnarray*}
Therefore, we can rewrite the output dynamics for $\Pcal=\mathcal{B}_1$ as 
\begin{eqnarray}
\frac{d\vu(\tau)}{d\tau}&=&-\vu(\tau)+ \lambda_{SM}(1 - \beta)[\vW_{YH}(t)\vh(\tau)-\bar{\vM}_Y(t)\vD_2(t)\vy(\tau)],\nonumber\\
\vy_{t,i}(\tau)&=& \text{ST}_{\lambda_1(\tau)}\left(\frac{\vu_i(\tau)}{\lambda_{SM}(1 - \beta){\Gamma_Y}_{ii}(t){D_{2,ii}(t)}} \right), \hspace{0.2in}\text{for } i=1, \ldots n, \nonumber\\
\frac{da(\tau)}{d\tau}&=& -a(\tau)+\sum_{k=0}^n|\vy_{t,k}(\tau)|-1+\lambda_1(\tau), \qquad
\lambda_1(\tau)=\text{ReLU}(a(\tau)). \nonumber
\end{eqnarray}

Figure \ref{fig:NNBone} illustrates the WSM based Det-Max neural network for sparse BSS. 

\begin{figure}[ht]
\begin{center}
    {\includegraphics[width=0.95\columnwidth, trim={0cm 15cm 5cm 0cm},clip]{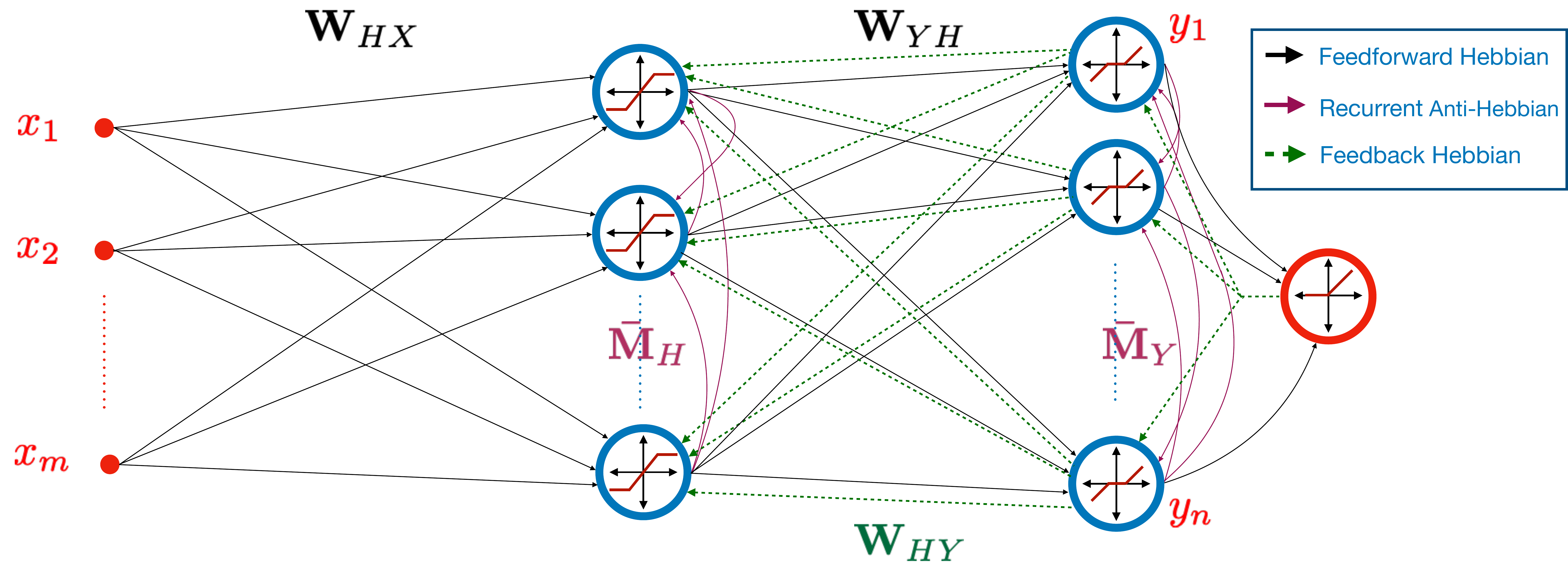}}
\end{center}
\caption{WSM Det-Max neural network for sparse sources ($\Pcal=\mathcal{B}_{1}$).}
\label{fig:NNBone}
\end{figure}

\subsection{Unit simplex sources}
\label{appsec:deltadynamics}
The unit simplex set $\Delta$ is a face of the polytope $\Pcal=\mathcal{B}_{1,+}$ which is the domain for nonnegative sparse sources. Therefore, we replace the $\ell_1$-norm inequality constraint in (\ref{eq:nonnegsparse2}) with the equality constraint to obtain the Det-Max WSM optimization problem for the unit simplex domain:
\begin{eqnarray}
\underset{\vh_t,\vy_t}{\text{minimize}}& & \beta c_1(\vh_t)+(1-\beta)c_2(\vh_t,\vy_t) \nonumber \\ 
\mbox{subject to}& &\|\vy_t\|_1= 1, \quad \vy_t\ge 0\nonumber
\end{eqnarray}
Therefore, for the Lagrangian based formulation
\begin{eqnarray}
\underset{\lambda_1}{\text{maximize }}\underset{\vh_t,\vy_t}{\text{minimize}}& & \beta c_1(\vh_t)+(1-\beta)c_2(\vh_t,\vy_t)+\lambda_1(\|\vy_t\|_1-1),\nonumber 
\end{eqnarray}
we no longer require $\lambda$ to be nonnegative. Therefore, for $\Pcal=\Delta$ only required change relative to $\mathcal{B}_{1,+}$ is the replacement of the  ReLU activation function of the rightmost inhibition neuron in Figure \ref{fig:NNBoneplus} with the linear activation. As a result, the output dynamics for the unit simplex sources can be written as
\begin{eqnarray}
\frac{d\vu(\tau)}{d\tau}&=&-\vu(\tau)+ \lambda_{SM}(1 - \beta)[\vW_{YH}(t)\vh(\tau)-\bar{\vM}_Y(t)\vD_2(t)\vy(\tau)],\nonumber\\
\vy_{t,i}(\tau)&=& \text{ReLU}\left(\frac{\vu_i(\tau)}{\lambda_{SM}(1 - \beta){\Gamma_Y}_{ii}(t){D_{2,ii}(t)}}- \lambda_1(\tau) \right), \hspace{0.2in}\text{for } i=1, \ldots n, \nonumber\\
\frac{d\lambda_1(\tau)}{d\tau}&=& -\lambda_1(\tau)+\sum_{k=0}^n\vy_k(\tau)-1+\lambda_1(\tau). \nonumber
\end{eqnarray}

Figure \ref{fig:NNdelta} shows the WSM-based Det-Max neural network for the unit-simplex sources.

\begin{figure}[ht]
\begin{center}
    {\includegraphics[width=0.95\columnwidth, trim={0cm 15cm 5cm 0cm},clip]{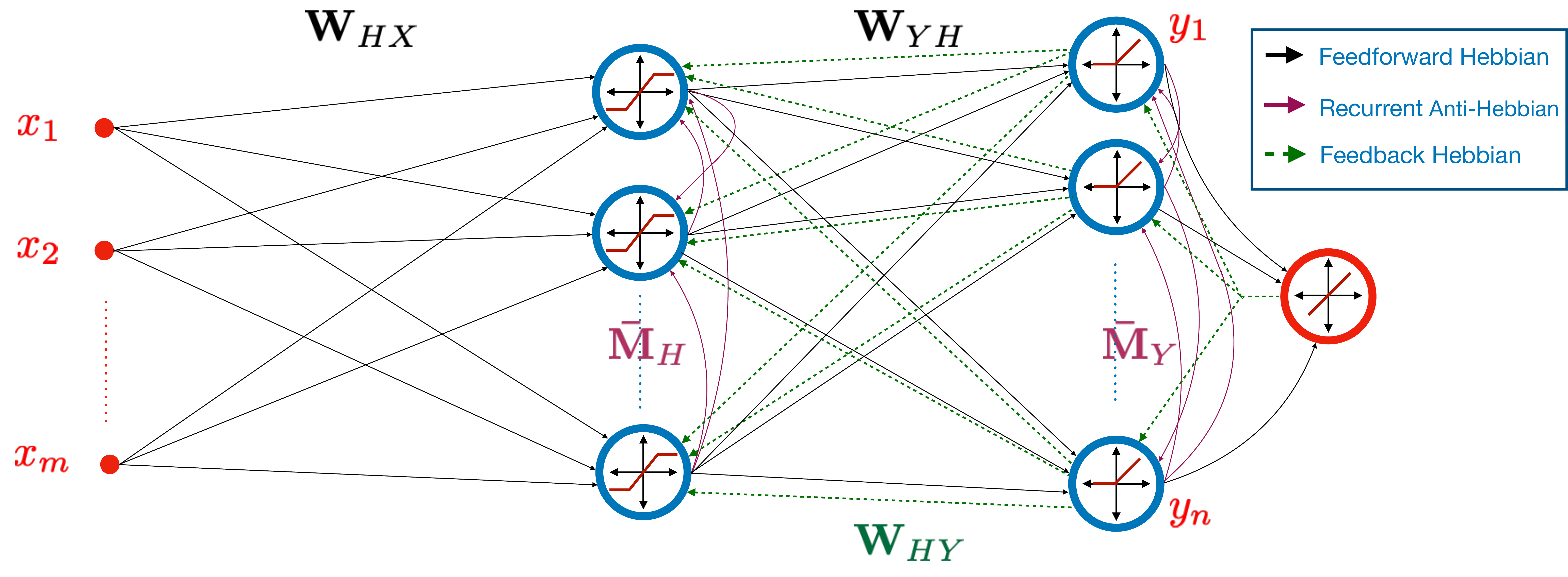}}
\end{center}
\caption{WSM Det-Max neural network for unit-simplex sources ($\Pcal=\Delta$).}
\label{fig:NNdelta}
\end{figure}

\subsection{Sources with mixed attributes}
\label{appsec:mixedattributesdynamics}
We consider the following polytope example provided in \cite{tatli2021tspsubmitted}
\begin{eqnarray}
\Pcal_{ex}=\left\{\mathbf{s}\in \mathbb{R}^3\ \middle\vert \begin{array}{l}   s_1,s_2\in[-1,1],s_3\in[0,1],\\ \left\|\left[\begin{array}{c} s_1 \\ s_2 \end{array}\right]\right\|_1\le 1,\left\|\left[\begin{array}{c} s_2 \\ s_3 \end{array}\right]\right\|_1\le 1 \end{array}\right\}, \label{eq:mixedpolytope}
\end{eqnarray}
which is an example of domains where source attributes such as nonnegativity and sparsity defined only at the subvector level.

The Det-Max WSM optimization setting for this case can be written as
\begin{eqnarray}
\underset{\vh_t,\vy_t}{\text{minimize}}& & C(\vh_t,\vy_t) \nonumber \\ 
\mbox{subject to}& &\left\|\left[\begin{array}{c} y_1 \\ y_2 \end{array}\right]\right\|_1\le 1, \quad \left\|\left[\begin{array}{c} y_2 \\ y_3 \end{array}\right]\right\|_1\le 1 \quad y_3\ge 0\nonumber
\end{eqnarray}
for which the Lagrangian based reformulation can be written as
\begin{eqnarray}
\underset{\lambda_1,\lambda_2\ge 0}{\text{maximize }}\underset{\vh_t,y_3\ge 0, y_1,y_2}{\text{minimize}}& & C(\vh_t,\vy_t)+\lambda_1\left(\left\|\left[\begin{array}{c} y_1 \\ y_2 \end{array}\right]\right\|_1-1\right)+\lambda_2 \left(\left\|\left[\begin{array}{c} y_2 \\ y_3 \end{array}\right]\right\|_1-1\right).\nonumber 
\end{eqnarray}
The proximal operator corresponding to the Lagrangian terms can be defined as
\begin{eqnarray}
\text{prox}_{\lambda_1,\lambda_2}(\mathbf{v})=\underset{q_3\ge0,q_1,q_2}{\text{argmin}}\left(\frac{1}{2}\|\mathbf{v}-\mathbf{q}\|_2^2+\lambda_1\left\|\left[\begin{array}{c} q_1 \\ q_2 \end{array}\right]\right\|_1+\lambda_2\left\|\left[\begin{array}{c} q_2 \\ q_3 \end{array}\right]\right\|_1\right).
\end{eqnarray}
Let $\mathbf{q}^*$ the output of the proximal operator. From the subdifferential set based optimality condition
\begin{itemize}
    \item if $q_1^*\neq 0$ then $q^*_1-v_1+\lambda_1\text{sign}(v_1)=0$ which implies $q^*_1=v_1-\lambda_1\text{sign}(v_1)$,
    \item if $q_2^*\neq 0$ then $q^*_2=v_2-(\lambda_1+\lambda_2)\text{sign}(v_2)$,
    \item if $q_3^*\neq 0$ then $q^*_3=v_3-\lambda_2$.
\end{itemize}
Therefore, we can write $q_1=\text{ST}_{\lambda_1}(v_1)$, $q_2=\text{ST}_{\lambda_1+\lambda_2}(v_2)$ and $q_3=\text{ReLU}(v_3-\lambda_2)$.
As a result, we can write the corresponding output dynamics expressions in the form
\begin{eqnarray*}
\frac{d\vu(\tau)}{d\tau}&=&-\vu(\tau)+ \lambda_{SM}(1 - \beta)[\vW_{YH}(t)\vh(\tau)-\bar{\vM}_Y(t)\vD_2(t)\vy(\tau)],\nonumber\\
\vy_{t,1}(\tau)&=& \text{ST}_{\lambda_1(\tau)}\left(\frac{\vu_1(\tau)}{\lambda_{SM}(1 - \beta){\Gamma_Y}_{11}(t){D_{2,11}(t)}} \right), \\
\vy_{t,2}(\tau)&=& \text{ST}_{\lambda_1(\tau)+\lambda_2(\tau)}\left(\frac{\vu_2(\tau)}{\lambda_{SM}(1 - \beta){\Gamma_Y}_{22}(t){D_{2,22}(t)}} \right), \\
\vy_{t,3}(\tau)&=& \text{ReLU}\left(\frac{\vu_3(\tau)}{\lambda_{SM}(1 - \beta){\Gamma_Y}_{33}(t){D_{2,33}(t)}}-\lambda_2(\tau) \right),\\
\frac{da_1(\tau)}{d\tau}&=& -a_1(\tau)+|\vy_{t,1}(\tau)|+|\vy_{t,2}(\tau)|-1+\lambda_1(\tau),\\
\lambda_1(\tau)&=&\text{ReLU}(a_1(\tau)), \nonumber\\
\frac{da_2(\tau)}{d\tau}&=& -a_2(\tau)+|\vy_{t,2}(\tau)|+\vy_{t,3}(\tau)-1+\lambda_2(\tau),\\
\lambda_2(\tau)&=&\text{ReLU}(a_2(\tau)). \nonumber
\end{eqnarray*}

Figure \ref{fig:NNmixed} shows the Det-Max WSM neural network for the source domain in (\ref{eq:mixedpolytope}).

\begin{figure}[ht]
\begin{center}
    {\includegraphics[width=0.99\columnwidth, trim=0cm 15cm 0cm 0cm,clip]{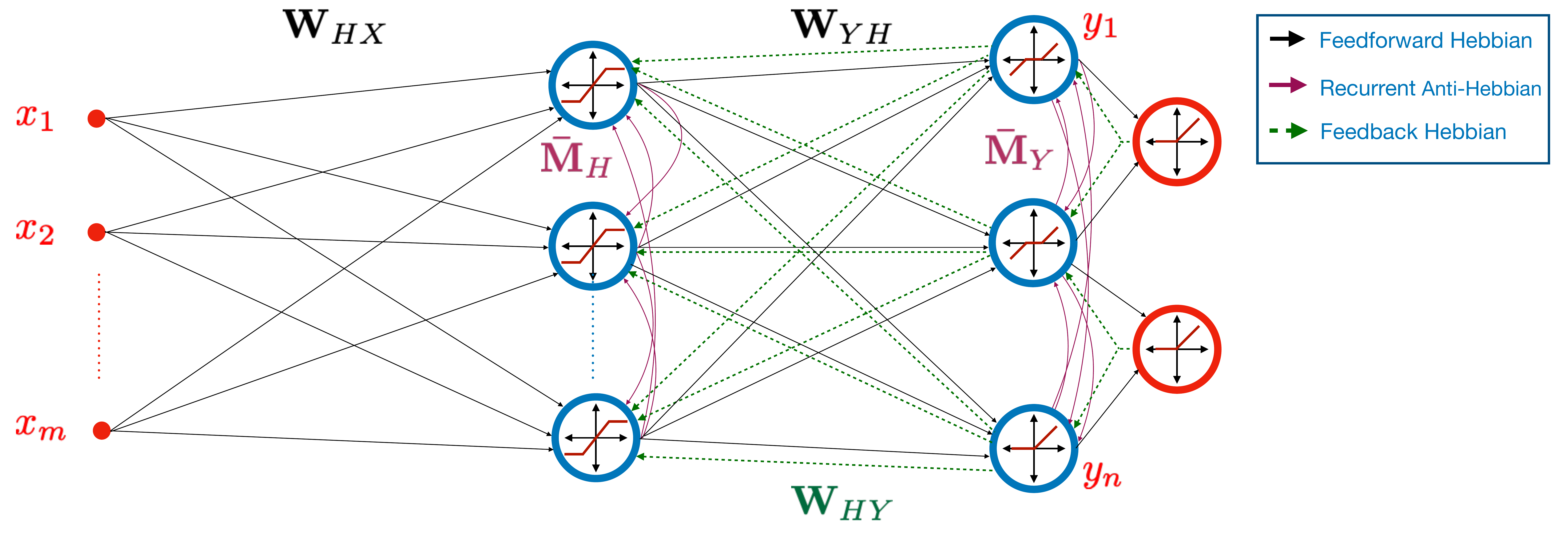}}
\end{center}
\caption{WSM Det-Max neural network for the polytope in (\ref{eq:mixedpolytope}).}
\label{fig:NNmixed}
\end{figure}


\section{Supplementary  on numerical experiments}
\label{sec:numexpappendix}
Update dynamics for the hidden layer $\vh_t$ and the output vector $\vy_t$ are defined by differential equations depending on the selection of the source domain which lead to recursive neural dynamic iterations. Algorithm \ref{alg:neuraldynamiciterationsanti} summarizes the neural dynamic iterations for anti-sparse sources covered in Section \ref{linfdynamics}. Very similar output dynamic calculations for each source assumption can be acquired based on the derivations in Section \ref{sec:dersourcedomains}. We run the neural dynamic iterations until a convergence check is satisfied or a predetermined maximum number of iterations $\tau_{max}$ is reached. In Algorithm \ref{alg:neuraldynamiciterationsanti}, $\epsilon$ denotes the tolerance in the relative error check for the stopping condition, and $\eta(\tau)$ represents the learning rate at the iteration count $\tau$. In the following subsections, we provide the experimental details and additional source separation examples for different assumptions on the sources.
\begin{algorithm}[H]
\begin{algorithmic}[1]
\STATE Initialize  $\tau_{\text{max}}$, $\epsilon$, and $\tau = 1$

\WHILE{($||\vv(\tau) - \vv(\tau-1)||/||\vv(\tau)|| > \epsilon$ or $||\vu(\tau) - \vu(\tau-1)||/||\vu(\tau)|| > \epsilon$) and $\tau < \tau_{\text{max}}$}
\STATE $\vv(\tau) = \vv(\tau - 1) + \eta(\tau)\frac{d\vv(\tau - 1 )}{d(\tau -1) } $
\STATE Apply Equation \ref{eq:htnn} for  $\vh_{t,i}(\tau)$
\STATE $\vu(\tau) = \vu(\tau - 1) + \eta(\tau)\frac{d\vu(\tau -1)}{d(\tau-1)} $
\STATE Apply Equation \ref{eq:ytnn} for  $\vy_{t,i}(\tau)$
\STATE $\tau = \tau + 1$
\ENDWHILE
\end{algorithmic}
\caption{Neural dynamic iterations for anti-sparse sources}
\label{alg:neuraldynamiciterationsanti}
\end{algorithm}

\subsection{Batch algorithms with correlated source separation capability}
In this section, we briefly discuss two batch learning algorithms for blind separation of correlated sources, which reflect the Det-Max problem \ref{eq:bssobjective}: 1. Polytopic Matrix Factorization \cite{tatli2021tspsubmitted}, 2. Log-Det Mutual Information Maximization \cite{erdogan2022information}.
\begin{itemize}
\item {\it \textbf{Polytopic Matrix Factorization}}: \cite{tatli2021tspsubmitted} recently introduced the Polytopic Matrix Factorization (PMF) as a structured matrix factorization framework that models the columns of the input matrix, i.e., the mixture signals in our problem, as a linear transformation of source vectors from a polytope. The choice of the underlying polytope in the PMF framework reflects the attributes of the sources possibly in a heterogeneous perspective; e.g., the polytope discussed in Section \ref{appsec:mixedattributesdynamics} provides an example of heterogeneous feature assumptions at the subvector level such as mutual sparsity. Taking into account the mixing model in Section \ref{sec:mixingmodel}, PMF uses the following optimization problem,

\begin{maxi!}[l]<b>
{\vY(t)\in \mathbb{R}^{n \times t}, \vH \in \mathbb{R}^{m \times n}}{ \log(\det(\vY(t)\vY(t)^T))\label{eq:pmfdetmaxobjective}}{\label{eq:pmfobjective}}{}
\addConstraint{\vX(t) = \vH \vY(t)}{\label{eq:pmfequalityconstraint}}{}
\addConstraint{\vy_i \in \mathcal{P}, i=1, \dots, t,}{\label{eq:pmfdetmax}}{}
\end{maxi!}

where $\vH$ and $\vY(t)$ correspond to the unknown mixing matrix and the source estimates, respectively. The aim of PMF is to obtain the original factors of $\vA$ and $\vS(t)$ up to some acceptable sign and permutation ambiguities, i.e., $\vY(t)=\vP \bLambd\vS(t)$ and $\vH =\vA \vP^T\bLambd^{-1} $. The reference \cite{tatli2021tspsubmitted} provides the sufficient condition for the identifiability of the original factors of $\vA$ and $\vS(t)$ based on the sufficiently scattering condition discussed in Section \ref{appsec:suffscat}, i.e., if the source vectors are sufficiently scattered in a permutation-and/or sign only invariant polytope $\Pcal$, then all global optima of the problem \ref{eq:pmfobjective} lead to the ideal separation. For the corresponding algorithm to solve the problem \ref{eq:pmfobjective}, we refer to the pseudo-code in \cite{tatli2021tspsubmitted}, which is a batch algorithm with a projected gradient search.

\item {\it \textbf{Log-Det Mutual Information Maximization}}: The reference \cite{erdogan2022information} brings a statistical interpretation to the PMF framework based on a log-determinant (LD) based  mutual information measure. According to this approach, the LD-mutual information between the input and output is maximized, under the constraint that the outputs are in the presumed source domain. The corresponding optimization setting is given by

\begin{maxi!}[l]<b>
{\vY(t)\in \mathbb{R}^{n \times t}}{ \frac{1}{2}\log\det(\boldsymbol{\hat{R}_y} + \epsilon \boldsymbol{I} ) - \frac{1}{2}\log\det(\boldsymbol{\hat{R}_y} - \boldsymbol{\hat{R}_{yx}}(\epsilon \boldsymbol{I} + \boldsymbol{\hat{R}_{x}})^{-1} \boldsymbol{\hat{R}_{yx}}^T+ \epsilon \boldsymbol{I} )\label{eq:ldmidetmaxobjective}}{\label{eq:ldmiobjective}}{}
\addConstraint{\vy_i \in \mathcal{P}, i=1, \dots, t,}{\label{eq:ldmidetmax}}{}
\end{maxi!}

where the objective \ref{eq:ldmidetmaxobjective} is defined in terms of sample covariance matrices, i.e., $\boldsymbol{\hat{R}_y} = \frac{1}{t} \vY(t)\vY(t)^T - \frac{1}{t^2}\vY(t) \boldsymbol{1}\boldsymbol{1}^T\vY(t)^T$, and $\boldsymbol{\hat{R}_{yx}} = \frac{1}{t} \vY(t)\vX(t)^T - \frac{1}{t^2}\vY(t) \boldsymbol{1}\boldsymbol{1}^T\vX(t)^T$. Similar to the PMF framework, the LD-InfoMax approach assumes that the source vectors are drawn from a presumed polytope $\Pcal$. The LD-InfoMax approach is capable of separating correlated sources, since it does not assume any statistical independence or uncorrelatedness on the source vectors. Reference \cite{erdogan2022information} proposes a projected gradient ascent-based algorithm to solve the problem \ref{eq:ldmiobjective} as a batch learning approach.
\end{itemize}

We compare our algorithm with the PMF and LD-InfoMax frameworks for correlated source separation experiments in Sections \ref{sec:numexpcorrcopula}, \ref{sec:numexpcorrcopulaantisparse}, \ref{sec:imageseparaionappendix}, and for sparse source separation experiment in Section \ref{appsec:sparsesourceseparation}.

\subsection{Synthetically correlated source separation with nonnegative anti-sparse sources}
In this section, we provide the training details and hyperparameter selections for the numerical experiment provided in Section \ref{sec:numexpcorrcopula}. For this network, we used the following hyperparameter selections and variable initializations:
\begin{itemize}
    \item $\vD_1 = \vI$, and $\vD_2 = \vI$, where $\vI$ is the identity matrix.
    \item $\mu_{\vD_1} = 1$, and $\mu_{\vD_2} = 10^{-2}$.
    \item $\beta = 0.5,\  \lambda_{SM} = 1 - 10^{-5}$.
    \item $1 - \gamma^2$ is dynamically adjusted using $1 - \gamma^2 = max\{\nu/(1 + \log(1 + t)),0.001\}$,  where $t$ is the data sample index, and $\nu =\left\{\begin{array}{cc}  0.1 & \rho \leq 0.4, \\
   0.05 & \text{otherwise.} \end{array} \right.$.
    \item $\vM_H = 2\vI,  \vM_Y = \vI$.
    \item $\vW_{HX} = \vI, \vW_{YH}  = \vI$.
    \item Learning rate for the neural dynamic iterations is adjusted using $max\{0.75/(1 + \tau \times 0.005),0.05\}$, where $\tau$ is the neural dynamic iteration count.
    \item The maximum number of neural dynamic iterations is restricted to $\tau_{\text{max}} = 500$ if the stopping condition is not satisfied.
    \item For the stability of the learning process, we keep the diagonal weights of $\vD_1$ and $\vD_2$ in a predetermined range, i.e., $0.2 \prec \mbox{diag}(\vD_1) \prec 10^6$ and $0.2 \prec \mbox{diag}(\vD_2) \prec 5$.
\end{itemize}

\subsection{Synthetically correlated source separation with anti-sparse sources}
\label{sec:numexpcorrcopulaantisparse}
To illustrate the correlated source separation of WSM neural networks with antisparse sources, we consider a numerical experiment with four copula-T distributed sources in the range $[-1, 1]$ with a Toeplitz correlation calibration matrix whose first row is $\begin{bmatrix}1 & \rho & \rho^2 & \rho^3 \end{bmatrix}$. We consider the range $\rho \in \left[0,0.6\right]$ for the correlation level. The sources are mixed with an $8\times4$ random matrix with i.i.d. standard normal entries, and corrupted by i.i.d. standard normal noise corresponding to $30$dB SNR level. Antisparse-WSM neural network is employed in this experiment, which is illustrated in Figure \ref{fig:NNBinfty}. We compare the SINR performance of WSM algorithm with the  BSM algorithm \citep{erdogan2020blind}, Infomax ICA algorithm \citep{GramfortEtAl2013a}, PMF algorithm \cite{tatli2021tspsubmitted}, and LD-InfoMax algorithm \cite{erdogan2022information}. Figure \ref{fig:Copulafigantisparseappendix} illustrates the SINR performances of these algorithms (averaged over 100 realizations) with respect to the correlation parameter $\rho$. Similar to the results for nonnegative antisparse source separation experiments provided in Section \ref{sec:numexpcorrcopula}, the WSM approach maintains its immunity against source correlations, whereas the BSM and ICA algorithms, which assume uncorrelated sources, deteriorate with increasing source correlation. LD-InfoMax and PMF algorithms achieve relatively similar SINR behaviors while their performance remains comparatively steady with respect to increasing source correlation. We note that both PMF and LD-InfoMax typically achieve better performances compared to our proposed online algorithm since these approaches utilize batch algorithms.

\begin{figure}[H]
\begin{center}
\includegraphics[trim = {0.0cm 0cm 2.0cm 1.0cm},clip,width=0.7\textwidth]{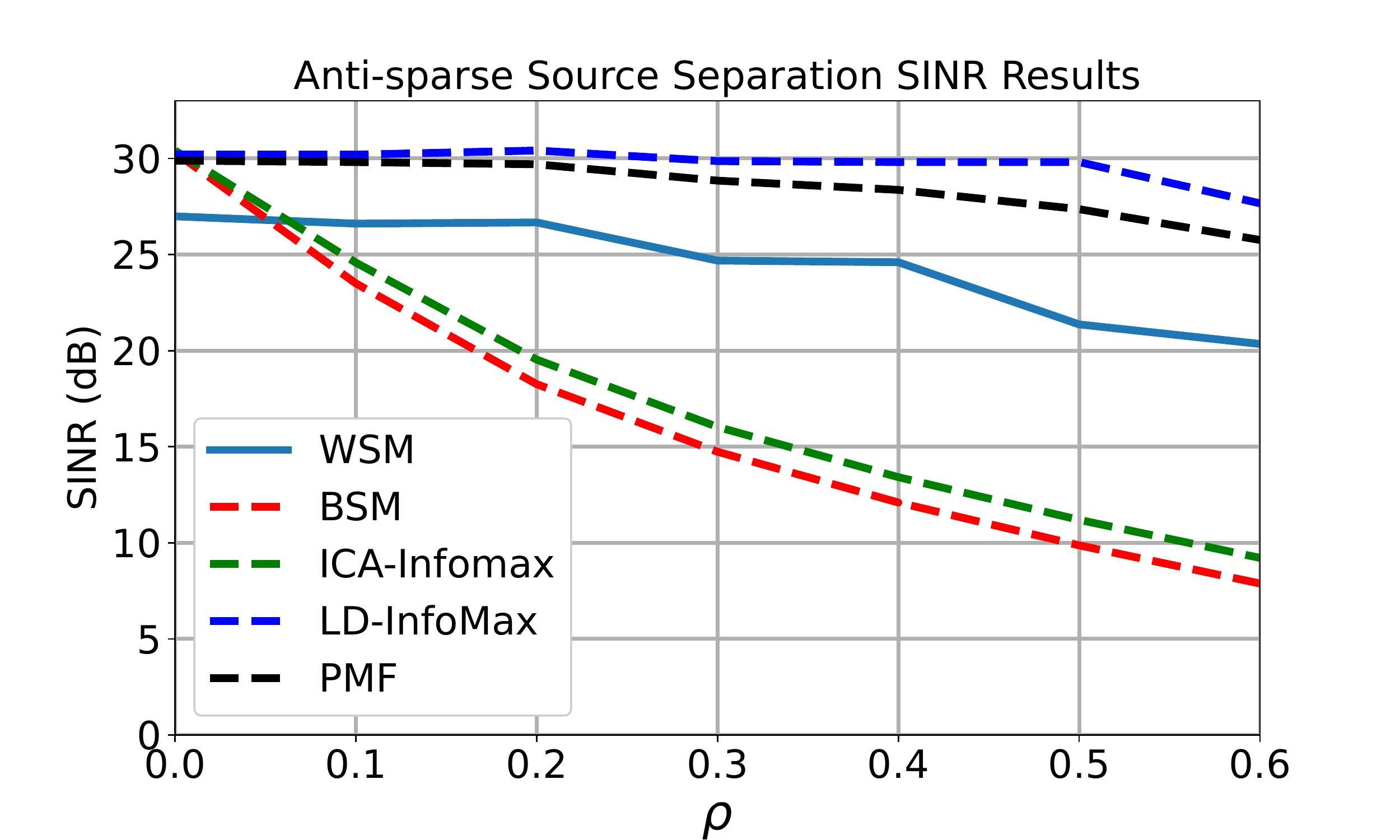}
\end{center}
\caption{The SINR performances of the WSM, BSM, ICA, LD-InfoMax, and PMF algorithms as a function of the correlation parameter $\rho$.}
\label{fig:Copulafigantisparseappendix}
\end{figure}

For the antisparse source separation setting, we used the following hyperparameter selections and variable initializations:
\begin{itemize}
    \item $\vD_1 = \vI$, and $\vD_2 = \vI$, 
    \item $\mu_{\vD_1} = 1.125$, and $\mu_{\vD_2} = 0.2$,
    \item $\beta = 0.5,\  \lambda_{SM} = 1 - 5 \times 10^{-5}$,
    \item $1 - \gamma^2$ is dynamically adjusted using $1 - \gamma^2 = max\{\nu/(1 + \log(1 + t)),0.001\}$,  where $t$ is the data sample index, and $\nu =\left\{\begin{array}{cc}  0.6 & \rho \leq 0.4, \\
   0.25 & \text{otherwise.} \end{array} \right.$.
    \item $\vM_H = 2\vI,  \vM_Y = \vI$,
    \item $\vW_{HX} = \vI, \vW_{YH}  = \vI$.
    \item Learning rate for the neural dynamic iterations is adjusted using $max\{0.75/(1 + \tau \times 0.005),0.05\}$, where $\tau$ is the neural dynamic iteration count.
    \item The maximum number of neural dynamic iterations is restricted to $\tau_{\text{max}} = 750$ if the stopping condition is not satisfied.
    \item For the stability of the learning process, we keep the diagonal weights of $\vD_1$ and $\vD_2$ in a predetermined range, i.e., $0.2 \prec \mbox{diag}(\vD_1) \prec 10^6$ and $0.2 \prec \mbox{diag}(\vD_2) \prec 5$.
\end{itemize}

\subsection{Image separation}
\label{sec:imageseparaionappendix}





For the image separation example provided in Section \ref{sec:numexpimageseparation}, the WSM Det-Max Neural Network illustrated in Figure \ref{fig:NNBinftyplus} is employed. For this network, we used the following hyperparameter selections and variable initializations:

\begin{itemize}
    \item $\vD_1 = \vI$, and $\vD_2 = \vI$.
    \item $\mu_{\vD_1} = 3.725$, and $\mu_{\vD_2} = 1.125$.
    \item $\beta = 0.5,\  \lambda_{SM} = 1 - 10^{-5}$.
    \item $1 - \gamma^2$ is dynamically adjusted using $1 - \gamma^2 = max\{0.11/(1 + \log(1 + t)),0.001\}$,  where $t$ is the data sample index.
    \item $\vM_H = 2\vI,  \vM_Y = \vI$.
    \item $\vW_{HX} = \vI, \vW_{YH}  = \vI$.
    \item Learning rate for the neural dynamic iterations is adjusted using $max\{0.75/(1 + \tau \times 0.005),0.05\}$, where $\tau$ is the neural dynamic iteration count.
    \item Maximum number of neural dynamic iterations is restricted to be $\tau_{\text{max}} = 500$ if stopping condition is not satisfied.
    \item For the stability of the learning process, we keep the diagonal weights of $\vD_1$ and $\vD_2$ in a predetermined range, i.e., $10^{-3} \prec \mbox{diag}(\vD_1) \prec 10^6$ and $10^{-3} \prec \mbox{diag}(\vD_2) \prec 20$.
\end{itemize}
In this section, we also include the results of batch algorithms PMF and LD-InfoMax as illustrated in Figure \ref{fig:imageseparationallappendixcontinued} in addition to the source images, mixture images, and the outputs of the ICA, NSM, and WSM algorithms with better resolutions compared to Figure \ref{fig:imageseparation}. Recall that our WSM-based network outputs illustrated in Figure \ref{fig:imoutputswsmappendix} achieves SINR level of 27.49 dB. LD-InfoMax algorihtm's outputs in Figure \ref{fig:imoutputsldinfomaxappendix} obtain SINR level of 28.65 dB, and the PMF algorithm's outputs in Figure \ref{fig:imoutputspmfappendix} obtain the SINR level of 31.92 dB. As expected, both PMF and LD-InfoMax algorithms achieve better performances due to their batch nature whereas our proposed approach's output is compatible with these frameworks.

\begin{figure}[H]
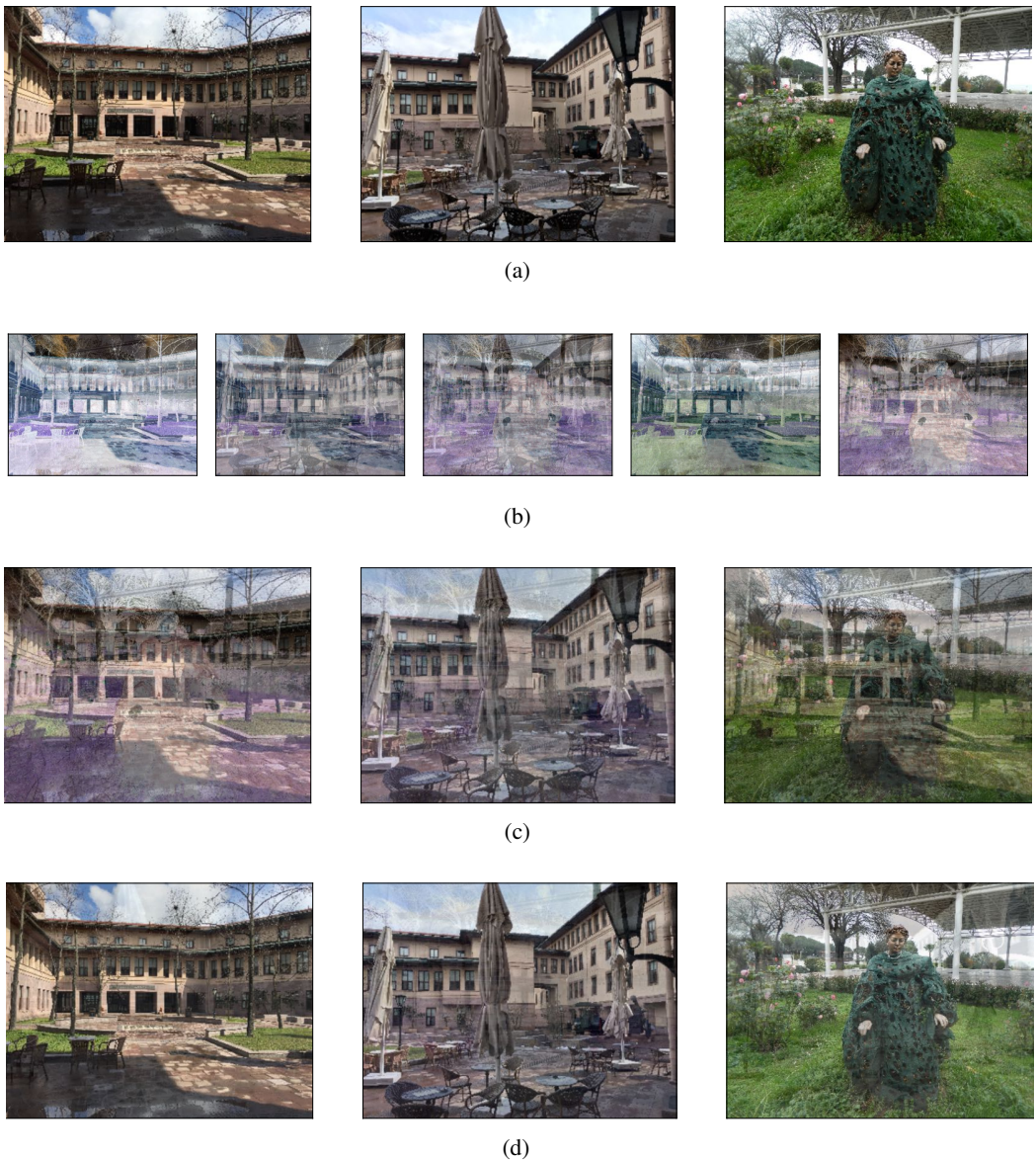

\centering
\subfloat[a][]{
\includegraphics[trim = {2cm 0cm 2cm 0cm},clip,width=0.99\textwidth]{picture_results_new/ImageSeparationOriginal_Images32.pdf}
\label{fig:imsourcesappendix}}\\
\subfloat[b][]{
\includegraphics[trim = {1.2cm 0cm 1.2cm 0cm},clip,width=0.99\textwidth]{picture_results_new/ImageSeparationMixture_Images32.pdf}
\label{fig:immixturesappendix}
}

\subfloat[c][]{{
\includegraphics[trim = {2cm 0cm 2cm 0cm},clip,width=0.99\textwidth]{picture_results_new/ICAOutput2.pdf} }\label{fig:imoutputsicaappendix}}

\subfloat[d][]{{
\includegraphics[trim = {2cm 0cm 2cm 0cm},clip,width=0.99\textwidth]{picture_results_new/NSM_wPreWhiteningOutput2.pdf}}\label{fig:imoutputsnsmappendix}}

\caption{(a) Original RGB images, (b) mixture RGB images, (c)  ICA outputs,  (d)  NSM outputs.}
\qquad
\hfill
\label{fig:imageseparationallappendix}
\end{figure}

\begin{figure}[H]\ContinuedFloat
\centering

\subfloat[e][]{{
\includegraphics[trim = {2cm 0cm 2cm 0cm},clip,width=0.99\textwidth]{picture_results_new/WSM_Output3_Clipped2.pdf}}\label{fig:imoutputswsmappendix}}

\subfloat[f][]{{
\includegraphics[trim = {2cm 0cm 2cm 0cm},clip,width=0.99\textwidth]{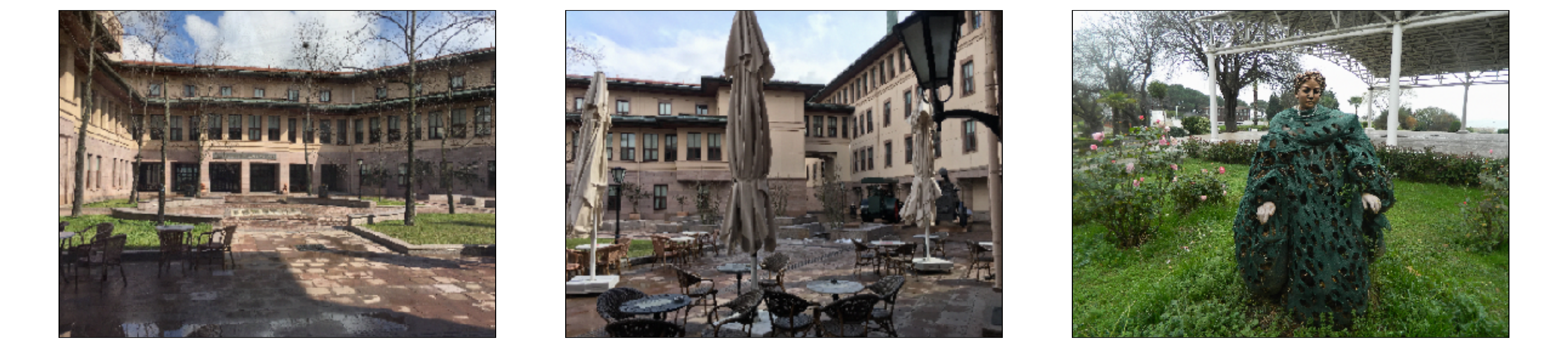}}\label{fig:imoutputsldinfomaxappendix}}

\subfloat[g][]{{
\includegraphics[trim = {2cm 0cm 2cm 0cm},clip,width=0.99\textwidth]{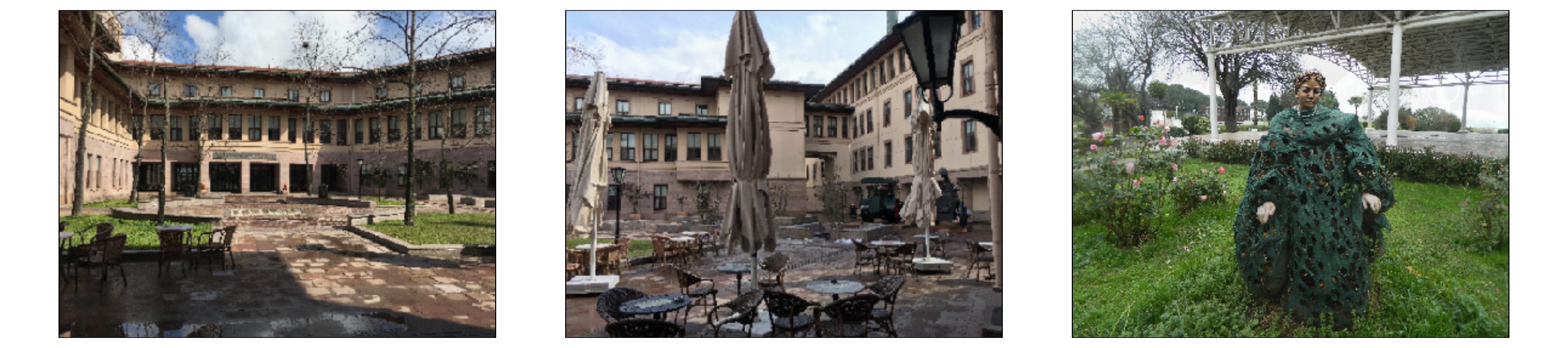}}\label{fig:imoutputspmfappendix}}

\caption{ (e)  WSM outputs, (f) LD-InfoMax Outputs, (g) PMF Outputs. }
\label{fig:imageseparationallappendixcontinued}
\qquad
\hfill
\end{figure}

\subsection{Sparse source separation}
\label{appsec:sparsesourceseparation}
In order to illustrate the use of the proposed framework for a different source domain, we consider sparse sources where $\Pcal=\mathcal{B}_1$. We generate $n=5$ dimensional source vectors, by projecting i.i.d. uniform vectors in $\mathcal{B}_\infty$ to $\mathcal{B}_1$. The mixing matrix is a $10\times 5$-matrix with i.i.d. standard normal entries. The mixtures are used to train the sparse-WSM Det-Max network in Figure \ref{fig:NNBone} introduced in Appendix \ref{appsec:l1dynamics}. For this network, we used the following hyperparameter selections and variable initializations:
\begin{itemize}
    \item $\vD_1 = 8\vI$, and $\vD_2 = \vI$.
    \item $\mu_{\vD_1} = 20$, and $\mu_{\vD_2} = 10^{-2}$.
    \item $\beta = 0.5,\  \lambda_{SM} = 1 - 10^{-5}$.
    \item $1 - \gamma^2$ is dynamically adjusted using $1 - \gamma^2 = max\{0.25/(1 + \log(1 + t)),0.001\}$,  where $t$ is the data sample index.
    \item $\vM_H = 0.02\vI,  \vM_Y = 0.02\vI$.
    \item $\vW$ matrices are initialized first with i.i.d. standard normal random variables. Then, we normalized the Euclidean norm of all rows  to $0.0033$ by proper scaling.
    \item Learning rate for the neural dynamic iterations is determined to be $0.5$.
    \item Maximum number of neural dynamic iterations is restricted to be $\tau_{\text{max}} = 750$ if stopping condition is not satisfied.
    \item For the stability of the learning process, we keep the diagonal weights of $\vD_1$ and $\vD_2$ in a predetermined range, i.e., $10^{-6} \prec \mbox{diag}(\vD_1) \prec 10^6$ and $1 \prec \mbox{diag}(\vD_2) \prec 1.001$.
\end{itemize}
Figure \ref{fig:spex1} illustrates the SINR convergence behavior for the sparse-WSM network, as a function of update iterations, for the input SNR level of $30$dB (averaged over 200 realizations).

\begin{figure}[H]
\centering
\includegraphics[width=0.75\columnwidth]{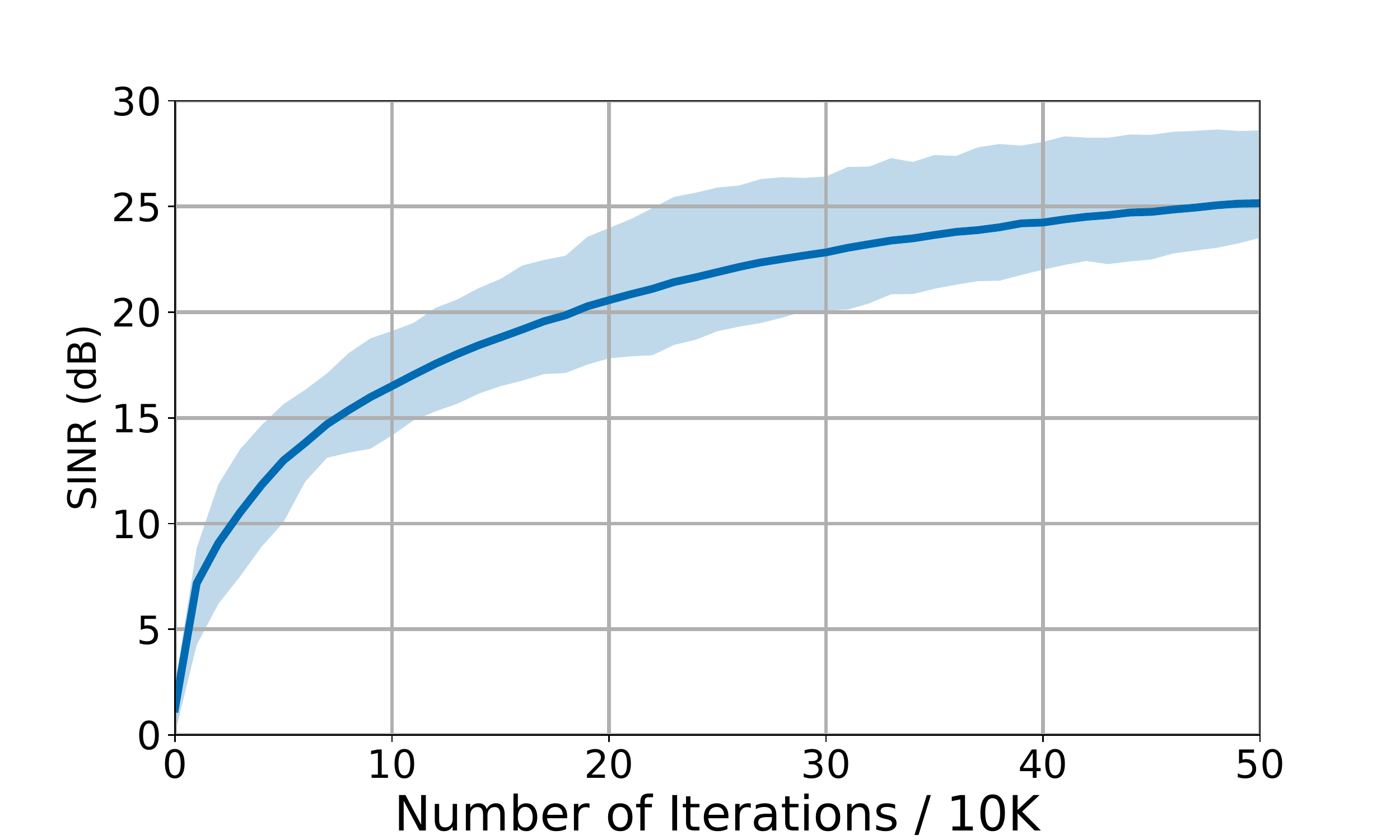}%
\caption{The SINR convergence curve for the sparse-WSM for $30$dB input SNR level: mean-solid line with 25/75-percentile envelope.}
\label{fig:spex1}
\end{figure}

Figure \ref{fig:spex2} demonstrates the separation performance of the sparse-WSM network for different noise levels. 
\begin{figure}[H]
\centering
\includegraphics[width=0.75\columnwidth]{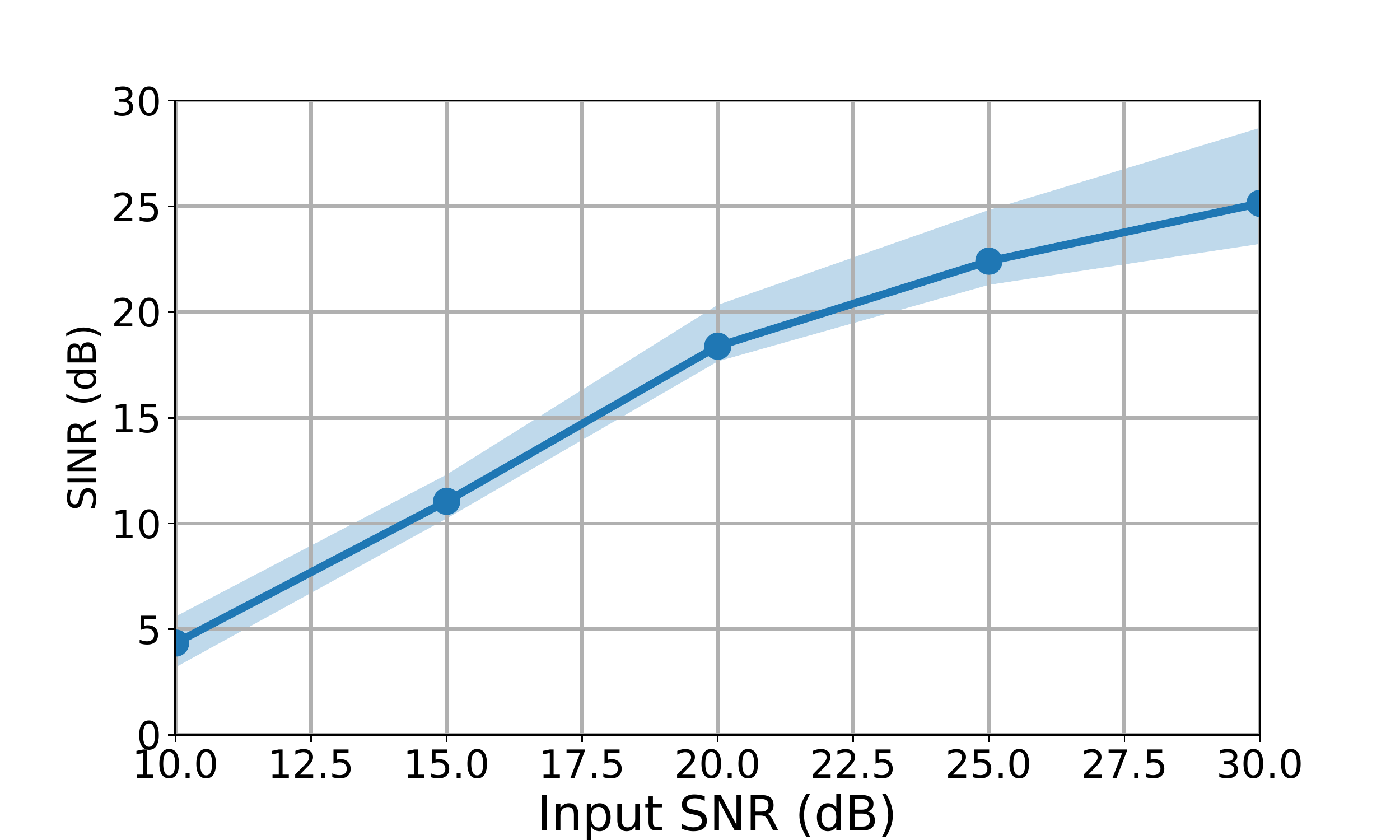}%
\caption{The output SINR with respect to the input SNR level for the sparse-WSM: mean-solid line with 25/75-percentile envelope. }
\label{fig:spex2}
\end{figure}

To compare our online approach with the batch algorithms LD-InfoMax and PMF, we also performed experiments with these algorithms for the input SNR level of 30 dB. Table \ref{tab:sparseBSSResults} summarizes the averaged SINR results of each algorithm over 200 realizations for 30 dB input SNR level. In these experiments, we observe that both PMF and LD-InfoMax obtain better SINR performances on average compared to our WSM Det-Max network. This condition is due to the batch nature of both PMF and LD-InfoMax as discussed earlier.
\begin{table}[h!]
    \centering
\caption{Sparse source separation averaged SINR results of WSM, PMF, and LD-InfoMax.}
    \label{tab:sparseBSSResults}
    \vspace{0.1in} 

    \begin{tabular}{c c | c l l }
    \hline
    & {\bf Algorithm} & {\bf WSM} & {\bf PMF} & {\bf  LD-InfoMax}  \\
    \hline 
    & {\bf SINR} & 25.14 &  30.17 & 30.0 \\

    \hline

    \end{tabular}
\end{table}

\subsection{Sparse dictionary learning}
\label{sec:sparsedictlearningappendix}
Related to the previous example, we consider the well-known example of sparse coding, which is the dictionary learning for natural image patches \citep{olshausen1997sparse}. 
For this experiment, we used $12 \times 12$ prewhitened image patches obtained from  the website, \href{http://www.rctn.org/bruno/sparsenet}{http://www.rctn.org/bruno/sparsenet}. We used the vectorized versions of these patches to train the sparse Det-Max WSM neural network in Figure \ref{fig:NNBone}. Figure \ref{fig:spdictionary} shows the receptive field images obtained from the columns of the inverse of the sparse-WSM separator, which correspond to localized Gabor-like edge features. This example confirms that the sparse WSM neural network with a local update rule successfully captures the behavior observed in primates' primary visual cortical neurons.

\begin{figure}[H]
\centering
\includegraphics[width = 0.6\textwidth]{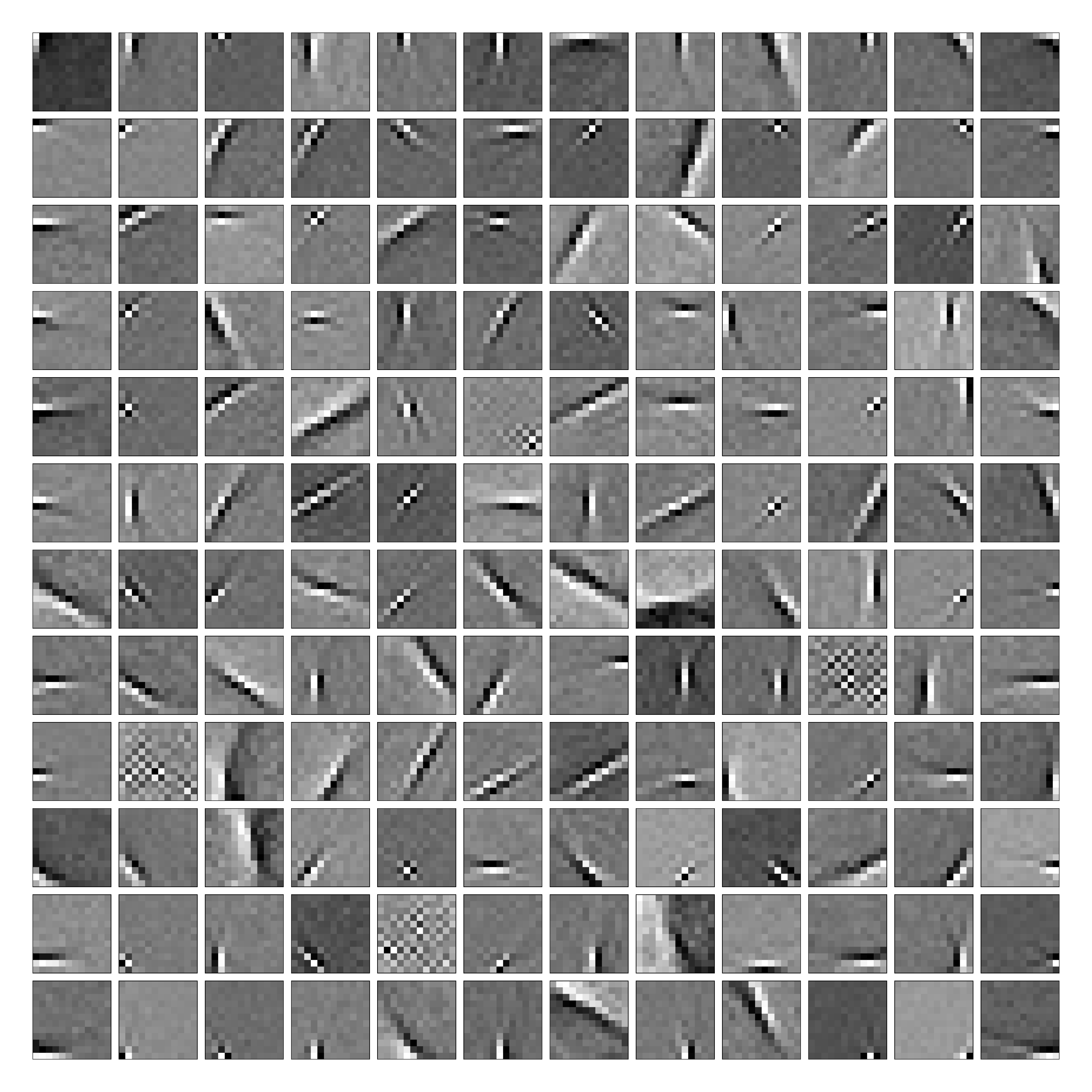}%
\caption{Dictionary obtained from the natural image patches by the sparse-WSM Network.}
\label{fig:spdictionary}
\end{figure}

\subsection{Source separation with mixed latent attributes}
In this section, we illustrate the source separation setting with different identifiable-enabling polytopes similar to the given example in \ref{appsec:mixedattributesdynamics}. These experiments demonstrate the capability of the proposed WSM Neural Network for general identifiable polytopes. The identifiability of the provided sets in this section are verified by the graph automorphism-based identifiability characterization algorithm presented in \cite{bozkurt:2022icassp}.

\subsubsection{Special polytope example in appendix \ref{appsec:mixedattributesdynamics}}
\label{sec:specialpolytopepmfexperimentappendix}
We provide numerical experiment results for the WSM Det-Max network in Figure \ref{fig:NNmixed} corresponding to the polytope in (\ref{eq:mixedpolytope}). To employ this WSM Det-Max Neural Network, we synthetically generated $n = 3$ dimensional uniform vectors in this polytope and mixed them by a random $6\times3$-matrix with i.i.d. standard normal entries. Also, the mixtures are corrupted by i.i.d. standard normal noise corresponding to $30$dB SNR level. Figure \ref{fig:spex3} illustrates the behavior of the overall SINR and individual source SNRs in addition to the behavior of diagonal weight matrices ($\vD_1$ and $\vD_2$) with respect to the number of update iterations for a single experiment. To measure the average behavior of this neural network, we run experiments for $100$ different source and mixing matrix generation, and Figure \ref{fig:spex4} illustrates the averaged SINR convergence behavior with the $25/75$-percentile envelope, as a function of update iterations. 

\begin{figure}[H]
\centering
\includegraphics[width=0.95\columnwidth]{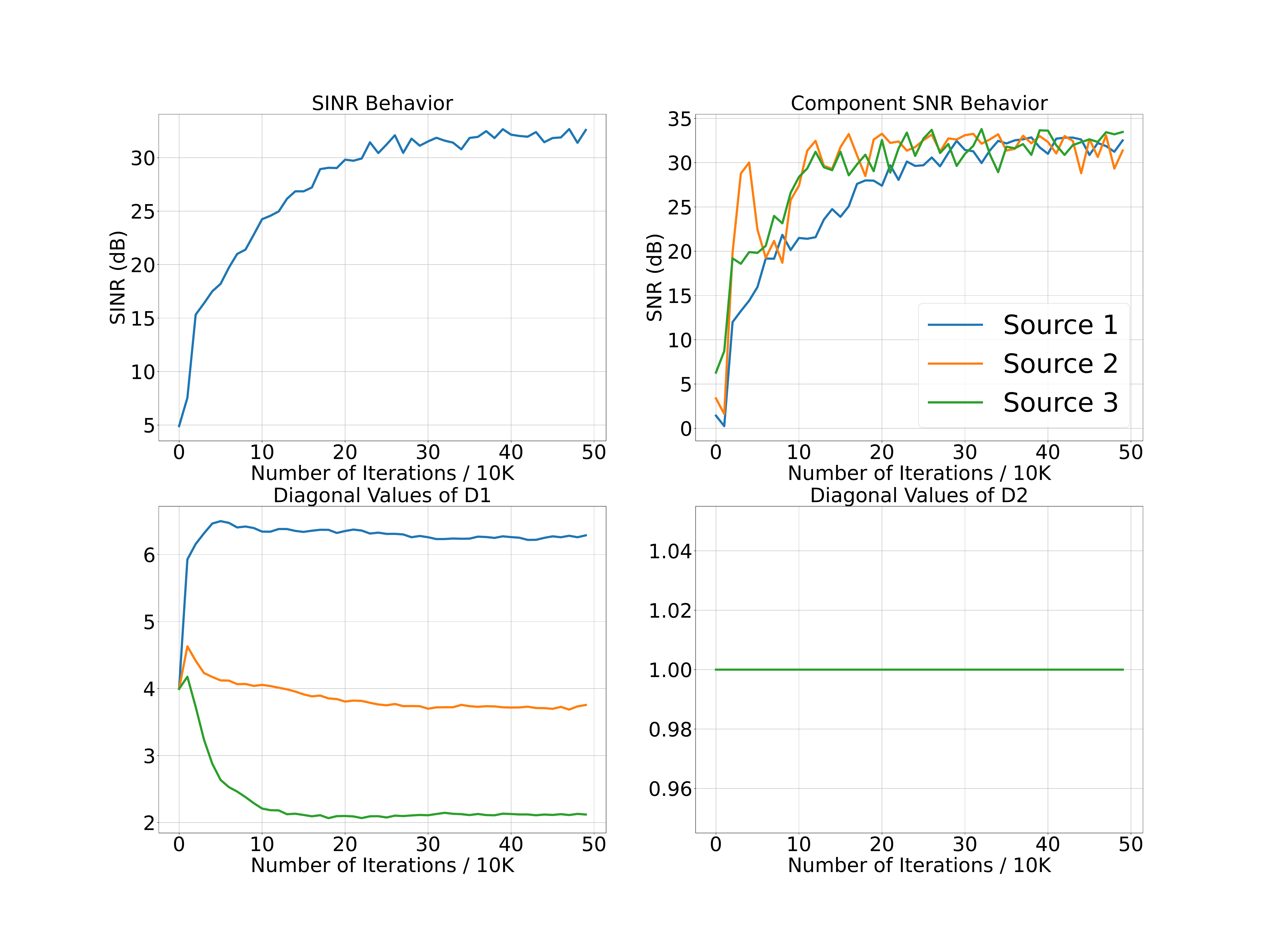}%
\caption{Example behaviors of SINR, component SNR values, and diagonal weights of $\vD_1$ and $\vD_2$ for a single experiment discussed in \ref{sec:specialpolytopepmfexperimentappendix} }
\label{fig:spex3}
\end{figure}

\begin{figure}[H]
\centering
\includegraphics[width=0.75\columnwidth]{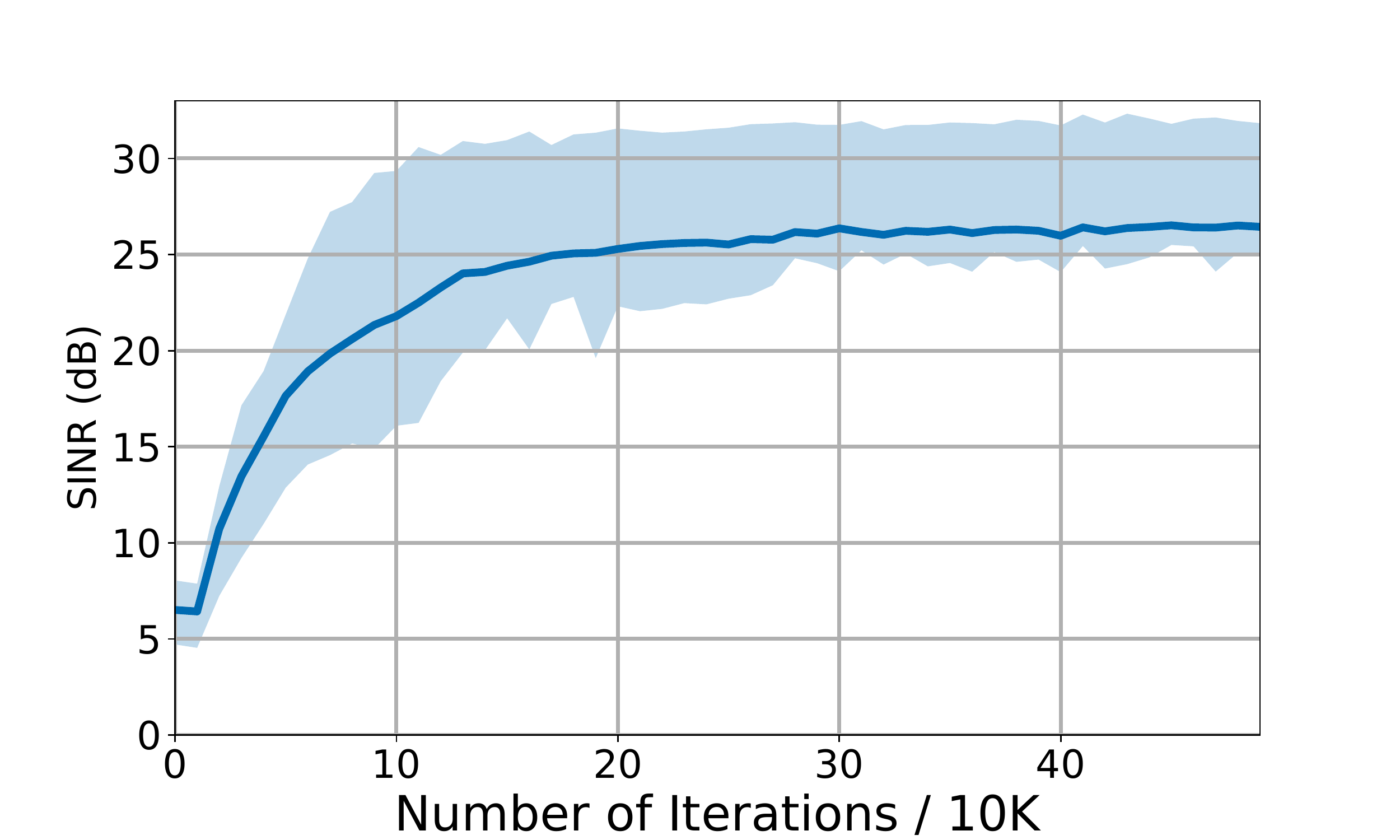}%
\caption{The SINR convergence curve for the experiments discussed in \ref{sec:specialpolytopepmfexperimentappendix}: mean-solid line with 25/75-percentile envelope.}
\label{fig:spex4}
\end{figure}

For this network, we used the following hyperparameter selections and variable initializations:
\begin{itemize}
    \item $\vD_1 = 4\vI$, and $\vD_2 = \vI$.
    \item $\mu_{\vD_1} = 5.725$, and $\mu_{\vD_2} = 10^{-2}$ ($\mu_{\vD_2} = 0$ for the experiment visualized in Figure \ref{fig:spex3}).
    \item $\beta = 0.5,\  \lambda_{SM} = 1 - 10^{-4}$.
    \item $1 - \gamma^2$ is dynamically adjusted using $1 - \gamma^2 = max\{0.25/(1 + \log(1 + t)),0.001\}$,  where $t$ is the data sample index.
    \item $\vM_H = 0.02\vI,  \vM_Y = 0.02\vI$.
    \item $\vW$ matrices are initialized first with i.i.d. standard normal random variables. Then, we normalized the Euclidean norm of all rows  to $0.0033$ by proper scaling.
    \item Learning rate for the neural dynamic iterations is determined to be $0.5$.
    \item Maximum number of neural dynamic iterations is restricted to be $\tau_{\text{max}} = 750$ if stopping condition is not satisfied.
    \item For the stability of the learning process, we keep the diagonal weights of $\vD_1$ and $\vD_2$ in a predetermined range, i.e., $10^{-6} \prec \mbox{diag}(\vD_1) \prec 10^6$ and $1 \prec \mbox{diag}(\vD_2) \prec 1.001$.
\end{itemize}

\subsubsection{Mixed anti-sparse and nonnegative anti-sparse sources}
\label{sec:mixedantinnantiappendix}
As another identifiable polytope example, we consider the following set which assigns mixed antisparse attributes to the source components: signed or nonnegative. For this experiment, we randomly selected two components to be nonnegative whereas the remaining three components are antisparse. The mixing matrix is a $10\times5-$matrix with i.i.d. standard normal entries. The mixtures are used to train the WSM Det-Max network similar to Figure \ref{fig:NNBinfty} where the clippings at the output layer corresponding to nonnegative sources are replaced with nonnegative clipping. 

\begin{eqnarray}
\Pcal_{}=\left\{\mathbf{s}\in \mathbb{R}^3\ \middle\vert \begin{array}{l}   s_{j_1},s_{j_2},s_{j_3}\in[-1,1],s_{j_4}, s_{j_5}\in[0,1], j_i \in \{1,2,3,4,5\}\\  \end{array}\right\}, \label{eq:mixedpolytopenumericexample2}
\end{eqnarray}
To train the WSM Det-Max network in this scenario, we used the following hyperparameter selections and variable initializations:
\begin{itemize}
    \item $\vD_1 = \vI$, and $\vD_2 = \vI$.
    \item $\mu_{\vD_1} = 1.125$, and $\mu_{\vD_2} = 0.1$.
    \item $\beta = 0.5,\  \lambda_{SM} = 1 - 5\times10^{-5}$.
    \item $1 - \gamma^2$ is dynamically adjusted using $1 - \gamma^2 = max\{0.4/(1 + \log(1 + t)),0.001\}$,  where $t$ is the data sample index.
    \item $\vM_H = 2\vI,  \vM_Y = \vI$.
    \item $\vW_{HX} = \vI, \vW_{YH}  = \vI$.
    \item Learning rate for the neural dynamic iterations is adjusted using $max\{0.75/(1 + \tau \times 0.005),0.05\}$, where $\tau$ is the neural dynamic iteration count.
    \item Maximum number of neural dynamic iterations is restricted to be $\tau_{\text{max}} = 750$ if stopping condition is not satisfied.
    \item For the stability of the learning process, we keep the diagonal weights of $\vD_1$ and $\vD_2$ in a predetermined range, i.e., $0.2 \prec \mbox{diag}(\vD_1) \prec 10^6$ and $0.5 \prec \mbox{diag}(\vD_2) \prec 5$.
\end{itemize}
\begin{figure}[H]
\centering
\includegraphics[width=0.75\columnwidth]{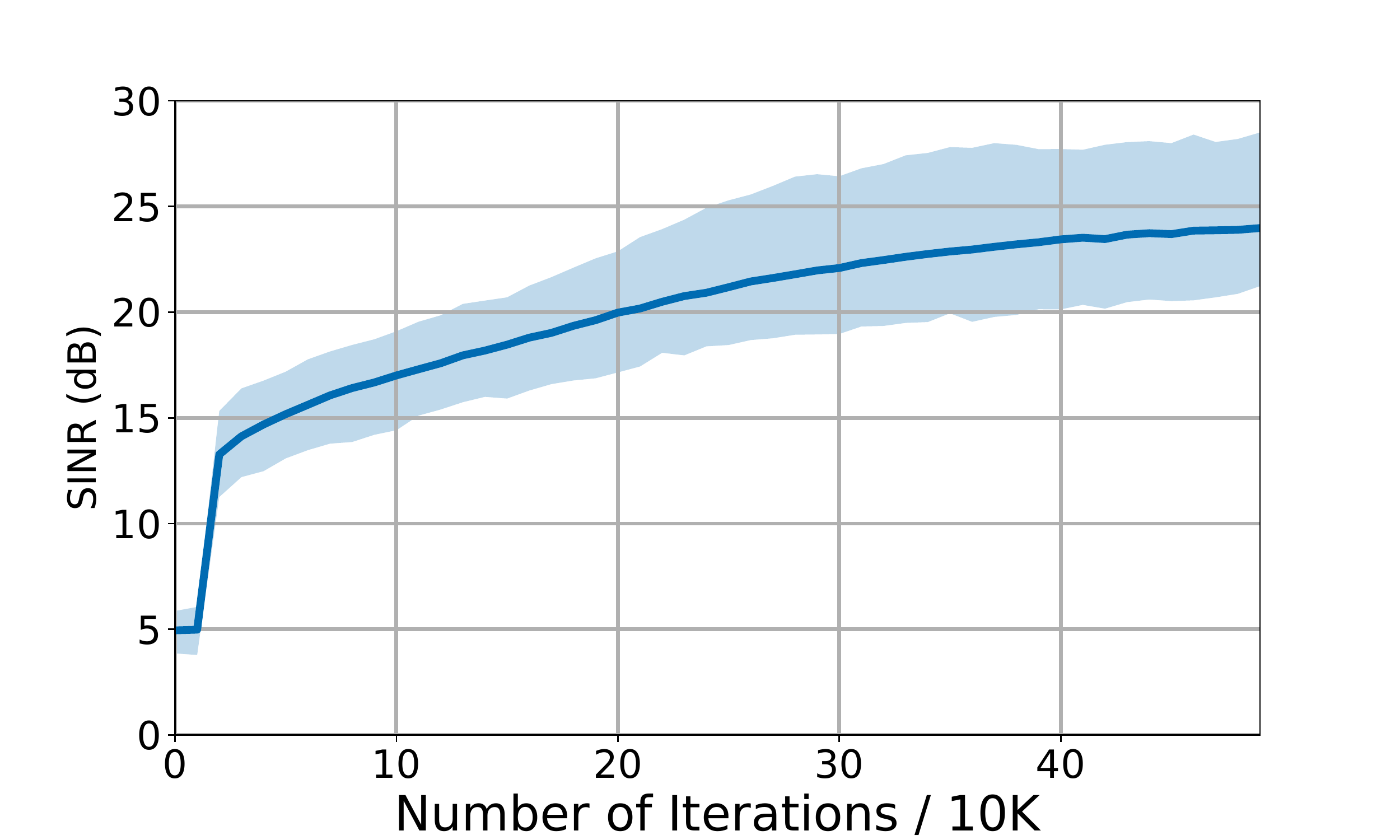}%
\caption{The SINR convergence curve for the experiments discussed in \ref{sec:mixedantinnantiappendix}: mean-solid line with 25/75-percentile envelope.}
\label{fig:spex5}
\end{figure}

\subsubsection{Mixed sparse and nonnegative anti-sparse sources}
\label{sec:mixedsparsennantiappendix}
As the last illustration of source separation on identifiable domains, we consider the following polytope,
\begin{eqnarray}
\Pcal_{}=\left\{\mathbf{s}\in \mathbb{R}^3\ \middle\vert \begin{array}{l} s_{j_1}\in[0,1],  \left\|\left[\begin{array}{c} s_{j_2} \\ s_{j_3} \\ s_{j_4} \\ s_{j_5} \end{array}\right]\right\|_1\le 1,  j_i \in \{1,2,3,4,5\}\end{array}\right\}, \label{eq:mixedpolytopenumericexample3}
\end{eqnarray}
where only one component is nonnegative and the subvector containing the remaining components is sparse. To demonstrate the source separation ability of WSM Det-Max Neural Network for this underlying domain, we generated $n = 5$ dimensional uniform vectors in this polytope. The sources are mixed with a $10\times5$ random matrix with standard normal entries. To train the WSM Det-Max network in this setting, we used the following hyperparameter selections and variable initializations:
\begin{itemize}
    \item $\vD_1 = 8\vI$, and $\vD_2 = \vI$.
    \item $\mu_{\vD_1} = 6$, and $\mu_{\vD_2} = 0.1$.
    \item $\beta = 0.5,\  \lambda_{SM} = 1 - 10^{-4}$.
    \item $1 - \gamma^2$ is dynamically adjusted using $1 - \gamma^2 = max\{0.25/(1 + \log(1 + t)),0.001\}$,  where $t$ is the data sample index.
    \item $\vM_H = 0.02\vI,  \vM_Y = 0.02\vI$.
    \item $\vW$ matrices are initialized first with i.i.d. standard normal random variables. Then, we normalized the Euclidean norm of all rows  to $0.0033$ by proper scaling.
    \item Learning rate for the neural dynamic iterations is adjusted using $max\{0.5/(1 + \tau \times 0.005),0.01\}$, where $\tau$ is the neural dynamic iteration count.
    \item Maximum number of neural dynamic iterations is restricted to be $\tau_{\text{max}} = 750$ if stopping condition is not satisfied.
    \item For the stability of the learning process, we keep the diagonal weights of $\vD_1$ and $\vD_2$ in a predetermined range, i.e., $10^{-6} \prec \mbox{diag}(\vD_1) \prec 10^6$ and $1 \prec \mbox{diag}(\vD_2) \prec 5$.
\end{itemize}
Figure \ref{fig:spex6} illustrates the SINR convergence behavior (averaged over 100 realizations) of the WSM Det-Max network for this scenario, as a function of update iterations.

\begin{figure}[H]
\centering
\includegraphics[width=0.75\columnwidth]{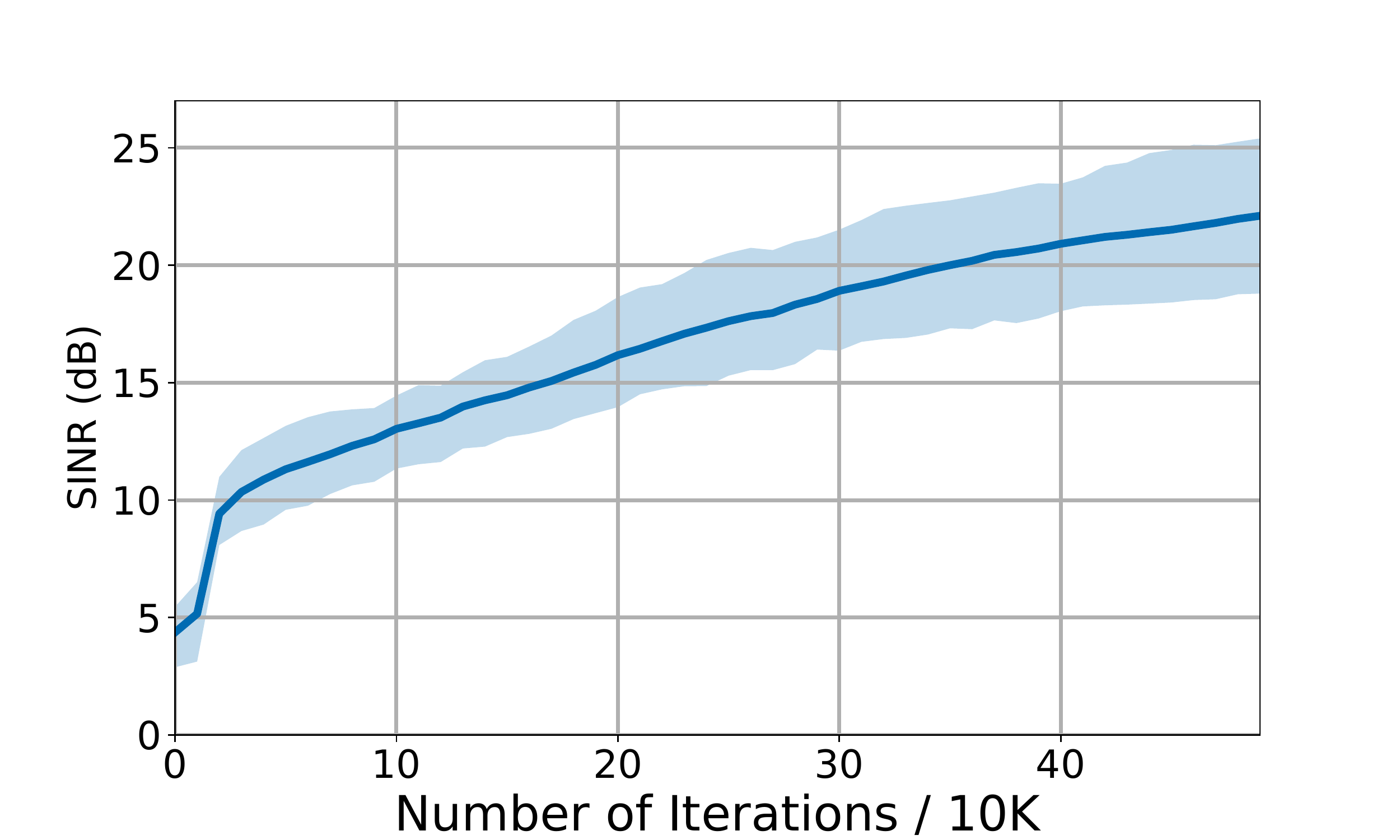}%
\caption{The SINR convergence curve for the experiments discussed in \ref{sec:mixedsparsennantiappendix}: mean-solid line with 25/75-percentile envelope.}
\label{fig:spex6}
\end{figure}

\subsection{Digital communication example: 4-PAM modulation scheme}
\label{sec:digitalcommunicationappendix}
We consider the 4 Pulse-amplitude modulation (4-PAM) scheme as a realistic application of blind separation of digital communication signals, with the symbols $\{\pm 3, \pm 1\}$. We consider a uniform symbol distribution, i.e., $P(s = i) = \frac{1}{4}\ \  \forall i = \pm3,\pm1$, where $s$ represents the transmitted symbol. We assume that $5$ sources are transmitted, with $400000$ samples each, and mixed through a $10\times5$ random matrix with standard normal entries. Without loss of generality, we make use of $\mathcal{B}_{\ell_\infty}$ polytope as the source domain assumption so that we feed the mixtures to the WSM Det-Max neural network for the antisparse sources. To train this network, we used the following hyperparameter selections and variable initializations:

\begin{itemize}
    \item $\vD_1 = 0.5\vI$, and $\vD_2 = 0.5\vI$.
    \item $\mu_{\vD_1} = 0.01$, and $\mu_{\vD_2} = 0.01$.
    \item $\beta = 0.5$, and  $\lambda_{SM} = 1 - 5\times10^{-3}$.
    \item $1 - \gamma^2$ is dynamically adjusted using $1 - \gamma^2 = max\{0.3/(1 + \log(1 + t)),0.05\}$,  where $t$ is the data sample index.
    \item $\vM_H = 2\vI,  \vM_Y = \vI$.
    \item $\vW$ matrices are initialized first with i.i.d. standard normal random variables. Then, we normalized the Euclidean norm of all rows  to $0.005$ by proper scaling.
    \item Learning rate for the neural dynamic iterations is adjusted using $max\{0.5/(1 + \tau \times 0.005),0.01\}$, where $\tau$ is the neural dynamic iteration count.
    \item Maximum number of neural dynamic iterations is restricted to be $\tau_{\text{max}} = 750$ if stopping condition is not satisfied.
    \item For the stability of the learning process, we keep the diagonal weights of $\vD_1$ and $\vD_2$ in a predetermined range, i.e., $0.2 \prec \mbox{diag}(\vD_1) \prec 10^6$ and $0.2 \prec \mbox{diag}(\vD_2) \prec 25$.
\end{itemize}

Figure \ref{fig:spex6} illustrates the SINR convergence behavior (averaged over 20 realizations) of the WSM Det-Max network for this scenario, as a function of update iterations.We conclude that our proposed approach is able to separate the source symbols from their mixtures.
\begin{figure}[H]
\centering
\includegraphics[width=0.75\columnwidth]{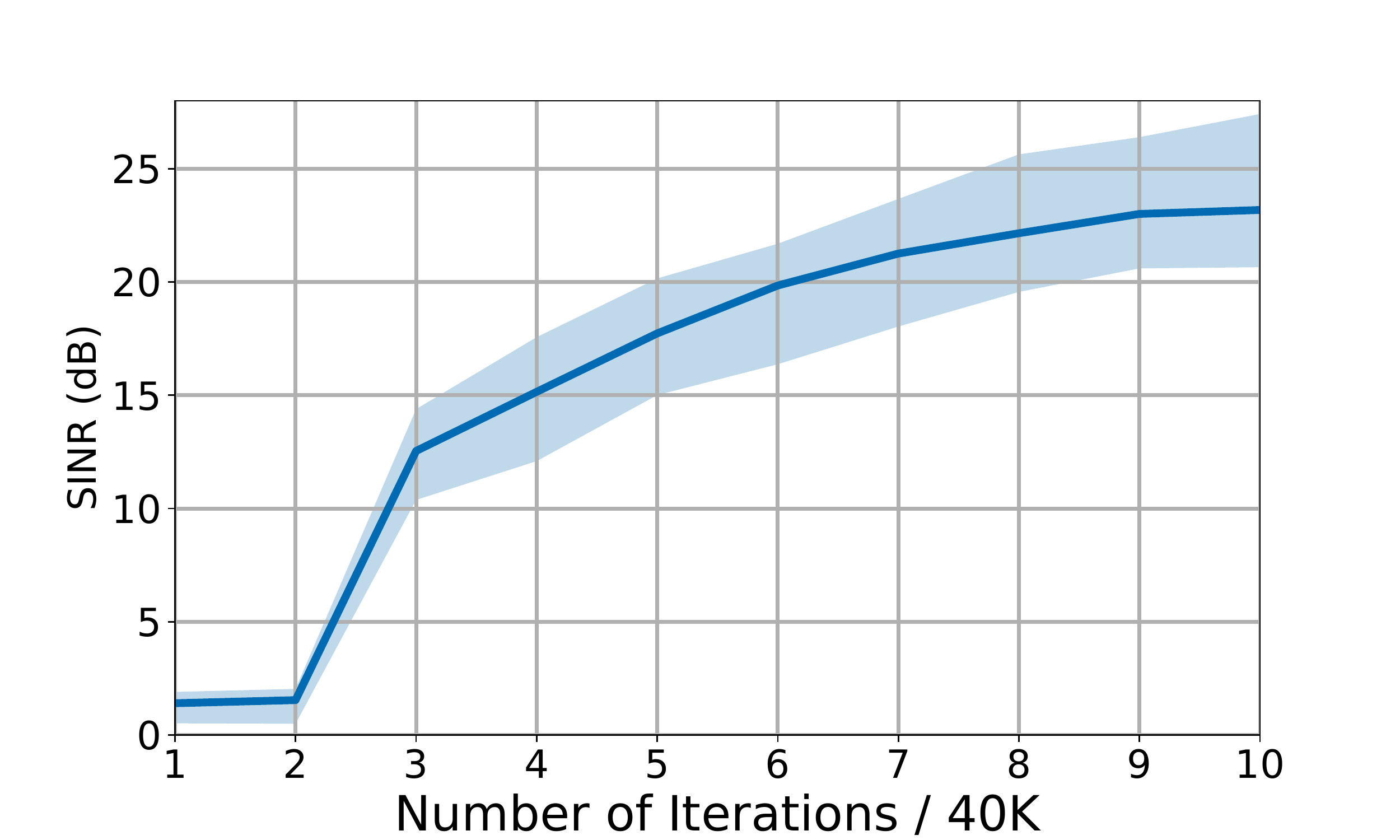}%
\caption{The SINR convergence curve for the 4-PAM digital communication signals: mean solid line with 25/75-percentile envelope.}
\label{fig:spex7}
\end{figure}

\subsection{Ablation study on hyperparameter selection for nonnegative sparse sources}
\label{sec:hyperparameterablationappendix}
The proposed  Det-Max WSM  framework requires many hyperparameter selections. In Section \ref{sec:numexpappendix}, we discuss the selection of these hyperparameters for different source domains. Most of the time, we find these hyperparameters by trial error and sensitivity analysis. Several ablation studies similar to grid search are useful to find the optimal values for the hyperparameters. In this section, we provide such ablation studies on effects of the selection of $\lambda_{\text{SM}}$, $\vD_1$, $\mu_{\vD_1}$, and $\gamma$. We chose to focus on $\lambda_{\text{SM}}$ here because we observed that it is one of the most sensitive parameters. Although the other parameters appear to have less of an effect on the final result than $\lambda_{\text{SM}}$, the cumulative impacts of the combined hyperparameter choices can substantially influence overall performance.

We consider nonnegative sparse source separation setup, i.e., $\Pcal=\mathcal{B}_{\infty,+}$. We generate $n=5$ dimensional source vectors uniformly in $\mathcal{B}_{\infty,+}$, and the mixing matrix is a $10\times 5$-matrix with i.i.d. standard normal entries. The mixtures train the nonnegative sparse-WSM Det-Max network illustrated in Figure \ref{fig:NNBoneplus}. In these ablation studies, we specifically consider the effect of hyperparameter selection for $1 - \lambda_{\text{SM}}$, initial $\vD_1$, $\mu_{\vD_1}$, and initial $1 - \gamma^2$. For each of the mentioned hyperparameters, we consider the following choices,
\begin{itemize}
    \item $1 - \lambda_{\text{SM}} \in \{10^{-3},10^{-4},10^{-5},10^{-6}\}$,
    \item $\vD_1 \in \{4\vI,8\vI,12\vI,16\vI\}$,
    \item $\mu_{\vD_1} \in \{5, 10 , 15, 20\}$,
    \item $\text{initial } 1 - \gamma^2 \in \{0.15, 0.20, 0.25, 0.30\}$
\end{itemize}
While experimenting with one hyperparameter, we fixed the rest of them as given in the following list,
\begin{itemize}
    \item $\vD_1 = 4\vI$, and $\vD_2 = \vI$.
    \item $\mu_{\vD_1} = 15$, and $\mu_{\vD_2} = 0.01$.
    \item $\beta = 0.5$, and  $\lambda_{SM} = 1 - 10^{-4}$.
    \item $1 - \gamma^2$ is dynamically adjusted using $1 - \gamma^2 = max\{0.25/(1 + \log(1 + t)),10^{-3}\}$,  where $t$ is the data sample index.
    \item $\vM_H = 0.02\vI,  \vM_Y = 0.02\vI$.
    \item $\vW$ matrices are first initialized with i.i.d. standard normal random variables. Then, we normalize the Euclidean norm of all rows to $0.0033$ by proper scaling.
    \item The learning rate for the neural dynamic iterations is adjusted using $max\{0.5/(1 + \tau \times 0.005),0.2\}$, where $\tau$ is the neural dynamic iteration count.
    \item Maximum number of neural dynamic iterations is restricted to be $\tau_{\text{max}} = 750$ if stopping condition is not satisfied.
    \item For the stability of the learning process, we keep the diagonal weights of $\vD_1$ and $\vD_2$ in a predetermined range, i.e., $10^{-6} \prec \mbox{diag}(\vD_1) \prec 10^6$ and $1 \prec \mbox{diag}(\vD_2) \prec 1.001$.
\end{itemize}


Figure \ref{fig:zetaablation} illustrates the SINR performance of the WSM Det-Max network concerning $1 - \lambda_{\text{SM}}$, and it demonstrates that it significantly affects the final SINR behavior of the proposed approach. We argue that the selection $\lambda_{\text{SM}} = 1 - 10^{-4}$ is a near-optimal for nonnegative sparse source separation with the WSM Det-Max network, whereas one can also implement a more detailed search based on possibly other hyperparameter dependencies. We also analyze the effect of initial $\vD_1$ on the final SINR, and Figure \ref{fig:D1startablation} demonstrates the performance change with $\vD_1$ gain initialization. We inspect that the WSM Det-Max network for nonnegative sparse sources relatively maintains its averaged performance against different initial gain parameters, whereas the selection of $\vD_1 = 4\vI$ leads to best performance with a significantly lower variance compared to other initialization choices. In Figure \ref{fig:muD1ablation}, we visualize the effect of learning rate choice for $\vD_1$. It is noticeable that $\mu_{\vD_1}$ is less effective in SINR performance compared to other considered hyperparameters, but $\mu_{\vD_1} = 15$ achieves the best average result with a lower variance. As the final ablation study on hyperparameter selection, we consider the initial value of $1 - \gamma^2$ which we dynamically adjust using  $max\{\nu /(1 + \log(1 + t)),10^{-3}\}$,  where $t$ is the data sample index and $\nu$ is the initial value. Figure \ref{fig:MUSablation} illustrates the effect for the initial value $\nu$, and it is remarkable that an improved result is attained for $\nu = 0.25$.

\begin{figure}[H]
\centering
\subfloat[a][]{
\includegraphics[trim = {0cm 0cm 0cm 0cm},clip,width=0.5\textwidth]{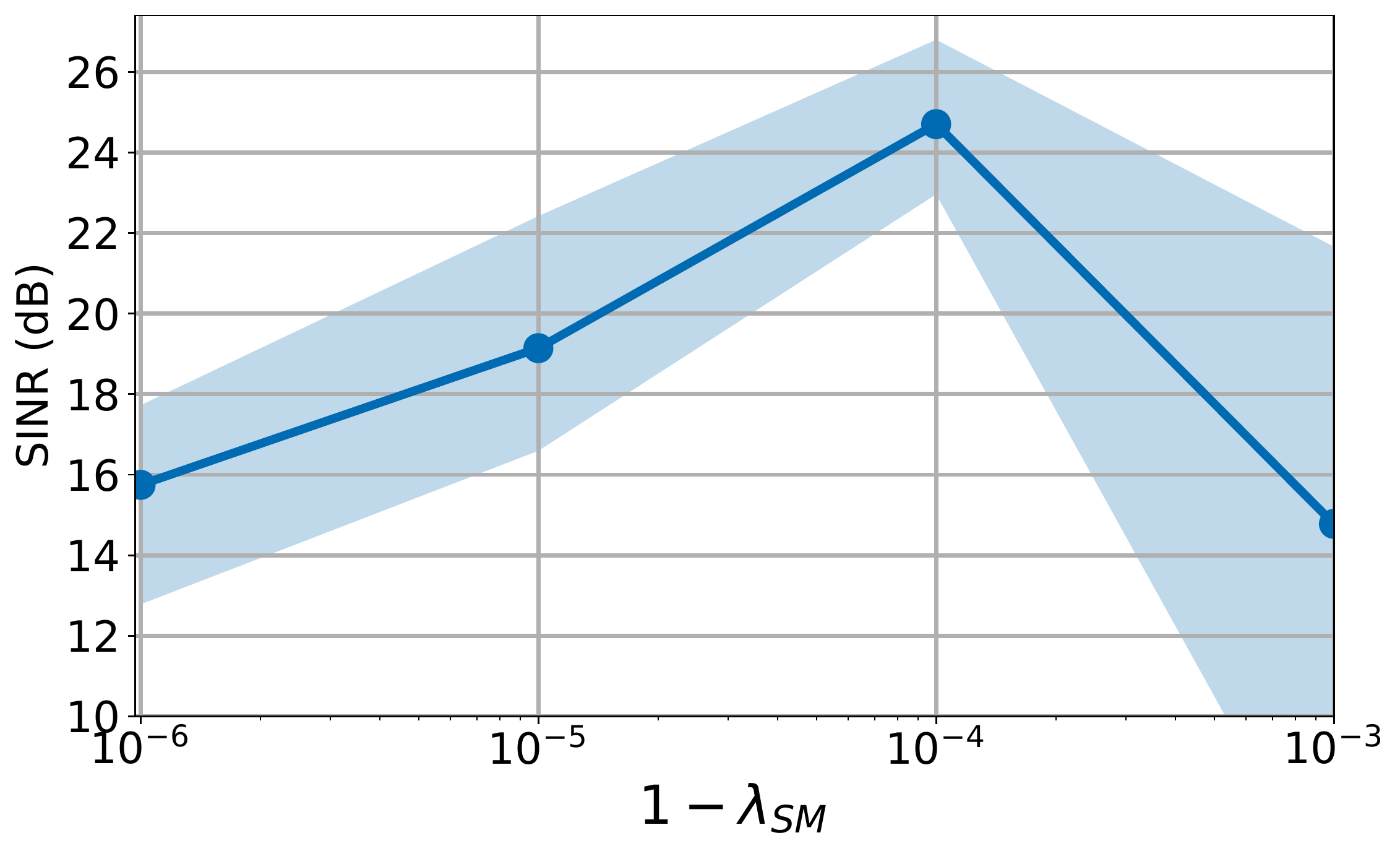}
\label{fig:zetaablation}}
\subfloat[b][]{
\includegraphics[trim = {0cm 0cm 0cm 0cm},clip,width=0.5\textwidth]{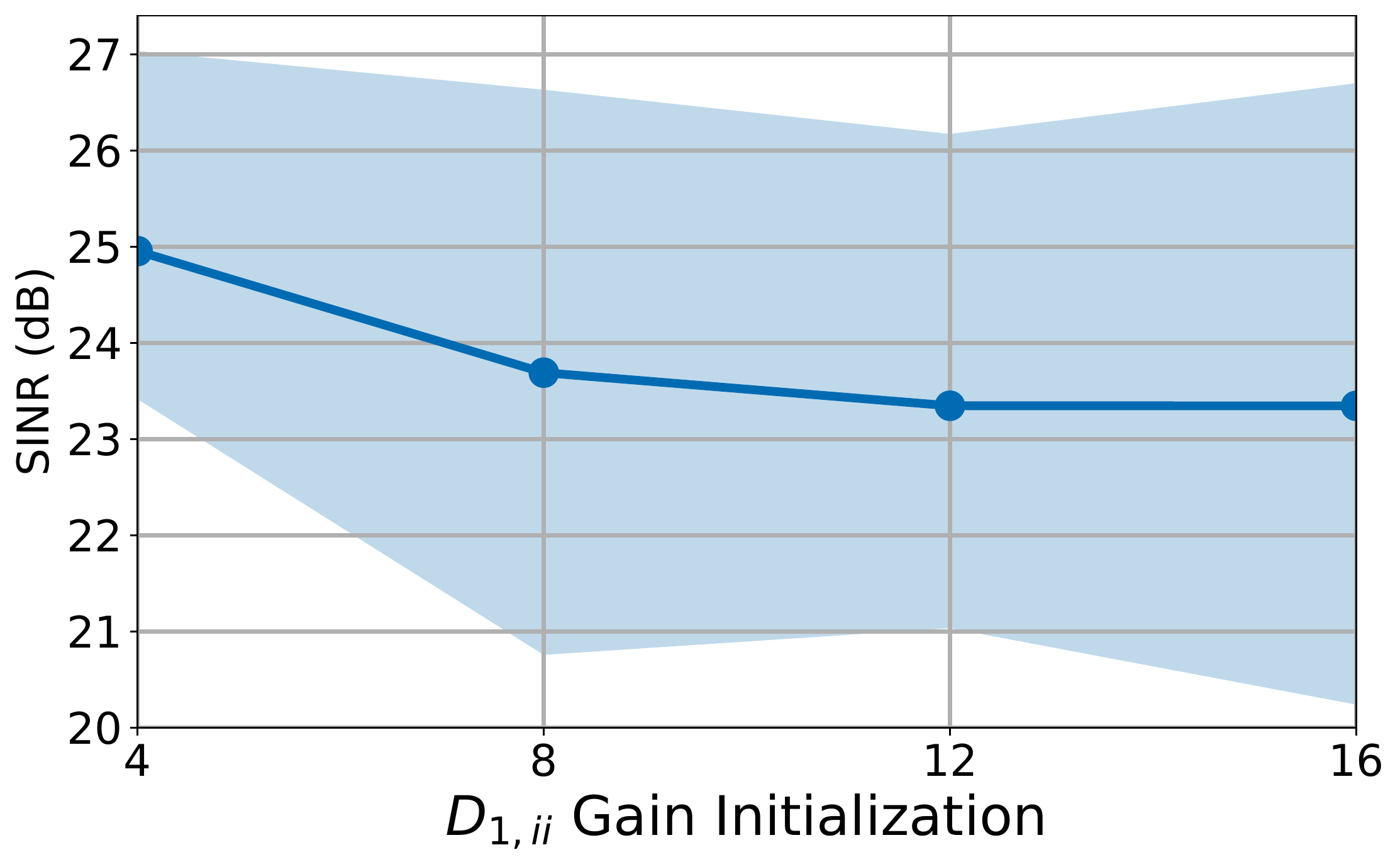}
\label{fig:D1startablation}
}
\end{figure}

\begin{figure}[H]\ContinuedFloat
\centering
\subfloat[c][]{
\includegraphics[trim = {0cm 0cm 0cm 0cm},clip,width=0.5\textwidth]{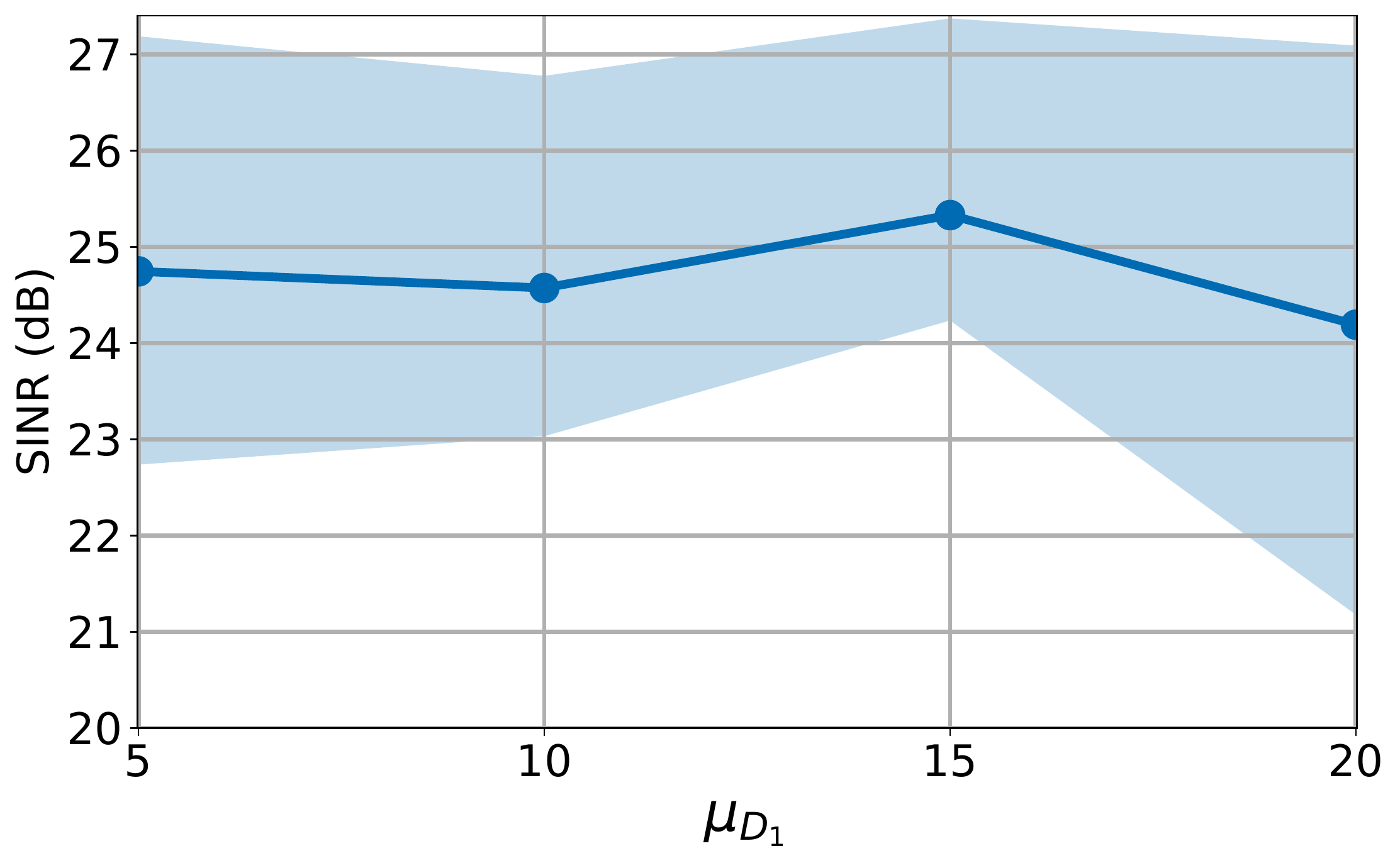}
\label{fig:muD1ablation}}
\subfloat[d][]{
\includegraphics[trim = {0cm 0cm 0cm 0cm},clip,width=0.5\textwidth]{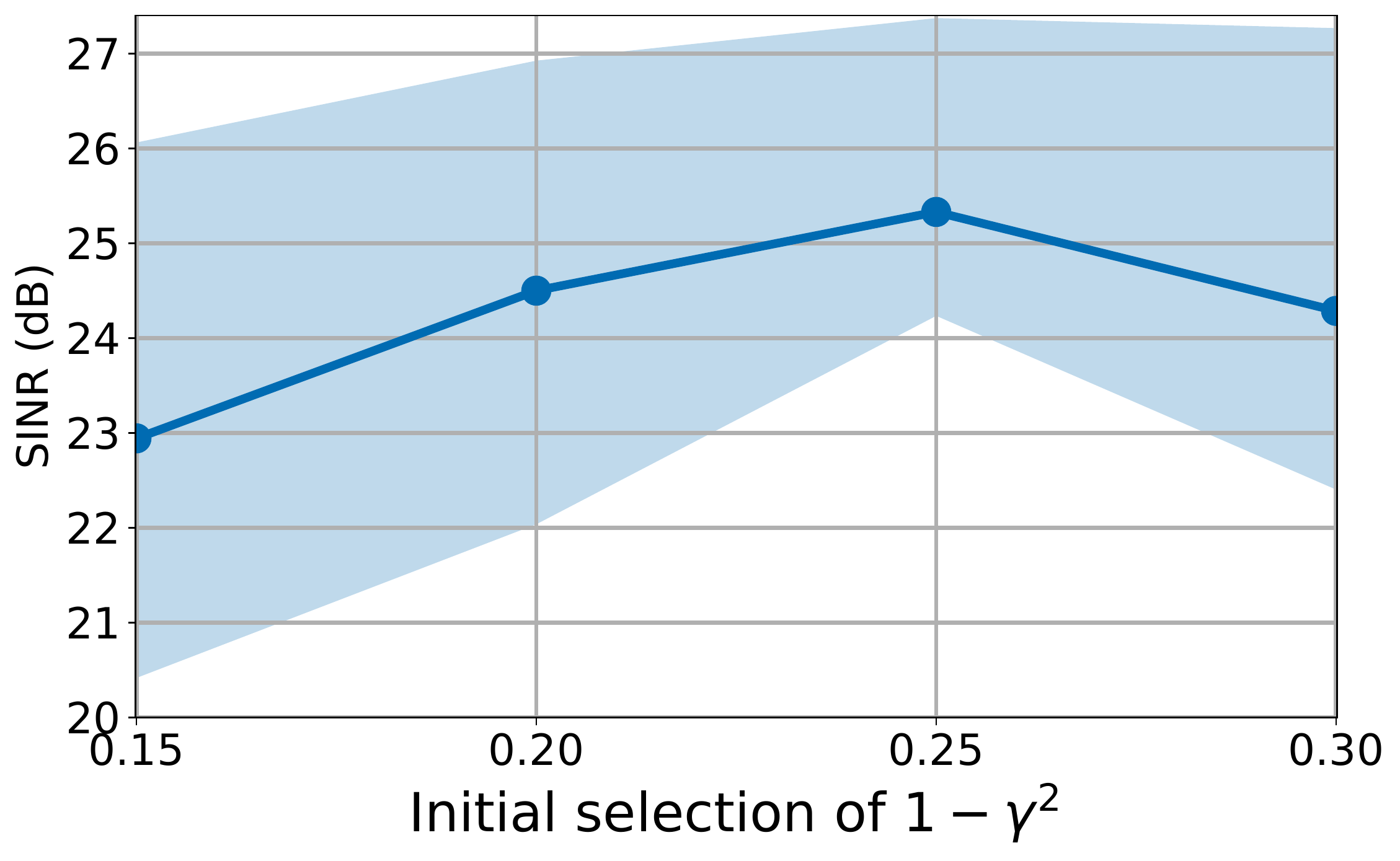}
\label{fig:MUSablation}}
\caption{SINR performances of WSM Det-Max networks for different hyperparameter selections (averaged over 50 realizations, mean solid lines with 25/75-percentile envelopes): (a) averaged SINR performance with respect to $1 - \lambda_{\text{SM}}$, (b) averaged SINR performance with respect to initial $\vD_1$, (c) averaged SINR performance with respect to $\mu_{\vD_1}$, (d) averaged SINR performance with respect to initial $1 - \gamma^2$.}
\qquad

\hfill
\label{fig:ablationfigure}
\end{figure}

\section{Discussion on the complexity of the proposed approach}
\label{sec:complexityappendix}

In this section, we discuss the computational complexity of the proposed WSM Det-Max neural network implementations. For simplicity, we consider the antisparse source separation cases discussed in Section \ref{linfdynamics}. Remarkably, the overall complexity is due to the output computation complexities which are determined by (\ref{eq:descv})-(\ref{eq:htnn}) and (\ref{eq:descu})-(\ref{eq:ytnn}). Note that these differential equations are naturally solved in neuromorphic implementations. However, in digital computer simulations, we need to implement loops to obtain their iterative solutions, as summarized in Algorithm \ref{alg:neuraldynamiciterationsanti}. As described in Section \ref{sec:mixingmodel}, assume that there are $n$ sources and $m$ mixtures, i.e., $\vx_t \in \mathbb{R}^m$, and $\vh_t, \vy_t \in \mathbb{R}^n$ for all $t$. Assuming that the factors $(1-\beta )\bar{\vM}_H(t)+\beta\vD_1(t)\bar{\vM}_H(t)\vD_1(t)$, $\beta \vD_1(t)\vW_{HX}(t)$, $\vW_{YH}(t)^T\vD_2(t)$, and $\bar{\vM}_Y(t)\vD_2(t)$ are computed outside the iterative loop of Algorithm \ref{alg:neuraldynamiciterationsanti}, the expressions in (\ref{eq:descv}) and (\ref{eq:descu}) require $2 n^2 + m n$ and $2 n^2$ multiplications, respectively. If we assume that the neural dynamic loop reaches to the pre-determined maximum number of iterations $\tau_\text{{max}}$, i.e., the numerical relative error check for the convergence is not satisfied, then the total number of multiplication is dominated by the factor $\tau_\text{{max}} (4 n^2 + m n)$. If we analyze the computational requirements of the factors $(1-\beta )\bar{\vM}_H(t)+\beta\vD_1(t)\bar{\vM}_H(t)\vD_1(t)$, $\beta \vD_1(t)\vW_{HX}(t)$, $\vW_{YH}(t)^T\vD_2(t)$, and $\bar{\vM}_Y(t)\vD_2(t)$, these calculations require multiplications of $(n^2-n)/2 + n^2+3n$, $m n + n$, $n^2$, and $n^2$, respectively, since $\vD_1(t)$ and $\vD_2(t)$ are diagonal matrices and $\bar{\vM}_H(t), \bar{\vM}_Y(t)$ are symmetric matrices. Therefore, the complexity of the neural dynamics of our proposed approach is dominated by the factor of $\tau_\text{{max}} (4 n^2 + m n)$. The complexity of the update rules of the gain variables expressed in equations (\ref{eq:D1update}) and (\ref{eq:D2update}) is dominated by the multiplication factor of $3n$ for all $2n$ variables, leading to the dominant multiplication factor of $6n^2$. Moreover, the update rules of the synaptic weight updates expressed in equation (\ref{eq:sydynWYH}) are dominated by the multiplication factor of $n^2$ or $mn$, leading to $4(n^2 + mn)$ number of multiplications. Therefore, the worst-case complexity of our proposed method per sample in terms of the big-O notation is 
$\mathcal{O}(\tau_{\text{max}}mn)$.

We now compare this with the complexity of the NSM and BSM algorithms. We first consider the prewhitening layer introduced in \citep{pehlevan2017blind}, as both algorithms require input to be prewhitened. Taking into account equations (28), (29), and (30) in \citep{pehlevan2017blind} for output computation and synaptic weight updates of the prewhitening layer, the complexity can be expressed in terms of big-O notation as $\mathcal{O}(\tau_{\text{max}}^\text{(NSM)}(m + k)n)$, where $k\geq n$ is an integer introduced as a result of the Lagrangian multiplier in equation (12) in the reference \citep{pehlevan2017blind}, and $\tau_{\text{max}}^\text{(NSM)}$ is the maximum predetermined number of iterations for the neural dynamic loop of NSM (see equation (28) and (33) in the reference). The output dynamics and the synaptic weight updates of the second layer of the online NSM network is described by the equations (33), (34), and (35) in \citep{pehlevan2017blind} which lead to the complexity in terms of big-O notation of $\mathcal{O}( n^2)$. As a result, the overall complexity of the NSM algorithm per sample can be stated as $\mathcal{O}(\tau_{\text{max}}^\text{NSM}(m + k)n)$. Similar to the WSM network, the neural dynamic loop of BSM has the complexity of $\mathcal{O}(\tau_{\text{max}}^\text{(BSM)}mn)$ as a result of recursion defined by Equation (17) in \citep{erdogan2020blind}, where $\tau_{\text{max}}^\text{(BSM)}$ is the predetermined maximum number of iterations for the neural dynamic loop of BSM. Furthermore, synaptic weight and gain updates introduce $\mathcal{O}(n^2)$ complexity similar to the WSM algorithm. Therefore, combining with prewhitening, the overall complexity of the BSM algorithm per sample becomes $\mathcal{O}( \tau_{\text{max}}^\text{(BSM)}(m + k)n)$

In conclusion, for the biologically plausible neural network solutions to the blind source separation problem, the overall complexity is determined by the recursive neural dynamic loops due to the implicit definition of the network output. Although this condition makes the implementation of such algorithms less feasible for digital hardware, they enable low-power implementations in future analog neuromorphic systems with local learning constraints.

\end{document}